\documentclass[12pt,a4paper]{article}

\usepackage{lmodern}
\usepackage{color}
\usepackage{amsmath,amsfonts,amsthm,eucal}
\usepackage{bm,bbm}
\usepackage{hyperref,url,cite}

\usepackage{tikz}
\usetikzlibrary{decorations.shapes}
\usetikzlibrary{decorations.markings}
\usetikzlibrary{hobby,decorations.pathreplacing,calc,calligraphy,backgrounds}
\usetikzlibrary{patterns}
\usepackage{anyfontsize}
\usepackage{adjustbox}

\allowdisplaybreaks[3]

\renewcommand{\Tilde}{\widetilde}

\newcommand{\N}{\mathbb{N}}     
\newcommand{\Z}{\mathbb{Z}}     
\newcommand{\R}{\mathbb{R}}     
\newcommand{\C}{\mathbb{C}}     
\newcommand{\T}{\mathbb{T}}     

\newcommand{\cO}{\mathcal{O}}   
\newcommand{\cZ}{\mathcal{Z}}
   
\newcommand{\cT}{\mathcal{T}}   

\newcommand{\eps}{\varepsilon}

\newcommand{\one}{\mathbbm{1}}

\renewcommand{\epsilon}{\varepsilon}

\newcommand{\dd}{\,\mathrm{d}}
\newcommand{\dvol}{\,\mathrm{dvol}}

\newcommand{\rad}{\mathrm{rad}}

\DeclareMathOperator{\ran}{ran}
\DeclareMathOperator{\diam}{diam}
\DeclareMathOperator{\dist}{dist}

\newcommand{\scalar}[1]{\ensuremath{\langle #1 \rangle}}
\numberwithin{equation}{section}
\numberwithin{figure}{section}

\newtheorem{theorem}{Theorem}[section]
\newtheorem{prop}[theorem]{Proposition}
\newtheorem{lemma}[theorem]{Lemma}
\newtheorem{corol}[theorem]{Corollary}

\theoremstyle{definition}

\newtheorem{defin}[theorem]{Definition}
\newtheorem{example}[theorem]{Example}
\newtheorem{remark}[theorem]{Remark}

\begin{document}

\title{\bf Embedded trace operator\\ for infinite metric trees}

\author{Valentina Franceschi\footnote{Universit\`a di Padova, Dipartimento di Matematica Tullio Levi-Civita, Via Trieste 63,  35121 Padova PD, Italy, E-mail: \url{valentina.franceschi@unipd.it}, Webpage: \url{https://sites.google.com/site/franceschivale/home}}
\and
Kiyan Naderi\footnote{Carl von Ossietzky Universit\"at Oldenburg, Institut f\"ur Mathematik, Carl-von-Ossietzky-Stra{\ss}e 9--11, 26129 Oldenburg, Germany, E-mail:
\url{kiyan.naderi@uol.de}}
\and
Konstantin Pankrashkin\footnote{Carl von Ossietzky Universit\"at Oldenburg, Institut f\"ur Mathematik, Carl-von-Ossietzky-Stra{\ss}e 9--11, 26129 Oldenburg, Germany,
E-mail: \url{konstantin.pankrashkin@uol.de},
Webpage: \url{https://uol.de/pankrashkin}}
}

\date{}

%
%
%

{

\small 	
\maketitle

}

\begin{abstract}
We consider a class of infinite weighted metric trees obtained as perturbations of self-similar regular trees. Possible definitions of the boundary traces of functions in the Sobolev space on such a structure are discussed by using identifications of the tree boundary with a surface. Our approach unifies some constructions proposed by Maury, Salort, Vannier (2009) for dyadic discrete weighted trees (expansion in orthogonal bases of harmonic functions on the graph and using Haar-type bases on the domain representing the boundary),
and by Nicaise, Semin (2018) and Joly, Kachanovska, Semin (2019) for fractal metric trees (approximation by finite sections and identification of the boundary with a interval): we show that both machineries give the same trace map, and for a range of parameters we establish the precise Sobolev regularity of the traces. In addition, we introduce new geometric ingredients by proposing an identification with arbitrary Riemannian manifolds. It is shown that any compact manifold admits a suitable multiscale decomposition and, therefore, can be identified with a metric tree boundary in the context of trace theorems.
\end{abstract}



\tableofcontents




\section{Introduction}

In the present paper we discuss some aspects of the trace theory for Sobolev spaces on infinite metric graphs. By a metric graph we mean a geometric configuration which arises if one replaces the edges of a discrete graph by intervals, and one introduces a differential operator on such a structure by defining a differential expression on each interval and by imposing a gluing condition at each node. A metric graph together with a differential operator
on it is often called a quantum graph. While
quantum graphs represent by now a well-established theory \cite{berk,bk, AGA}, the most attention was concentrated on the study of regular configurations with suitable lower bounds on the edge lengths and other parameters: in that case it is known that gluing conditions at the nodes are sufficient to define a self-adjoint operator or a non-self-adjoint one
with good properties \cite{HKS}. More recent papers \cite{rc1,rc2,dkmpt,ekmn,haes,kmn,lsv,SSVW} initiated the discussion
of the most general quantum graphs, which shows that in many cases additional ``boundary conditions at the external boundary'' must be imposed. It should be noted that the notion of boundary for general graphs is not obvious, which is a well-known issue for both metric and discrete infinite graphs \cite{rc2,kn,klr,kong,LSW,woess}; we recall that metric and discrete graphs show a number of common features \cite{kn,jvb,vB,Cat,kp1,LP}, and in case of equilateral metric graphs even a kind of unitary equivalence between respective Laplacians can be established \cite{kp2,kp3}. If the boundary is naturally defined (for example, for a tree, the set of infinite paths starting at a fixed vertex
can be naturally viewed as the boundary), one
arrives at the question of the description of possible boundary conditions,
which requires a construction of suitable function spaces at the boundary as well as a rigorous definition of boundary values for functions defined on the graph. For some classes of infinite trees, the abstract boundary can be endowed with a metric structure, which gives rise to Sobolev-type spaces and associated trace theorems: we refer to \cite{bms,rc2,bb,bb2,kosk1,kosk2} for related results.
On the other hand, if a tree models a structure embedded
into a space, then in many cases the boundary can be naturally
identified with a surface.
The aim of the present work is to define the boundary traces on metric trees using such an identification.

Our work is mainly motivated by the papers \cite{jks,msv,nisem,semin}
dealing with the analysis of Dirichlet-to-Neumann operators and wave equations
on trees viewed as a model of human lung. In particular, our main object (geometric tree) is directly borrowed from \cite{jks}. The paper \cite{msv}
deals with weigthed discrete Laplacians on an infinite dyadic tree, and it established a trace theorem for discrete Sobolev spaces by identifying the boundary with a Euclidean domain. The works \cite{jks,nisem,semin} proposed
a modified model with the help of the continuous weighted Laplacians, and an identification of the boundary trace
with an interval was addressed. Moreover in \cite{semin} the Sobolev regularity was partially studied. The notion of the boundary trace
was then used as a theoretical tool in \cite{nisem,semin} in order to establish the equivalence between various definitions of the Sobolev spaces on the fractal trees, which further surved for numerical approximation of infinite trees by their finite truncations when solving the wave equation \cite{jk1,jk2,jk3,js}. It should be said that the approaches of \cite{jks} and  \cite{jks,nisem,semin}
to the definition of the boundary trace were quite different: the paper \cite{msv}
uses an orthonormal basis of harmonic functions (so that the definition of the boundary trace of an arbitrary function is recovered from its expansion in this basis), while \cite{jks,nisem,semin} used more explicit approximations by finite truncations. As a by-result of our analysis one obtains that the both approaches are equivalent. In addition, we discuss for the first time an indentification
of the boundary with general Riemannian manifolds (in particular, Euclidean surfaces), which goes beyond the Euclidian domains considered in previous works.

We now describe our configuration and the main results in greater detail. Let $p\in\N$ with $p\ge 2$ be given and a root $o$ be given. We glue to $o$ an edge $e_{0,0}$ represented by an interval of length $\ell_{0,0}$, the second vertex of $e_{0,0}$ will be called $X_{0,0}$. If all $e_{n,k}$ and $X_{n,k}$ with $n\in\N_0$ and $k\in\{0,\dots,p^n-1\}$ are already constructed, then to each $X_{n,k}$ we attach
$p$ new edges $e_{n+1,pk+j}$, with $j\in\{0,\dots,p-1\}$, having lengths $\ell_{n+1,pk+j}$,
and the pendant vertices of $e_{n+1,pk+j}$, to be denoted by $X_{n+1,pk+j}$, will be viewed as children of $X_{n,k}$. This process continues infinitely, which creates a infinite rooted metric tree $\cT$. The subtree of $\cT$ starting at $X_{n,k}$, i.e. the subtree spanned by the offsping of $X_{n,k}$ (the children, the children of the children etc.), will be denoted by $\cT_{n,k}$. See Figure \ref{fig1}.

\begin{figure}[t]
\centering\fontsize{8pt}{1.1}

\begin{tabular}{ccc}
			\adjustbox{valign=m}{\begin{tikzpicture}[font=\scriptsize,edge from parent/.style={draw},scale=0.9]
	
	\tikzstyle{solid node}=[circle,draw,inner sep=1.2,fill=black];
	\tikzstyle{hollow node}=[circle,draw,inner sep=1.2];
	\tikzstyle{level 1}=[level distance=10mm,sibling distance=50mm]
	\tikzstyle{level 2}=[level distance=8mm,sibling distance=25mm]
	\tikzstyle{level 3}=[level distance=12mm,sibling distance=18mm]
	\tikzstyle{level 4}=[level distance=14mm,sibling distance=14mm]
	\node(0)[above]{$o$}
	child{node[above]{$X_{00}$}
		child{node[left]{$X_{1,0}$}
			child{node[below]{$X_{2,0}$}
				child{node[below]{} edge from parent[dotted]}
				child{node[below]{} edge from parent[dotted]}	
				edge from parent }
			child{node[below]{$X_{2,p-1}$}
				child{node[below]{} edge from parent[dotted]}
				child{node[below]{} edge from parent[dotted]}	
				edge from parent }
			edge from parent 
		}
		child{node[right]{$X_{1,p-1}$}
			child{node[below]{$X_{2,p(p-1)}$}
				child{node[below]{} edge from parent[dotted]}
				child{node[below]{} edge from parent[dotted]}	
				edge from parent }
			child{node[below]{$X_{2,p^2-1}$}
				child{node[below]{} edge from parent[dotted]}
				child{node[below]{} edge from parent[dotted]}	
				edge from parent }
			edge from parent 
		}
		edge from parent 
	};
	\node at (0,0) {};
	\node at (0,-1.4) {$\dots$};
	\node at (-1.7,-2.4) {$\dots$};
	\node at (1.9,-2.4) {$\dots$};
	\node at (2.75,-4) {$\dots$};
	\node at (1,-4) {$\dots$};
	\node at (-0.8,-4) {$\dots$};
	\node at (-2.6,-4) {$\dots$};
	
	%
	\node  at (-0.3, -0.15) {$e_{0,0}$};
	\node  at (-1, -0.8) {$e_{1,0}$};
	\node  at (1.1, -0.8) {$e_{1,p-1}$};
	\node  at (0.8, -2) {$e_{2,p(p-1)}$};
	\node  at (2.9, -2) {$e_{2,p^2-1}$};
\end{tikzpicture}}
&\quad&
	\adjustbox{valign=m}{\begin{tikzpicture}[font=\footnotesize,edge from parent/.style={draw,thick},scale=0.68]
	
	\tikzstyle{solid node}=[circle,draw,inner sep=1.2,fill=black];
	\tikzstyle{hollow node}=[circle,draw,inner sep=1.2];
	\tikzstyle{level 1}=[level distance=8mm,sibling distance=20mm]
	\tikzstyle{level 2}=[level distance=8mm,sibling distance=30mm]
	\tikzstyle{level 3}=[level distance=12mm,sibling distance=18mm]
	\node(0)[above]{}
	child{node[left]{$X_{n,k}$}
		child{node[below]{$X_{n+1,kp}$}
			child{node[below]{} edge from parent[dotted]}
			child{node[below]{} edge from parent[dotted]}	
			edge from parent[solid]}
		child{node[below]{$X_{n+1,kp+(p-1)}$}
			child{node[below]{} edge from parent[dotted]}
			child{node[below]{} edge from parent[dotted]}	
			edge from parent[solid] }
		edge from parent[dashed]
	}
	child{node[right]{} edge from parent[dashed]
	};
	\node at (0,0) {};
	\node at (-1.5,-1.55) {$\dots$};
	\node at (-1.5,-3.3) {$\dots$};
	\draw[decorate,decoration={brace,amplitude=6pt}] (-3.9,-3.0) -- (-3.9,-0.5) node[pos=0.5, left=0.2cm]{$\cT_{n,k}$};
\end{tikzpicture}}\\
(a) & &(b)
\end{tabular}

\caption{(a) The tree $\cT$. (b) A subtree $\cT_{n,k}$}\label{fig1}	

\end{figure}

For subsequent constructions it will be useful to introduce coordinates on $\cT$. Denote by $L_{n,k}$ the distance between the root $o$ and  $X_{n,k}$, i.e. the length of the unique path between $o$ and $X_{n,k}$ obtained by summing the lengths of all edges in the path. Then by $(n,k,t)$ with $t\in[L_{n,k}-\ell_{n,k},L_{n,k}]$ we denote the point of $e_{n,k}$ which is at the distance  $L_{n,k}-t$ from $X_{n,k}$. In this notation,
\[
X_{n,k}=(n,k,L_{n,k})=(n+1,pk+j,L_{n,k}) \text{ for any } j\in\{0,\dots,p-1\}.
\]
Let $w:\cT\to(0,\infty)$ be a locally bounded measurable function, which will be used
as an integration weight: for $f:\cT\to\C$ one defines
\[
\int_\cT f\dd\mu:=\sum_{n=0}^\infty\sum_{k=0}^{p^n-1} \int_{L_{n,k}-\ell_{n,k}}^{L_{n,k}}f(n,k,t)\,w(n,k,t)\dd t,
\]
then
\[
L^2(\cT):=\big\{f:\cT\to\C:\ \|f\|^2_{L^2(\cT)}:=\int_\cT |f|^2\dd\mu<\infty\big\}.
\quad
\]
Due to the above definition the set of vertices has zero measure. Therefore, each measurable function $f:\cT\to\C$
can be identified with a family of functions $(f_{n,k})$,
\begin{gather*}
f_{n,k}:=f(n,k,\cdot):\ (L_{n,k}-\ell_{n,k},L_{n,k})\to \C,\quad
n\in\N_0,
\quad
k\in\{0,\dots,p^n-1\}.
\end{gather*}
Then $f=(f_{n,k})$ belongs to $L^2(\cT)$ if and only if
\[
\|f\|^2_{L^2(\cT)}:=\sum_{n=0}^\infty\sum_{k=0}^{p^n-1} \int_{L_{n,k}-\ell_{n,k}}^{L_{n,k}}\big|f_{n,k}(t)\big|^2w_{n,k}(t)\dd t<\infty.
\]
If $f=(f_{n,k})$ such that all $f_{n,k}$ are weakly differentiable, we denote $f':=(f'_{n,k})$.
The first Sobolev space $H^1(\cT)$ is then introduced as
\begin{align*}
	H^1(\cT)&:=\{f\in L^2(\cT): f \text{ is continuous with } f'\in  L^2(\cT)\},\\
	\|f\|^2_{H^1(\cT)}&:=\|f\|^2_{L^2(\cT)}+\|f'\|^2_{L^2(\cT)}.
\end{align*}
Moreover, we denote
\begin{align*}
	H^1_c(\cT)&:=\{f\in H^1(\cT): \text{ there exists $N\in\N$ such that}\\
	&\qquad f_{n,k}\equiv 0 \text{ for all $(n,k)$ with $n>N$}\},\\
	H^1_0(\cT)&:=\text{the closure of $H^1_c(\cT)$ in $H^1(\cT)$.}
\end{align*}

\begin{figure}[t]
	\centering
	\begin{tikzpicture}[font=\footnotesize,edge from parent/.style={draw,thick}]
		\draw[fill=gray!50, draw=none, shift={(3.5, -3.2)},scale=0.7, rotate=180]
		(0, 0) to[out=20, in=140] (3, -0.4) to [out=60, in=160]
		(10, 0.5) to[out=130, in=60]
		cycle;
		
		\tikzstyle{solid node}=[circle,draw,inner sep=1.2,fill=black];
		\tikzstyle{hollow node}=[circle,draw,inner sep=1.2];
		\tikzstyle{level 1}=[level distance=10mm,sibling distance=50mm]
		\tikzstyle{level 2}=[level distance=8mm,sibling distance=22mm]
		\tikzstyle{level 3}=[level distance=8mm,sibling distance=15mm]
		\tikzstyle{level 4}=[level distance=12mm,sibling distance=10mm]
		\node(0)[above]{$o$}
		child{node[above]{$X_{00}$}
			child{node[left]{$X_{1,0}$}
				child{node[below]{$X_{2,0}$}
					child{node[below]{} edge from parent[dotted]}
					child{node[below]{} edge from parent[dotted]}	
					edge from parent }
				child{node[below]{$X_{2,p-1}$}
					child{node[below]{} edge from parent[dotted]}
					child{node[below]{} edge from parent[dotted]}	
					edge from parent }
				edge from parent 
			}
			child{node[right]{$X_{1,p-1}$}
				child{node[below]{$X_{2,p(p-1)}$}
					child{node[below]{} edge from parent[dotted]}
					child{node[below]{} edge from parent[dotted]}	
					edge from parent }
				child{node[below]{$X_{2,p^2-1}$}
					child{node[below]{} edge from parent[dotted]}
					child{node[below]{} edge from parent[dotted]}	
					edge from parent }
				edge from parent 
			}
			edge from parent 
		};
		\node at (0,0) {};
		\node at (1,0) {{\large $\cT$}};
		\node at (0,-1.4) {$\dots$};
		\node at (-1.6,-2.2) {$\dots$};
		\node at (1.7,-2.2) {$\dots$};
		\node at (2.47,-3.5) {$\dots$};
		\node at (0.95,-3.5) {$\dots$};
		\node at (-0.8,-3.5) {$\dots$};
		\node at (-2.3,-3.5) {$\dots$};
		

		\shade[thin, left color=gray!10, right color=gray!50, draw=none,
		shift={(3.5, -3.2)},scale=0.7, rotate=180]
		(0, 0) to[out=10, in=140] (6.6, -0.2) to [out=60, in=120] (10, 0.5)
		to[out=130, in=60] cycle;
		\node at (0,-4.5) {{\large $\Omega$}};
	\end{tikzpicture}
	\caption{The boundary of $\cT$ viewed as a surface $\Omega$}\label{fig2b}
\end{figure}

One arrives at the following quite natural questions:
\begin{itemize}
	\item[(a)] Do we have $H^1_0(\cT)=H^1(\cT)$? 
	\item[(b)] If not, can we characterize the functions in $H^1_0(\cT)$ by their ``behavior of infinity'', i.e. by the behavior of $f=(f_{n,k})\in H^1(\cT)$ for $n\to \infty$?
	\item[(c)] Can this ``behavior at infinity'' be characterized by a function defined on some set $\Omega$ viewed as the ``boundary'' of $\cT$?
\end{itemize}
Remark that the $H^1$-norm on $\cT$ represents the sesquilinear form
of the Neumann Laplace operator $\Delta$, which is important for the study of various diffusion
processes on $\cT$. One can also consider first a thickened version $\cT_\eps$ of $\cT$ (i.e. one embeds $\cT$ in $\R^n$ and takes the $\eps$-neighborhood) and consider the associated Neumann Laplacian $\Delta_\eps$, then one has a suitably defined convergence of $\Delta_\eps$ to $\Delta$ as $\eps\to 0$, see \cite{pwz,post}. The above problem (a) 
is related to the question whether the boundary of $\cT$ is penetrable, i.e. whether one can impose alternative boundary conditions at the tree boundary, and the problems (b) and (c)
are closely related to a concrete representation of such conditions and
to the existence and uniqueness of solutions of the associated boundary value problems.

We provide answers to the above questions by identifying the abstract boundary $\partial\cT$
with a geometric object, more precisely, an open set $\Omega$ with compact closure in a $d$-dimensional Riemannian manifold (in particular, $\Omega$ is allowed to be an arbitrary compact Riemannian manifold), see Figure \ref{fig2b}. The main assumption on
$\Omega$ is that it admits a special decomposition: there exists $\Omega_{n,k}\subset \Omega$, $n\in\N_0$, $k\in\{0,\dots,p^n-1\}$, constructed as follows. One sets $\Omega_{0,0}:=\Omega$. If some $\Omega_{n,k}$ is constructed, one chooses $p$ non-empty disjoint subsets $\Omega_{n+1,pk+j}\subset \Omega_{n,k}$, $j\in\{0,\dots,p-1\}$, such that
\[
\Big| \Omega_{n,k}\setminus \bigcup_{j=0}^{p-1}\Omega_{n+1,pk+j}\Big|=0,
\]
and this process continues infinitely (Figure \ref{fig2}). In addition, one needs to impose some geometric conditions on $\Omega_{n,k}$ for large $n$: informally, all $\Omega_{n,k}$ must have approximately the same volume, and their shape is not allowed to become ``too complicated''. A decomposition satisfying all necessary assumptions will be called \emph{a regular strongly balanced $p$-multiscale decomposition} of $\Omega$ (we refer to Subsection \ref{ssmult} for rigorous definitions concerning Euclidean open sets and to Subsection \ref{sec44} for an extension to the case of manifolds). 

\begin{figure}[t]
	\centering
\begin{minipage}{0.25 \textwidth}
	\begin{tikzpicture}[scale=0.75]
		\draw (0,0) -- (4,0) -- (4,4) -- (0,4) -- (0,0);
		\node at (2,2) {$\Omega_{0,0}$};
	\end{tikzpicture}
\end{minipage}
$\to$
\begin{minipage}{0.25 \textwidth}
	\begin{tikzpicture}[scale=0.75]
		\draw (0,0) -- (4,0) -- (4,4) -- (0,4) -- (0,0);
		\draw (2,0) -- (2,4);
		\node at (1,2) {$\Omega_{1,0}$};
		\node at (3,2) {$\Omega_{1,1}$};
	\end{tikzpicture}
\end{minipage}
$\to$
\begin{minipage}{0.25 \textwidth}
	\begin{tikzpicture}[scale=0.75]
		\draw (0,0) -- (4,0) -- (4,4) -- (0,4) -- (0,0);
		\draw (2,0) -- (2,4);
		\draw (2,0) -- (2,4);
		\draw (0,2) -- (4,2);
		\node at (1,1) {$\Omega_{2,0}$};
		\node at (1,3) {$\Omega_{2,1}$};
		\node at (3,1) {$\Omega_{2,2}$};
		\node at (3,3) {$\Omega_{2,3}$};
	\end{tikzpicture}
\end{minipage}
$\to\dots$
	\caption{An example of a multiscale decomposition (for $p=2$)}\label{fig2}
\end{figure}

Remark that the combinatorial structure of the family $(\Omega_{n,k})$ repeats the combinatorial structure of the family of substrees $(\cT_{n,k})$: for arbitrary $(n,k)$ and $(n',k')$ one has
\begin{itemize}
	\item $\Omega_{n,k}\subset \Omega_{n',k'}$ if and only if $\cT_{n,k}\subset \cT_{n',k'}$,
	\item $\Omega_{n,k}\cap \Omega_{n',k'}\ne\emptyset$ if and only if $\cT_{n,k}\cap \cT_{n',k'}\ne\emptyset$,
\end{itemize}
and this observation is used to create a link between the functions defined on $\cT$
and those defined on $\Omega$. More precisely, one imagines that the boundary of $\cT$
is glued to $\Omega$ in such a way that the boundary of each $\cT_{n,k}$ is glued to $\Omega_{n,k}$. In this case, if a function $f$ on $\cT$ has a constant value $\alpha_{n,k}$ along some $\cT_{n,k}$ and is zero on all other subtrees $\cT_{n,j}$ with $j\ne k$, then
it natural to identify the boundary trace of $f$ with the function $\alpha_{n,k}\one_{\Omega_{n,k}}$ (Figure \ref{fig4}).
It appears that this somewhat naive definition can be given a rigorous form, and
a part of our main results can be summarized as follows:

\begin{figure}
	\centering
		\begin{tikzpicture}[font=\footnotesize,edge from parent/.style={draw,thick},scale=0.8]
		
		\tikzstyle{solid node}=[circle,draw,inner sep=1.2,fill=black];
		\tikzstyle{hollow node}=[circle,draw,inner sep=1.2];
		\tikzstyle{level 1}=[level distance=10mm,sibling distance=50mm]
		\tikzstyle{level 2}=[level distance=8mm,sibling distance=30mm]
		\tikzstyle{level 3}=[level distance=8mm,sibling distance=20mm]
		\tikzstyle{level 4}=[level distance=12mm,sibling distance=10mm]
		\node(0)[above]{$o$}
		child{node[above]{$X_{0,0}$}
			child{node[left]{$X_{1,0}$}
				child{node[below]{$X_{2,0}$}	
					edge from parent[solid]}
				child{node[below]{}
					child{node[below]{} edge from parent[dotted]}
					child{node[below]{} edge from parent[dotted]}	
					edge from parent[black, dotted] }
				edge from parent[]
			}
			child{node[right]{}
				child{node[below]{}
					edge from parent}
				child{node[below]{}	
					edge from parent}
				edge from parent[dotted, black]
			}
			edge from parent[black]
		};
		\node at (0,0) {};
		\node at (1,0.3) {{\large $\cT$}};
		\node at (4.2,-5.5) {{\large $\Omega$}};
		\draw (0.5,-3.9) to [closed, curve through = {(3.5,-4.3) (-0.2,-5.8)}]  (-4,-4.4);
		\draw (0.5,-3.9) .. controls (0,-5) .. (-0.2,-5.8)
		(-4,-4.4) .. controls (0,-5.3) .. (3.5,-4.3);
		\draw[dotted, thick] (-2.9,-2.67) -- (-2.8,-4.1)
		(-2.65,-2.67) -- (-1,-4.3);
		\node[black] at (-4.4,-2.5) {{\normalsize $\cT_{2,0}$}};
		\draw [thick, <->] (-4.4,-2.8) -- (-4.4,-3.5);
		\node[black] at (-4.4,-3.8) {{\normalsize $\Omega_{2,0}$}};
		\filldraw[pattern=north west lines, pattern color=black, opacity=0.2](0.5,-3.9)  -- (0,-5.05) {[rounded corners] 
			--(-1.1,-5)} --(-4,-4.4) {[rounded corners]  --(-3.79,-3.8) --(-3.45,-3.5) --(-3.1,-3.42) --(-2.5,-3.43) --(0,-3.91)} -- cycle; 
	\end{tikzpicture}
\caption{Identifying $\cT_{n,k}$ with $\Omega_{n,k}$}\label{fig4}
\end{figure}

\begin{theorem}\label{thm11}
Assume that there exist constants $\alpha>0$, $0<\ell<1$ and $c\ge 1$ such that for any $n\in\N_0$ and $k\in\{0,\dots,p^n-1\}$ there holds
\begin{equation}
	 \label{cc}
c^{-1}\ell^n\le \ell_{n,k}\le c\ell^n,\quad 
c^{-1}\alpha^n\le w_{n,k}\le c\alpha^n.
\end{equation}
Then $H^1_0(\cT)\ne H^1(\cT)$ if and only if
\begin{equation}
	\label{lap0}
\ell<\alpha p<\dfrac{1}{\ell}.
\end{equation}

Assume that \eqref{lap0} is satisfied and let $\Omega$ be a non-empty open set with compact closure in a $d$-dimensional Riemannian manifold of bounded geometry. We denote by $H^s(\Omega)$ the associated fractional Sobolev spaces of order $s\ge 0$ and require that $\Omega$ admits a regular strongly balanced $p$-multiscale decomposition $(\Omega_{n,k})$ as defined in Subsection~\ref{sec44}.
Denote
\begin{equation*}
	\sigma:=\frac{1}{2}\Big ( 1-\frac{\log\ell-\log\alpha}{\log p}\Big)>0
\end{equation*}
and let $0\le s<\frac{1}{2}$ such that $s\le \sigma d$, then for any $f\in H^1(\cT)$
there exists the limit
\[
\gamma^\cT_\Omega f:=\lim_{N\to\infty}\sum_{K=0}^{p^N-1} f(X_{N,K})\one_{\Omega_{N,K}}\in H^s(\Omega).
\]
The embedded trace operator defined by $\gamma^\cT_\Omega:H^1(\cT)\to H^s(\Omega)$
is a bounded linear operator with $\ker \gamma^\cT_\Omega=H^1_0(\Omega)$, and
\[
\gamma^\cT_\Omega\big(H^1(\cT)\big)=H^{\sigma d}(\Omega)
\text{ if } \sigma d<\frac{1}{2}.
\]
\end{theorem}

\begin{remark}
If $\Omega$ is a Euclidean open set, then one can show that the linear map $\gamma^\cT_\Omega$ given by the same expression is bounded and surjective
as a map from $H^1(\cT)$ to $A^{\sigma d}(\Omega)$ for \emph{any} value of $\sigma$,
where $A^{\sigma d}(\Omega)$ is a so-called approximation space (which happens to coincide with $H^{\sigma d}$ if $\sigma d<\frac{1}{2}$): we refer to Subsection \ref{sec43} for more detailed formulations. This settles the open question \cite[Sec. 5, Question 2]{nisem} about the range of the embedded trace operator for our class of metric graphs, even for
a more general geometric trace realization.

For $d=1$ our result is very close to the construction of the bounded trace operator in \cite[Thm. 5.4.13]{semin} and \cite[Sec. 3.1--3.2]{jks}, but even in this case our result is stronger (for the class of trees we consider) as we show its surjectivity for a range of parameters.
\end{remark}

A large part of the paper is devoted to the proof of the assertions of Theorem \ref{thm11} for the case $c=1$ in \eqref{cc}. For this special case, the tree $\cT$ will be denoted by $\T$ and called \emph{geometric tree} following the convention proposed in \cite{jks}. The advantage of the geometric tree is that it allows for a decomposition
into a direct sum of one-dimensional problems, and trace theorems in one dimension are much simpler to study. Such a decomposition
is well-known \cite{nasol,pwz}, but we need a number of explicit formulas
for various intermediate transformation operators, which are missing in the existing literature, so we opted for a self-contained
presentation  in Section \ref{ttt}. This part of analysis
is concluded by constructing an abstract trace operator in Subsection \ref{sec-abstrace}, which maps $H^1(\T)$ into a discrete $\ell^2$-type space inherited from the direct sum decomposition.

In Section \ref{secappr} we introduce
approximation spaces $A^r(\Omega)$, which consist of the functions defined on $\Omega$ that can be ``well approximated'' by linear combinations of indicator functions of some subsets of $\Omega$. In Subsection \ref{sec31} we recall the most important constructions for fractional Sobolev spaces which are used in the analysis. In Subsection \ref{ssmult} we introduce special decompositions of Euclidean domains and define the associated approximation spaces. In Subsection \ref{sec33} we show that
in some important cases the approximation spaces coincide with the usual fractional Sobolev spaces. The constructions of Subsections \ref{ssmult} and \ref{sec33} are an adaptation
of the respective $2$-adic spaces in \cite{msv}, which were in turn motivated by more general considerations coming from the wavelet analysis \cite{cohen,meyer}. In Subsection \ref{sec44} we transfer these constructions to the case of open sets on manifolds
using the traditional approach with local charts. We note that Sections \ref{ttt}
and \ref{secappr} are independent from each other. They also contain a lot of introductory material and we hope that they can be of independent interest
beyond the immediate scope of the present work.

In Section \ref{sec4} we make last steps in the construction of the embedded trace operator.
First, in Subsection \ref{sec41} we identify the $\ell^2$-space from the construction of the abstract trace operator with the approximation spaces $A^r(\Omega)$ using an identification of suitable bases. The embedded trace operator is then obtained as the superposition of this identification with the abstract trace operator, and the resulting properties are summarized
in Subsection \ref{sec42}. At this point, all assertions of Theorem \ref{thm11} are proved for the geometric tree $\T$, and in Subsection \ref{sec43} we transfer them to the general $\cT$ using a coordinate change.

All preceding results require the existence of decompositions of open sets or manifolds into pieces with special properties; it seems that these questions
were not addressed in sufficient generality in earlier works. In the last Section \ref{sec52} we show that such decompositions exist for large classes of $\Omega$, in particular, for all convex polyhedrons, all convex smooth domains and all compact manifolds. This is done by adapting the existing results from  very diverse areas of analysis to the context of multiscale decompositions.

We consider the present work as an initial component for the systematic analysis
of boundary value and transmission problems on infinite metric graphs, which will be continued in several directions. A key role in our analysis
is played by the decomposition of trees into a direct sum of one-dimensional
problems. It was noted in \cite{breuer} that such decomposition actually exists
for a much larger class of metric graphs, so we hope that at least some elements of our analysis will be useful beyond the context of trees. The possibility of the identification of the tree boundary with a prescribed surface gives a possible
approach to describe the interaction between fractal trees touching each other along some interface and to include
fractal building blocks in the so-called hybrid spaces, \cite{bg,bgp,exner,hybrid,pry,post2}. Such applications will be covered in ongoing works.

\subsection*{Acknowledgments}
The authors thank Nadine Gro{\ss}e, Dorothee Haroske, Michael Hinz, Patrick Joly, Maryna Kachanovska, Massimo Lanza de Cristoforis, Marius Mitrea and Noema Nicolussi for useful comments on preliminary versions of the work and bibliographic hints. KN thanks the mathematics department of the University of Padua for the hospitality during the stay in November 2022. KP had several fruitful discussions on the topic of the article during
the workshop ``Spectral theory of differential operators in quantum theory"
at the Erwin Schr\"odinger International Institute for Mathematics and Physics in Vienna in November 2022, and he thanks the institute for the support provided.

\section{Analysis on geometric trees}\label{ttt}

\subsection{Tree structure and function spaces}

In this section we analyze in greater detail the ``ideal'' case $\ell_{n,k}:=\ell^n$ and $w_{n,k}=\alpha^n$ for all $(n,k)$, i.e. with $c=1$ in \eqref{cc}. The corresponding tree will be denoted $\T$ (as opposite to $\cT$ for the general case) and called a \emph{geometric tree}. The geometric trees have a lot of symmetries, which will be exploited for the analysis, and some expressions can be written in a slightly different form.

Remark that the underlying combinatorial
graph is $G:=(V,E)$, with the set of vertices $V$ and the set of edges $E$ given by
\begin{align*}
	V&:=\{o\}\cup \big\{X_{n,k}:\ n\in\N_0, \ k\in\{0,1,\dots,p^n-1\}
	\big\},\\
	E&:=\big\{e_{n,k}:\, n\in\N_0,\ k\in\{0,1,\dots,p^n-1\}\big\},\\
	e_{n,k}&:=\begin{cases}
		(0,X_{0,0}), & n=0,\\
		(X_{n-1, [\log_p k]},X_{n,k}), & n\ge 0,
	\end{cases}
\end{align*}
where $[t]$ stands for the integer part of $t\in[0,\infty)$, i.e. the largest integer not exceeding $t$. Remark that each edge $e_{n,k}$ connects each $X_{n,k}$ with its uniquely defined parent, which is $X_{n-1,[\log_p k]}$ for $n\ge 1$ and $o$ for $n=0$.
%
%
All vertices except the root have the degree $p+1$ (i.e. have $p+1$ neighbors: $p$ children and $1$ parent), and the degree of the root is~1. For $p=1$ the graph $G$ is simply a half-infinite chain, so we assume from now on that $p\ge 2$.

Consider the numbers
\begin{equation}
t_{-1}:=0, \quad
t_n=\sum_{k=0}^{n}\ell^k, \quad n\in\N_0,
\quad
L:=\lim_{n\to\infty} t_n\in (0,\infty].
\label{tnl}
\end{equation}
By construction, the numbers $t_{-1}<t_0<t_1<\dots$ subdivide $(0,L)$ into the infinitely many intervals $(t_{n-1},t_n)$ of length $\ell^n$, $n\in \N_0$ (Figure \ref{fig3}).
The combinatorial tree $G$ is related to the metric tree $\T$ as follows:
we identify each $e_{n,k}$ with a copy of $[t_{n-1},t_n]$ using the convention that the endpoint of $e_{n,k}$ is identified with the initial point of each of its children.
In other words, 
\[
\T:=\big\{ \big((n,k), t\big):\ n\in\N_0,\  k\in\{0,\dots,p^n-1\},\  t\in[t_{n-1},t_n]  \big\}/\sim
\]
for the identification $\sim$ defined by
\begin{gather*}
\big((n,k), t_n\big)\sim \big((n+1,pk+j),t_n\big),\\
\text{$n\in\N_0$, $k\in\{0,\dots,p^n-1\}$, $j\in\{0,\dots,p-1\}$}.
\end{gather*}

\begin{figure}[h]
	\centering
		\begin{tikzpicture}[font=\scriptsize,edge from parent/.style={draw}, scale=1]
		
		\tikzstyle{solid node}=[circle,draw,inner sep=1.2,fill=black];
		\tikzstyle{hollow node}=[circle,draw,inner sep=1.2];
		\tikzstyle{level 1}=[level distance=10mm,sibling distance=50mm]
		\tikzstyle{level 2}=[level distance=8mm,sibling distance=25mm]
		\tikzstyle{level 3}=[level distance=12mm,sibling distance=18mm]
		\tikzstyle{level 4}=[level distance=14mm,sibling distance=14mm]
		\node(0)[above]{$o$}
		child{node[above]{$X_{00}$}
			child{node[left]{$X_{1,0}$}
				child{node[below]{$X_{2,0}$}
					child{node[below]{} edge from parent[dotted]}
					child{node[below]{} edge from parent[dotted]}	
					edge from parent }
				child{node[below]{$X_{2,p-1}$}
					child{node[below]{} edge from parent[dotted]}
					child{node[below]{} edge from parent[dotted]}	
					edge from parent }
				edge from parent 
			}
			child{node[right]{$X_{1,p-1}$}
				child{node[below]{$X_{2,p(p-1)}$}
					child{node[below]{} edge from parent[dotted]}
					child{node[below]{} edge from parent[dotted]}	
					edge from parent }
				child{node[below]{$X_{2,p^2-1}$}
					child{node[below]{} edge from parent[dotted]}
					child{node[below]{} edge from parent[dotted]}	
					edge from parent }
				edge from parent 
			}
			edge from parent 
		};
		\node at (0,0) {};
		\node at (1,0) {{\large $\T$}};
		\node at (0,-1.4) {$\dots$};
		\node at (-1.7,-2.4) {$\dots$};
		\node at (1.9,-2.4) {$\dots$};
		\node at (2.75,-4) {$\dots$};
		\node at (1,-4) {$\dots$};
		\node at (-0.8,-4) {$\dots$};
		\node at (-2.6,-4) {$\dots$};
		
		\draw [thick, |->] (4,0) -- (4,-4.5);
		\node [thick] at (3.7,0) {0};
		\node [thick] at (4,-0.7) {--};
		\node [thick] at (3.5,-0.7) {$1=t_0$};
		\node [thick] at (4,-1.5) {--};
		\node [thick] at (3.35,-1.5) {$1+\ell=t_1$};
		\node [thick] at (4,-2.9) {--};
		\node [thick] at (3.7,-2.9) {$t_2$};
		\draw [decorate, decoration = {calligraphic brace}] (4.2,-0.1) -- (4.2,-0.6);
		\node [scale=0.8] at (5,-0.35) {length = 1};
		\draw [decorate, decoration = {calligraphic brace}] (4.2,-0.8) -- (4.2,-1.4);
		\node [scale=0.8] at (5,-1.15) {length = $\ell$};
		\draw [decorate, decoration = {calligraphic brace}] (4.2,-1.6) -- (4.2,-2.8);
		\node [scale=0.8] at (5,-2.2) {length = $\ell^2$};
		\node at (4,-4.7) {$L$};
		
		\node  at (-0.7, -0.15) {weight$=1$};
		\node  at (-1, -0.8) {$\alpha$};
		\node  at (1, -0.8) {$\alpha$};
		\node  at (1.2, -2) {$\alpha^2$};
		\node  at (2.55, -2) {$\alpha^2$};
		\node  at (-1, -2) {$\alpha^2$};
		\node  at (-2.35, -2) {$\alpha^2$};
	\end{tikzpicture}
\caption{The structure of a geometric tree}\label{fig3}
\end{figure}

For what follows for $x,y\in\T$ we write
\begin{itemize}
	\item $x\le y$ if the path from $o$ to $y$ passes through $x$ (equivalently one can say that $y$ belongs to the offspring of $x$),
	\item $x<y$ if $x\le y$ and $x\ne y$.
\end{itemize}

We will also need to consider some special subgraphs of $\T$. For $n\in\N$ denote
\[
\T^n:= \text{ the tree truncated after the $n$th generation,}
\]
i.e. $\T^n$ is composed of all edges $e_{m,k}$ with $m\le n$.
For $n\in \N$, $k\in\{0,\dots,p^{n}-1\}$, consider
\[
\T_{n,k}:=\{x\in\T:\ X_{n,k}\le x\}.
\]
Remark that 
\[
\T_{n,k}=\bigcup_{j=0}^{p-1} \T^j_{n,k}, \qquad
\T^j_{n,k}:=e_{n+1,pk+j}\cup \T_{n+1,pk+j},
\]
see Figure \ref{fig5}.
By construction, each $\T^j_{n,k}$ is a rooted metric tree (having the same combinatorial structure as $\T$ itself) with $X_{n,k}$ being the root. The vertices of $\T^j_{n,k}$ are $X_{n,k}$ and $X_{n+m,k p^m +jp^{m-1}+r}$
with $m\in \N$ and $r\in\{0,\dots,p^{m-1}-1\big\}$.

\begin{figure}[h]
\centering

	\begin{tikzpicture}[font=\footnotesize,edge from parent/.style={draw,thick}]
	
	\tikzstyle{solid node}=[circle,draw,inner sep=1.2,fill=black];
	\tikzstyle{hollow node}=[circle,draw,inner sep=1.2];
	\tikzstyle{level 1}=[level distance=8mm,sibling distance=20mm]
	\tikzstyle{level 2}=[level distance=8mm,sibling distance=30mm]
	\tikzstyle{level 3}=[level distance=12mm,sibling distance=18mm]
	\node(0)[above]{}
	child{node[left]{$X_{n,k}$}
		child{node[below]{$X_{n+1,kp}$}
			child{node[below]{} edge from parent[dotted]}
			child{node[below]{} edge from parent[dotted]}	
			edge from parent[solid]}
		child{node[below]{$X_{n+1,kp+(p-1)}$}
			child{node[below]{} edge from parent[dotted]}
			child{node[below]{} edge from parent[dotted]}	
			edge from parent[solid] }
		edge from parent[dashed]
	}
	child{node[right]{} edge from parent[dashed]
	};
	\node at (0,0) {};
	\node at (-1.5,-1.55) {$\dots$};
	\node at (-1.5,-3.3) {$\dots$};
	\draw[decorate,decoration={brace,amplitude=8pt}] (-2.1,-3.1) -- (-3.8,-3.1) node[midway, below,yshift=-10pt,]{$\T^0_{n,k}$};
	\draw[decorate,decoration={brace,amplitude=8pt}] (0.9,-3.1) -- (-0.8,-3.1) node[midway, below,yshift=-10pt,]{$\T^{p-1}_{n,k}$};
	\draw[decorate,decoration={brace,amplitude=6pt}] (-3.9,-3.0) -- (-3.9,-0.5) node[pos=0.5, left=0.2cm]{$\T_{n,k}$};
\end{tikzpicture}
\caption{The subtrees $\T_{n,k}$ and $\T^j_{n,k}$}\label{fig5}
\end{figure}

The set $\T$ becomes a metric space if one considers the natural distance $\rho$,
\begin{align*}
\rho(x,y)&:=\text{the length of the unique path between $x,y\in \T$,}\\
|x|&:=\rho(x,o) \text{ for } x\in \T,
\end{align*}
which gives rise to the notion of a continuous function on $\T$.
The number $L$ in \eqref{tnl} is usually referred to as the height of $\T$,
and
\[
\fbox{\begin{minipage}{100mm}for all subsequent constructions we assume
\[
L<\infty \qquad \text{ or, equivalently,} \qquad \ell<1.
\]
\end{minipage}
}
\]

We consider the measure $\mu$ on $\T$ which coincides with $\alpha^n\dd t$ along $e_{n,k}$, where $\dd t$ is the one-dimensional Lebesgue measure and $\alpha>0$ is a fixed constant. A function $f:\T\to \C$ is measurable if 
each of its components
\[
f_{n,k}: [t_{n-1},t_n]\ni t\mapsto f\big((n,k),t\big)\in\C, \qquad f_{n,k}:=f|_{e_{n,k}} \text{ for short,}
\]
is measurable; in most cases we will identify $f$ with the set of its components $(f_{n,k})$.
The integral of such $f$ over $\T$ with respect to $\mu$ is then given by
\[
\int_\T f \dd \mu:=\sum_{n=0}^\infty\sum_{k=0}^{p^n-1}\alpha^n\int_{t_{n-1}}^{t_n} f_{n,k}(t)\dd t.
\]

The above integration gives rise to the naturally defined space $L^2(\T)$:
\begin{align*}
	L^2(\T)&:=\big\{f:\T\to \C \text{ measurable}:\ \|f\|^2_{L^2(\T)}:=\int_\T |f|^2\dd\mu<\infty\big\},\\
	\int_\T |f|^2\dd\mu&:=\sum_{n=0}^\infty\sum_{k=0}^{p^n-1} \alpha^n\|f_{n,k}\|^2_{L^2(e_{n,k})},
		\quad
		\|f_{n,k}\|^2_{L^2(e_{n,k})}:=\int_{t_{n-1}}^{t_n} \big|f_{n,k}(t)\big|^2\dd t.
\end{align*}
In addition we consider the Sobolev-type space $H^1(\T)$ defined by
\begin{align*}
	H^1(\T):=\big\{f\in L^2(\T):\ &f_{n,k}\in H^1(t_{n-1},t_n) \text{ for any $(n,k)$,}\\
	&f':=(f'_{n,k})\in L^2(\T) \text{ and $f$ is continuous on $\T$}\big\}.
\end{align*}
Recall that $H^1(t_{n-1},t_n)\subset C^0\big([t_{n-1},t_n]\big)$ due to Sobolev embedding theorem, so the continuity
of $f$ on $\T$ in the above definition of $H^1(\T)$ actually means the continuity at the vertices,
\begin{equation*}
	f_{n,k}(t_n)=f_{n+1,pk+j}(t_n) \text{ for all $n\in\N_0$, $k\in\{0,\dots,p^n-1\}$, $j\in\{0,\dots,p-1\}$.}
\end{equation*}
We equip $H^1(\T)$ with the scalar product $\langle\cdot,\cdot\rangle_{H^1(\T)}$ defined by
\[
\langle f,g\rangle_{H^1(\T)}:=
\langle f,g\rangle_{L^2(\T)}
+
\langle f',g'\rangle_{L^2(\T)}
\]
and the induced norm $\|\cdot\|_{H^1(\T)}$, then one easily checks that  $H^1(\T)$
 becomes a Hilbert space. The following result is known from \cite[Section 3.5]{jks}, and it will be important below:

\begin{lemma}\label{prop-compact}
	For any $ p\in\N$, $\ell\in(0,1)$, $\alpha>0$ the embedding $H^1(\T) \hookrightarrow L^2(\T)$ is compact.
\end{lemma}

%


Now we denote
\begin{align*}
H^1_c(\T)&:=\big\{
f=(f_{n,k})\in H^1(\T):\, \text{ $\exists N\in \N$ such that $f_{n,k}\equiv 0$ for $n>N$}\big\},\\
H^1_0(\T)&:=\text{the closure of $H^1_c(\T)$ in $H^1(\T)$.}
\end{align*}


\subsection{Orthogonal decompositions}\label{sec-fsym}

All constructions in this section are essentially from the paper \cite{nasol}, but we need several intermediate objects which did not appear there explicitly, so we prefer to give a complete argument.

We now introduce several subspaces of $L^2(\T)$ determined through additional invariance properties.
We start with the space of radial functions,
\begin{align*}
	L^2_\rad(\T) := \big\{f\in L^2(\T):
		\text{ for any $x,y\in \T$ with $|x|=|y|$ one has $f(x)=f(y)$}\big\},
\end{align*}
which will be considered with the induced scalar product.
Remark that $f\in L^2_\rad(\T)$ means the existence of a function $F:(0,L)\to \C$ with $f(x)=F(|x|)$ for all $x\in \T$, which means that the components $(f_{n,k})$ satisfy $f_{n,k}(t)=F(t)$ for all $e_{n,k}\in E$ and $t\in(t_{n-1},t_n)$.
This yields
\begin{align*}
\|f\|^2_{L^2(\T)}&=\sum_{n=0}^\infty\alpha^n\sum_{k=0}^{p^n-1} \int_{t_{n-1}}^{t_n}\big|f_{n,k}(t)\big|^2\dd t=\sum_{n=0}^\infty\alpha^n\sum_{k=0}^{p^n-1} \int_{t_{n-1}}^{t_n}\big|F(t)\big|^2\dd t\\
&=\sum_{n=0}^\infty\alpha^n p^n \int_{t_{n-1}}^{t_n}\big|F(t)\big|^2\dd t=\int_0^L \big|F(t)\big|^2\, q(t)\dd t
\end{align*}
for the weight function
\[
q:\ (0,L)\mapsto (0,\infty), \quad
q(t)=(\alpha p)^n \text{ for }t\in (t_{n-1},t_n).
\]
The above computation shows that the map
\[
	U_\rad:\ L^2\big((0,L), q(t)\dd t\big) \to L^2_\rad (\T), \qquad
	U_\rad F:\  \T\ni x\mapsto F(|x|),
\]
is a unitary operator.

Furthermore, consider the roots of unity:
\[
\theta_s:=e^{\frac{2\pi i}{p}s}, \quad s\in\{0,\dots,p-1\},
\]
and define, for $e_{n,k} \in E$ and $s\in\{1,\dots,p-1\}$,
\begin{align*}
	L^2_{n,k,s}(\T)&:= \big\{f \in L^2(\T):\ f|_{\T \setminus \T_{n,k}} = 0 \text{ and for any $j,j'\in\{0,\dots,p-1\}$ and}\\
	&\text{$x\in \T_{n,k}^j$ and $y\in \T_{n,k}^{j'}$ with $|x|=|y|$ one has $\theta_s^{-j} f(x)= \theta_s^{-j'}f(y)$}\big\},
\end{align*} 
Observe that each function $f\in L^2_{n,k,s}(\T)$ is radial on each subtree $\T_{n,k}^j$, which means that $f(x)$ only depends on the distance $\rho(x,X_{n,k})$ for all $x\in \T_{n,k}^j$.
More precisely, for some functions $F_j:(t_n,L)\to \C$, $j\in\{0,\dots,p-1\}$, there holds $f(x)=F_j\big(|x|\big)$ for all $x\in \T_{n,k}^j$, and, in addition, $F_j=\theta_s^j F_0$ for each $j\in\{0,\dots,p-1\}$, so $f$ is uniquely determined by $F:=F_0$.
By construction we have
\begin{align*}
	&\int_{\T_{n,k}^j} |f|^2\dd\mu= \sum_{m=1}^\infty \alpha^{n+m}
	\sum_{r=0}^{p^{m-1}-1} \int_{t_{n+m-1}}^{t_{n+m}} \big|f_{n+m,k p^{m}+jp^{m-1}+r}(t)\big|^2\dd t\\
	&=\sum_{m=1}^\infty \alpha^{n+m}
	\sum_{r=0}^{p^{m-1}-1} \int_{t_{n+m-1}}^{t_{n+m}} \big|F_j(t)\big|^2\dd t
	=\sum_{m=1}^\infty \alpha^{n+m} p^{m-1} \int_{t_{n+m-1}}^{t_{n+m}} \big|F_j(t)\big|^2\dd t\\
	&=p^{-n-1}\sum_{m=n+1}^\infty (\alpha p)^{m} \int_{t_{m-1}}^{t_{m}} \big|F_j(t)\big|^2\dd t
	=p^{-n-1} \int_{t_n}^L \big|F_j(t)\big|^2\,q(t)\dd t,
\end{align*}	
and
\begin{align*}
	\int_\T |f|^2\dd\mu&=\int_{\T_{n,k}} |f|^2\dd\mu=\sum_{j=0}^{p-1}\int_{\T_{n,k}^j} |f|^2\dd\mu\\
	&= p^{-n-1}\sum_{j=0}^{p-1} \int_{t_n}^L \big|F_j(t)\big|^2\,q(t)\dd t=p^{-n-1}\sum_{j=0}^{p-1} \int_{t_n}^L \big|\theta_s^j F(t)\big|^2\,q(t)\dd t\\
	&=p^{-n-1}\sum_{j=0}^{p-1} \int_{t_n}^L \big|F(t)\big|^2\,q(t)\dd t=p^{-n}\int_{t_n}^L \big|F(t)\big|^2\,q(t)\dd t.
\end{align*}
This computation shows that the map
\begin{gather*}
U_{n,k,s}:L^2\big((t_n,L),p^{-n}q(t)\dd t\big)\to L^2_{n,k,s}(\T),\\
U_{n,k,s}F:\  x\mapsto \begin{cases}
	\theta_s^j F\big(|x|\big), & \text{if $x\in \T_{n,k}^j$ for some $j\in\{0,\dots,p-1\}$},\\
	0, & \text{otherwise}.
\end{cases}
\end{gather*}
is a unitary operator.

We further denote
\[
H^1_\rad(\T):=H^1(\T)\cap L^2_\rad(\T), \quad
H^1_{n,k,s}(\T):=H^1(\T)\cap L^2_{n,k,s}(\T),
\]
and note the obvious implications
\begin{equation}
	\label{prim-inv}
f\in H^1_\rad(\T)\,\Rightarrow\,f'\in L^2_\rad(\T),
\quad
f\in H^1_{n,k,s}(\T)\,\Rightarrow\,f'\in L^2_{n,k,s}(\T).
\end{equation}

\begin{theorem}[\cite{nasol}]
	\label{thm-expansion1}
One has the orthogonal direct sum decomposition
	\begin{equation}
		\label{expansionl2}
		L^2(\T) = L^2_\rad(\T) \oplus \bigoplus\limits_{e_{n,k}\in E}\bigoplus\limits_{s=1}^{p-1} L^2_{n,k,s}(\T). 
	\end{equation}
\end{theorem}
\begin{proof} (A) We begin with the orthogonality. Let $f\in L^2_{n,k,s}(\T)$ for some $(n,k,s)$, then
	$f=U_{n,k,s} F$ for some $F\in L^2\big((t_n,L),p^{-n}q(t)\dd t\big)$.

	(A.1) Let $g\in L^2_\rad (\T)$, then $g=U_\rad G$ for some $G\in L^2\big((0,L),q(t)\dd t\big)$ and
\begin{align*}
	\langle f,g\rangle_{L^2(\T)}&=\sum_{j=0}^{p-1} \int_{\T^j_{n,k}} f \Bar g\dd\mu\\
	=\sum_{j=0}^{p-1} \sum_{m=1}^\infty& \alpha^{n+m}
	\sum_{r=0}^{p^{m-1}-1} \int_{t_{n+m-1}}^{t_{n+m}} f_{n+m,k p^{m}+jp^{m-1}+r}(t)\,\overline{g_{n+m,k p^{m}+jp^{m-1}+r}(t)}\dd t\\
	&= \sum_{j=0}^{p-1} \sum_{m=1}^\infty \alpha^{n+m}\sum_{r=0}^{p^{m-1}-1} \int_{t_{n+m-1}}^{t_{n+m}} \theta_s^j F(t)\,\overline{G(t)}\dd t\\
	&= \sum_{j=0}^{p-1} \sum_{m=1}^\infty \alpha^{n+m}p^{m-1}\int_{t_{n+m-1}}^{t_{n+m}} \theta_s^j F(t)\,\overline{G(t)}\dd t\\
	&= \sum_{j=0}^{p-1} p^{-n-1} \theta_s^j\sum_{m=n+1} (\alpha p)^m  \int_{t_{m-1}}^{t_{m}}  F(t)\,\overline{G(t)}\dd t\\
	&= p^{-n-1} \Big(\sum_{j=0}^{p-1} \theta_s^j\Big) \int_{t_n}^L F(t)\,\overline{G(t)} \,q(t)\dd t
\end{align*}
while
\begin{equation}
	  \label{thetas}
\sum_{j=0}^{p-1} \theta_s^j=\dfrac{1-\theta_s^p}{1-\theta_s}\equiv \dfrac{1-1}{1-\theta_s}\equiv 0, 
\end{equation}
so $\langle f,g\rangle_{L^2(\T)}=0$.

(A.2) Now let $g\in L^2_{n',k',s'}(\T)$ with $(n,k,s)\neq (n',k',s')$, then $g=U_{n',k',s'} G$
for some $G\in L^2\big((t_n,L),p^{-n}q(t)\dd t\big)$. We have three possiblities:

(A.2.1) None of $X_{n,k}\le X_{n',k'}$ and $X_{n',k'}\le X_{n,k}$ holds. In this case one has $\T_{n,k}\cap \T_{n',k'}=\emptyset$, so $f$ and $g$ have disjoint supports and $\langle f,g\rangle_{L^2(\T)}=0$.

(A.2.2) One has $X_{n,k}<X_{n',k'}$ or $X_{n,k}<X_{n',k'}$. To be definite assume that
$X_{n',k'}<X_{n,k}$, then $\T_{n,k}\subset \T^{j'}_{n',k'}$ for some $j'\in\{0,\dots,p-1\}$.
\begin{align*}
	\langle& f,g\rangle_{L^2(\T)}=\int_{\T} f\Bar g\dd\mu=\int_{\T_{n,k}} f\Bar g\dd\mu=\sum_{j=0}^{p-1}
	\int_{\T^j_{n,k}} f\Bar g\dd\mu\\
	&=\sum_{j=0}^{p-1}\sum_{m=1}^\infty \alpha^{m+n}\sum_{r=0}^{p^{m-1}-1} \int_{t_{n+m-1}}^{t_{n+m}} f_{n+m,k p^{m}+jp^{m-1}+r}(t)\,\overline{g_{n+m,k p^{m}+jp^{m-1}+r}(t)}\dd t.
\end{align*}
By assumptions in each summand we have
\begin{align*}
	f_{n+m,k p^{m}+jp^{m-1}+r}(t)&= \theta_s^j F(t),&
	g_{n+m,k p^{m}+jp^{m-1}+r}(t)&= \theta_{s'}^{j'} G(t),
\end{align*}	
so similarly to the preceding computation we obtain
\begin{align*}
\langle f,g\rangle_{L^2(\T)}&=\sum_{j=0}^{p-1} \sum_{m=1}^\infty \alpha^{m+n}\sum_{r=0}^{p^{m-1}-1} \theta_s^j\theta_{s'}^{-j'}
\sum_{r=0}^{p^{m-1}-1} \int_{t_{n+m-1}}^{t_{n+m}} F(t)\,\overline{G(t)}\dd t\\
&=\sum_{j=0}^{p-1} \sum_{m=1}^\infty \alpha^{m+n}p^{m-1} \theta_s^j\theta_{s'}^{-j'}
\int_{t_{n+m-1}}^{t_{n+m}} F(t)\,\overline{G(t)}\dd t\\
&=\sum_{j=0}^{p-1} \sum_{m=n+1}^\infty (\alpha p)^m p^{-n-1} \theta_s^j\theta_{s'}^{-j'}
\int_{t_{m-1}}^{t_{m}} F(t)\,\overline{G(t)}\dd t\\
&=p^{-n-1} \theta_{s'}^{-j'} \underbrace{\sum_{j=0}^{p-1} \theta_s^j}_{=0 \text{ by }\eqref{thetas}} \int_{t_n}^L F(t)\,\overline{G(t)} q(t)\dd t=0.
\end{align*}

(A.2.3) One has $X_{n,k}=X_{n',k'}$. In this case $s\ne s'$,
\begin{align*}
	f_{n+m,k p^{m}+jp^{m-1}+r}(t)&= \theta_s^j F(t),&
	g_{n+m,k p^{m}+jp^{m-1}+r}(t)&= \theta_{s'}^{j} G(t),
\end{align*}	
for all $m\in\N$, $j\in\{0,\dots,p-1\}$ and $r\in\{0,\dots,p^{m-1}-1\}$, 
and
\begin{align*}
	\langle& f,g\rangle_{L^2(\T)}=\int_{\T} f\Bar g\dd\mu=\int_{\T_{n,k}} f\Bar g\dd\mu=\sum_{j=0}^{p-1}
	\int_{\T^j_{n,k}} f\Bar g\dd\mu\\
	&=\sum_{j=0}^{p-1}\sum_{m=1}^\infty \alpha^{m+n}\sum_{r=0}^{p^{m-1}-1} \int_{t_{n+m-1}}^{t_{n+m}} f_{n+m,k p^{m}+jp^{m-1}+r}(t)\,\overline{g_{n+m,k p^{m}+jp^{m-1}+r}(t)}\dd t\\
	&= \sum_{j=0}^{p-1}\sum_{m=1}^\infty \alpha^{m+n}\sum_{r=0}^{p^{m-1}-1} \theta_s^j\theta_{s'}^{-j}\int_{t_{n+m-1}}^{t_{n+m}} F(t)\,\overline{G(t)}\dd t\\
	&= \sum_{j=0}^{p-1}\sum_{m=1}^\infty \alpha^{m+n}p^{m-1}\theta_s^j\theta_{s'}^{-j}\int_{t_{n+m-1}}^{t_{n+m}} F(t)\,\overline{G(t)}\dd t\\
	&=\sum_{j=0}^{p-1} \big(\frac{\theta_s}{\theta_{s'}}\big)^j \sum_{m=n+1} (\alpha p)^m p^{-n-1} \int_{t_{m-1}}^{t_m} F(t)\,\overline{G(t)}\dd t\\
	&=p^{-n-1} \sum_{j=0}^{p-1} \big(\frac{\theta_s}{\theta_{s'}}\big)^j \int_{t_n}^L F(t)\,\overline{G(t)}\,q(t)\dd t,
\end{align*}
while
\[
 \sum_{j=0}^{p-1} \big(\frac{\theta_s}{\theta_{s'}}\big)^j=\dfrac{1-\big(\frac{\theta_s}{\theta_{s'}}\big)^p}{1-\frac{\theta_s}{\theta_{s'}}}
 \equiv
 \dfrac{1-1}{1-\frac{\theta_s}{\theta_{s'}}}=0,
\]
so $\langle f,g\rangle_{L^2(\T)}=0$.

The orthogonality of the decomposition is completely proved.

(B) In order to prove the totality of the decomposition it is sufficient to show that any function supported on a single edge belongs to the right-hand side of \eqref{expansionl2}.
So let $e_{m,r}\in E$ and $h\in L^2(t_{m-1},t_m)$. Define  $v\in L^2(\T)$ by
\[
v_{n,k}:=\begin{cases}
	h, & (n,k)=(m,r),\\
	0, & \text{otherwise},
\end{cases}
\]
then we need to show that $v$ belongs to the right-hand side of \eqref{expansionl2}.

Consider the $p^m$-dimensional subspace
\begin{align*}
	S := \big\{
	u\in L^2(\T):\  & u_{n,k}\in \C h \text{ for $n=m$ and $k=0,\dots,p^m-1$, }\\
	&u_{n,k}=0 \text{ for $n\ne m$}	\big\},
\end{align*}
then by construction we have $v\in S$. In addition,
\begin{align*}
	S\cap L^2_{n,k,s}(\T)&=\{0\} \text{ for $n\ge m$ and any $(k,s)$,}\\
	\dim \big(S\cap L^2_\rad(\T)\big)&=1,\\
	\dim \big(S\cap L^2_{n,k,s}(\T)\big)&=1  \text{ for $n<m$ and any $(k,s)$.}
\end{align*}
This yields
\begin{align*}
	\dim \Big[S &\cap \Big( L^2_\rad (\T) \oplus \bigoplus_{e_{n,k}\in E}\bigoplus_{s=1}^{p-1} L^2_{n,k,s}(\T)\Big)\Big]\\
	&=\dim \big(S\cap L^2_\rad(\T)\big)
	+ \sum_{n=0}^{m-1}\sum_{k=0}^{p^{n}-1}\sum_{s=1}^{p-1} \dim \big(S\cap L^2_{n,k,s}(\T)\big)\\
	&= 1 + \sum_{n=0}^{m-1}\sum_{k=0}^{p^{n}-1}\sum_{s=1}^{p-1} 1
	= 1 +\sum_{n=0}^{m-1} p^n(p-1)=p^m=\dim S,
\end{align*}
which shows that
\[
S\subset L^2_\rad (\T) \oplus \bigoplus_{e_{n,k}\in E}\bigoplus_{s=1}^{p-1} L^2_{n,k,s}(\T),
\]
and due to $v\in S$ we arrive at the conclusion.
\end{proof}

By using the inclusions \eqref{prim-inv} we immediately obtain:
\begin{corol}
	One has the orthogonal decomposition
	\begin{align*}
		H^1(\T)= H^1_\rad (\T) \oplus \bigoplus\limits_{e_{n,k}\in E}\bigoplus\limits_{s=1}^{p-1} H^1_{n,k,s}(\T).
	\end{align*}
\end{corol}

%


\begin{corol}\label{corol-l2}
The map
\[
U:=U_\rad \oplus \bigoplus_{e_{n,k}\in E}\bigoplus_{s=1}^{p-1} U_{n,k,s}
\]
defines a unitary operator
\begin{multline*}
L^2\big((0,L),q(t)\dd t\big)\oplus \bigoplus_{e_{n,k}\in E}\bigoplus_{s=1}^{p-1} L^2\big((t_n,L),p^{-n} q(t)\dd t\big)\\
\to
L^2_\rad (\T) \oplus \bigoplus_{e_{n,k}\in E}\bigoplus_{s=1}^{p-1} L^2_{n,k,s}(\T).
\end{multline*}
\end{corol}

In order to have a similar transformation for $H^1$-spaces, we recall that for an interval $(a,b)\subset\R$ and a piecewise continuous weight function $v:(a,b)\to (0,\infty)$ one usually denotes
\[
H^1\big((a,b),v(t)\dd t\big):=\big\{f\in L^2\big((a,b),v(t)\dd t\big):\, f'\in L^2\big((a,b),v(t)\dd t\big)\big\}
\]
which becomes a Hilbert space if considered with the scalar product
\begin{align*}
	\langle f,g\rangle_{H^1\big((a,b),v(t)\dd t\big)}&=\langle f,g\rangle_{L^2((a,b),v(t)\dd t)}+\langle f',g'\rangle_{L^2((a,b),v(t)\dd t)}\\
	&\equiv \int_a^b \Big ( f(t)\overline{g(t)} + f'(t)\overline{g'(t)}\Big)v(t)\dd t.
\end{align*}	
In addition we make the following observations:
\begin{itemize}
\item if $F\in L^2\big((0,L),q(t)\dd t\big)$, then:
\begin{itemize}
	\item $U_\rad F$ is continuous on $\T$ if and only if $F$ is continuous on $(0,L)$,
	\item $U_\rad F\in H^1(\T)$ if and only if $F\in H^1\big((0,L),q(t)\dd t\big)$,
	\item and in this case we have 
	\[
	\|(U_\rad F)'\|^2_{L^2(\T)}=\|U_\rad (F')\|^2_{L^2(\T)}=\|F'\|^2_{L^2((0,L),q(t)\dd t)},
	\]
\end{itemize}	
	\item if $F\in L^2\big((t_n,L),p^{-n}q(t)\dd t\big)$, then:
	\begin{itemize}
	\item $U_{n,k,s} F$ is continuous on $\T$ if and only if $F$ is continuous on $[t_n,L)$ with $F(t_n)=0$,
	\item in this case $U_{n,k,s} F\in H^1(\T)$ if and only if $F\in H^1\big((t_n,L),p^{-n}q(t)\dd t\big)$,
	\item and in this case 
	\[
    \|(U_{n,k,s} F)'\|^2_{L^2(\T)}=\|U_{n,k,s} (F')\|^2_{L^2(\T)}=\|F'\|^2_{L^2((t_n,L),p^{-n}q(t)\dd t)}.
    \]
	 \end{itemize}
\end{itemize}	
Therefore, it will be convenient to denote
\[
\Tilde H^1\big((t_n,L),p^{-n} q(t)\dd t\big):=\big\{
f\in H^1\big((t_n,L),p^{-n} q(t)\dd t\big):\, f(t_n)=0
\big\},
\]
then it follows that
\begin{align*}
U_\rad: H^1\big((0,L),q(t)\dd t\big)&\to H^1_\rad(\T),\\
U_{n,k,s}:\Tilde H^1\big((t_n,L),p^{-n} q(t)\dd t\big)&\to H^1_{n,k,s}(\T)
\end{align*}
are unitary operators, which gives:
\begin{corol}\label{corol-l3}
	The map
	\[
	U:=U_\rad \oplus \bigoplus_{e_{n,k}\in E}\bigoplus_{s=1}^{p-1} U_{n,k,s}
	\]
	defines a unitary operator
	\begin{align*}
	H^1\big((0,L),q(t)\dd t\big)\oplus \bigoplus_{e_{n,k}\in E}\bigoplus_{s=1}^{p-1} &\,\Tilde H^1\big((t_n,L),p^{-n} q(t)\dd t\big)\\
	&\to\,H^1_\rad (\T) \oplus \bigoplus_{e_{n,k}\in E}\bigoplus_{s=1}^{p-1} H^1_{n,k,s}(\T)
	\equiv H^1(\T).
	\end{align*}
\end{corol}

\subsection{Embeddings and equivalent norms}

We denote
\begin{equation*}
	\Tilde H^1(\T) := \big\{f \in H^1(\T):\, f(o)=0\big\},
\end{equation*}
which is a closed subspace of $H^1(\T)$.

\begin{lemma}[Poincaré inequality]\label{lem-poincare}
There is a constant $C>0$ such that
\begin{equation}
	 \label{poincare}
	\| f\|_{L^2(\T)} \leq C\|f'\|_{L^2(\T)} \text{ for all } f\in \Tilde H^1(\T).
\end{equation}
\end{lemma} 

\begin{proof}
For \eqref{poincare} it is sufficient to show that the number
\[
a:=\inf_{f\in \Tilde H^1(\T),\, f\not\equiv  0} \dfrac{\|f'\|^2_{L^2(\T)}}{\|f\|^2_{L^2(\T)}}
\]
is strictly positive (then one can choose $C:=1/\sqrt{a}$). Remark that
\[
\Tilde H^1(\T)\times \Tilde H^1(\T)\ni (f,g)\mapsto \langle f',g'\rangle_{L^2(\T)}
\]
is a closed, symmetric, non-negative, densely defined sesquilinear form in $L^2(\T)$, so it generates a non-negative self-adjoint
operator $A$ in $L^2(\T)$, and the above number $a$ is the bottom of the spectrum of $A$.
By Lemma \ref{prop-compact} the embedding $\Tilde H^1(T)\hookrightarrow L^2(\T)$ is compact, then $A$ has compact resolvent, and $a$ is its smallest eigenvalue due to the min-max principle. In particular, the infimum in the definition
of $a$ is attained on some eigenfunction $F\in \Tilde H^1(\T)$.
We have obviously $a\ge 0$. Assume that $a=0$,
then $\|F'\|^2_{L^2(\T)}=0$. This means that $F'_{n,k}=0$ for all $(n,k)$, i.e. all $F_{n,k}$ are constant.
The continuity of $F$ shows that $F$ is constant on $\T$, and $F\in \Tilde H^1(\T)$ means that $F(o)=0$ and then $F\equiv 0$,
which is impossible. Hence, $a>0$.
\end{proof}

In view of Lemma~\ref{lem-poincare}, the norm induced by the $\Tilde H^1$-scalar product
\[
\langle f ,g \rangle_{\Tilde H^1(\T)}:=\langle f',g'\rangle_{L^2(\T)}
\]
is equivalent to the initial $H^1(\T)$-norm. Moreover, we have the direct sum decomposition
\[
\Tilde H^1(\T)=\Tilde H^1_\rad (\T) \oplus \bigoplus_{e_{n,k}\in E}\bigoplus_{s=1}^{p-1} \Tilde H^1_{n,k,s}(\T),
\]
with
\begin{align*}
	\Tilde H^1_\rad (\T)&:=\Tilde H^1(\T)\cap H^1_\rad(\T)\equiv \{
	f\in H^1_\rad(\T):\, f(o)=0\},\\
	\Tilde H^1_{n,k,s}(\T)&:=\Tilde H^1(\T)\cap H^1_{n,k,s}(\T)\equiv \{
	f\in H^1_{n,k,s}(\T):\, f(o)=0\}\equiv H^1_{n,k,s}(\T),
\end{align*}
which is orthogonal with respect to both $\langle \cdot ,\cdot \rangle_{H^1(\T)}$
and $\langle \cdot ,\cdot \rangle_{\Tilde H^1(\T)}$.

We further remark that for $F\in H^1\big((0,L),q(t)\dd t\big)$
the inclusion $U_\rad F\in \Tilde H^1_\rad(\T)$ is equivalent to $F(0)=0$. So if we additionally denote
\[
\Tilde H^1\big((0,L),q(t)\dd t\big):=\big\{
F\in H^1\big((0,L),q(t)\dd t\big):\, F(0)=0
\big\},
\]
then we conclude that the map	
\[
	U:=U_\rad \oplus \bigoplus_{e_{n,k}\in E}\bigoplus_{s=1}^{p-1} U_{n,k,s}
\]
is an isomorphism between 
\[
\Tilde H^1\big( (0,L),q(t)\dd t\big)\oplus \bigoplus_{e_{n,k}\in E}\bigoplus_{s=1}^{p-1} \Tilde H^1\big((t_n,L),p^{-n}q(t)\dd t\big)
\]
equipped with the usual $H^1$ scalar products and
\begin{equation}
	\label{decomp1}
\Tilde H^1(\T)\equiv\Tilde H^1_\rad (\T) \oplus \bigoplus_{e_{n,k}\in E}\bigoplus_{s=1}^{p-1} H^1_{n,k,s}(\T)
\end{equation}
viewed with the $\Tilde H^1$ scalar product.

We will additionally introduce the space
\begin{align}
	\Tilde H^1_c(\T)&=\{f\in H^1_c(\T):\, f(o)=0\big\},\\
	\Tilde H^1_0(\T)&=\{f\in H^1_0(\T):\, f(o)=0\big\},
\end{align}
then it is standard to see that $\Tilde H^1_0(\T)$ is the closure of $\Tilde H^1_c(\T)$ in $\Tilde H^1(\T)$.
In particular, $\Tilde H^1_0(\T)$ is a closed subspace of $H^1(\T)$, so using the orthogonal decomposition \eqref{decomp1} we arrive at the decomposition
\begin{equation}
	\label{decomp2}
	\Tilde H^1_0(\T)=\Tilde H^1_{0|\rad} (\T) \oplus \bigoplus_{e_{n,k}\in E}\bigoplus_{s=1}^{p-1} H^1_{0|n,k,s}(\T)
\end{equation}
with
\begin{align*}
	\Tilde H^1_{0|\rad} (\T)&:=\Tilde H^1_\rad(\T)\cap H^1_0(\T),\\
	H^1_{0|n,k,s}(\T)&:= H^1_{n,k,s}(\T)\cap H^1_0(\T),
\end{align*}
which is orthogonal with respect to the $\Tilde H^1$-scalar product.
Further remark that for $F\in \Tilde H^1\big((0,L),q(t)\dd t\big)$ the inclusion
$U_\rad F\in H^1_c(\T)$ is equivalent to
\begin{align*}
F\in \Tilde H^1_c\big((0,L),q(t)\dd t\big)\equiv\big\{&F\in \Tilde H^1\big((0,L),q(t)\dd t\big):\\
&\, \exists a\in(0,L) \text{ such that $F(t)=0$ for all $t\ge a$}\big\}.
\end{align*}
Recall that $U_\rad$ is an isomorphism between $\Tilde H^1\big((0,L),q(t)\dd t\big)$ and  $\Tilde H^1_\rad(\T)$,
and $\Tilde H^1_\rad(\T)\cap H^1_c(\T)$ is dense in $\Tilde H^1_{0|\rad}(\T)$, so it follows
that
\begin{align*}
U_\rad^{-1}\Tilde H^1_{0|\rad}(\T)&=\text{the closure of $\Tilde H^1_c\big((0,L),q(t)\dd t\big)$ in $\Tilde H^1\big((0,L),q(t)\dd t\big)$}\\
&=:\Tilde H^1_0\big((0,L),q(t)\dd t\big).
\end{align*}
Remark that the usual mollifying procedure shows that
\[
\Tilde H^1_0\big((0,L),q(t)\dd t\big)=\text{the closure of $C_c^\infty(0,L)$ in $\Tilde H^1\big((0,L),q(t)\dd t\big)$},
\]
which will be useful below.

Similarly one shows that for any $(n,k,s)$ there holds
\begin{align*}
	U_{n,k,s}^{-1} H^1_{0|n,k,s}(\T)&=\Tilde H^1_0\big((t_n,L),p^{-n}q(t)\dd t\big)\\
	&:=\text{the closure of $C_c^\infty(t_n,L)$ in $\Tilde H^1\big((t_n,L),p^{-n}q(t)\dd t\big)$},
\end{align*}
and we conclude that the map
\[
U:=U_\rad \oplus \bigoplus_{e_{n,k}\in E}\bigoplus_{s=1}^{p-1} U_{n,k,s}
\]
is an isomorphism between 
\[
\Tilde H^1_0\big( (0,L),q(t)\dd t\big)\oplus \bigoplus_{e_{n,k}\in E}\bigoplus_{s=1}^{p-1} \Tilde H^1_0\big((t_n,L),p^{-n}q(t)\dd t\big)
\]
equipped with the usual $H^1$ scalar products and
\[
	\Tilde H^1_0(\T)\equiv\Tilde H^1_{0|\rad} (\T) \oplus \bigoplus_{e_{n,k}\in E}\bigoplus_{s=1}^{p-1} H^1_{0|n,k,s}(\T)
\]
viewed with the $\Tilde H^1$ scalar product.

\subsection{Harmonic functions}

By construction, $\Tilde H^1_0(\T)$ is a closed subspace of $\Tilde H^1(\T)$, so let us introduce
the subspace
\[
\Tilde H^1_\Delta(\T):=\text{the orthogonal complement of $\Tilde H^1_0(\T)$ in $\Tilde H^1(\T)$.}
\]
We would like to understand the structure of this subspace, in particular, to construct an orthonormal basis
(recall that the orthogonality is understood with respect to the $\Tilde H^1$-scalar product).

Remark first that in view of the orthogonal decompositions \eqref{decomp1} and \eqref{decomp2} we have the orthogonal decomposition
\begin{equation}
	\label{decomp3}
	\Tilde H^1_\Delta(\T)\equiv\Tilde H^1_{\Delta|\rad} (\T) \oplus \bigoplus_{e_{n,k}\in E}\bigoplus_{s=1}^{p-1} H^1_{\Delta|n,k,s}(\T)
\end{equation}
with
\begin{align*}
	\Tilde H^1_{\Delta|\rad} (\T)&:=\text{the orthogonal complement of $\Tilde H^1_{0|\rad}(\T)$ in $\Tilde H^1_\rad(\T)$},\\	
	 H^1_{\Delta|n,k,s} (\T)&:=\text{the orthogonal complement of $H^1_{0|n,k,s}(\T)$ in $ H^1_{n,k,s}(\T)$}.
\end{align*}

Let $f_\rad\in\Tilde H^1_\rad(\T)$, then $f_\rad\in \Tilde H^1_{\Delta|\rad}(\T)$ if and only if
\begin{equation}
	 \label{fg1}
\langle f_\rad',g'\rangle_{L^2(\T)}\equiv	
\langle f_\rad ,g\rangle_{\Tilde H^1(\T)}=0 \text{ for all } g\in \Tilde H^1_{0|\rad}(\T).
\end{equation}
We write $f_\rad=U_\rad F_\rad$ with $F_\rad\in \Tilde H^1\big((0,L),q(t)\dd t\big)$
and $g=U_\rad G$ with $G\in \Tilde H^1_0\big((0,L),q(t)\dd t\big)$
and recall that $C^\infty_c(0,L)$ is dense in $\Tilde H^1_0\big((0,L),q(t)\dd t\big)$,
then it follows that \eqref{fg1} holds if and only if
\[
\big\langle (U_\rad F_\rad)',(U_\rad G)'\big\rangle_{L^2(\T)}=0 \text{ for all } G\in C^\infty_c(0,L).
\]
As $(U_\rad F_\rad)'=U_\rad (F_\rad')$ and $(U_\rad G)'=U_\rad (G')$
and $U_\rad$ is unitary as an operator $L^2\big((0,L), q(t)\dd t\big) \to L^2_\rad (\T)$, this is equivalent to
\begin{equation}
\int_0^L F_\rad'(t)\overline{G'(t)}q(t)\dd t\equiv \langle F_\rad',G'\rangle_{L^2((0,L),q(t)\dd t)}=0
\text{ for all } G\in C^\infty_c(0,L),
\end{equation}	
which means that $(qF_\rad')'=0$ in $(0,L)$ in the sense of distributions, which shows that $F'_\rad=c/q$ for some $c\in \C$. We also recall that $F_\rad$ must be continuous with $F_\rad(0)=0$, so $F_\rad$ is completely determined by the value $c$ of its derivative in $(0,1)$: for any $n\in \N_0$ and $t\in (t_{n-1},t_n)$ one has
\begin{align*}
	F_\rad'(t)&= \dfrac{c}{(\alpha p)^n},\\
	F_\rad (t)&=c\Big(\dfrac{t-t_{n-1}}{(\alpha p)^n} + \sum_{k=0}^{n-1} \big(\frac{\ell}{\alpha p}\big)^k\Big)
	\equiv c\Big(\dfrac{t-t_{n-1}}{(\alpha p)^n} + \dfrac{1-\big(\frac{\ell}{\alpha p}\big)^n}{1-\frac{\ell}{\alpha p}}\Big).
\end{align*}
It remains to check if $F_\rad\in \Tilde H^1\big((0,L),q(t)\dd t\big)$, for $c\ne 0$. We have
\begin{align*}
	\|F'_\rad\|^2_{L^2((0,L),q(t)\dd t)}&=\int_0^L \big|F'_\rad(t)\big|^2 q(t)\dd t=\sum_{n=0}^{\infty} (\alpha p)^n\int_{t_{n-1}}^{t_n}\big|F'_\rad(t)\big|^2\dd t
	\\
	&= \sum_{n=0}^{\infty} (\alpha p)^n \dfrac{|c|^2}{(\alpha p)^{2n}}(\underbrace{t_n -t_{n-1}}_{=\ell^n})
	=|c|^2\sum_{n=0}^\infty \Big(\dfrac{\ell}{\alpha p}\Big)^n,
\end{align*}
which is finite if and only if $\ell<\alpha p$, which we assume from now on: in this case
\begin{equation}
	\|F'_\rad\|^2_{L^2((0,L),q(t)\dd t)}=|c|^2\Big(1-\dfrac{\ell}{\alpha p}\Big)^{-1}.
	\label{frad1}
\end{equation}
To compute $\|F_\rad\|^2_{L^2((0,L),q(t)\dd t)}$ we first represent, for any $n\in \N_0$ and $t\in (t_{n-1},t_n)$,
\begin{align*}
|F_\rad(t)|^2&=|c|^2\bigg(
\dfrac{(t-t_{n-1})^2}{(\alpha p)^{2n}}
+2 \dfrac{t-t_{n-1}}{(\alpha p)^n}
\dfrac{1-\big(\frac{\ell}{\alpha p}\big)^n}{1-\frac{\ell}{\alpha p}}
+\Big(\dfrac{1-\big(\frac{\ell}{\alpha p}\big)^n}{1-\frac{\ell}{\alpha p}}\Big)^2
\bigg),
\end{align*}
therefore,
\[
\int_{t_{n-1}}^{t_n} \big| F_\rad(t)\big|^2 \dd t=|c|^2\bigg(
\dfrac{\ell^{3n}}{3 (\alpha p)^{2n}}
+
\dfrac{\ell^{2n}}{(\alpha p)^n}
\dfrac{1-\big(\frac{\ell}{\alpha p}\big)^n}{1-\frac{\ell}{\alpha p}}\Big)
+
\ell^n \Big(\dfrac{1-\big(\frac{\ell}{\alpha p}\big)^n}{1-\frac{\ell}{\alpha p}}\Big)^2
\bigg)
\]
and
\begin{align*}
	\|F_\rad&\|^2_{L^2((0,L),q(t)\dd t)}=\int_0^L \big|F_\rad(t)\big|^2 q(t)\dd t=\sum_{n=0}^{\infty} (\alpha p)^n\int_{t_{n-1}}^{t_n}\big|F_\rad(t)\big|^2\dd t\\
&=|c|^2\bigg(
\dfrac{1}{3}\sum_{n=0}^\infty \Big(\dfrac{\ell^3}{\alpha p}\Big)^n
+\sum_{n=0}^\infty \ell^{2n} \dfrac{1-\big(\frac{\ell}{\alpha p}\big)^n}{1-\frac{\mathstrut\ell}{\alpha p}}\Big)+ \sum_{n=0}^\infty (\alpha p\ell)^n \Big(\dfrac{1-\big(\frac{\ell}{\alpha p}\big)^n}{1-\frac{\mathstrut\ell}{\alpha p}}\Big)^2
\bigg),
\end{align*} 
which is finite if and only if $\ell^3<\alpha p$ and $\alpha p\ell<1$ (we recall that $\ell<1$ by the initial assumption). For subsequent computations it will be useful to normalize $F_\rad$ in $\Tilde H^1\big((0,L),q(t)\dd t\big)$, in view of \eqref{frad1} this amounts to the choice
\[
c:=\sqrt{1-\dfrac{\ell}{\alpha p}},
\]
and gives
\begin{equation}
	\label{frad2}
\begin{aligned}
F_\rad (t)&=\sqrt{1-\dfrac{\ell}{\alpha p}}\Big(\dfrac{t-t_{n-1}}{(\alpha p)^n} + \sum_{k=0}^{n-1} \big(\frac{\ell}{\alpha p}\big)^k\Big),\quad t\in(t_{n-1},t_n),
\quad n\in \N_0.
\end{aligned}
\end{equation}
We summarize these computation as
\begin{gather*}
\Tilde H^1_\rad(T)\ne \Tilde H^1_{0|\rad}(T) \text{ if and only if }
\ell<\alpha p<\dfrac{1}{\ell},\\
\Tilde H^1_{\Delta|\rad}(\T)=\begin{cases}
\C U_\rad F_\rad,& \ell<\alpha p<\dfrac{1}{\ell},\\
0, &\text{otherwise.}
\end{cases}
\end{gather*}

The spaces $\Tilde H^1_{\Delta|n,k,s}$ are studied in the same way. Let us fix an admissible triple $(n,k,s)$
and recall that $U_{n,k,s}C^\infty_c(t_n,L)$ is dense in $\Tilde H^1_{0|n,k,s}(\T)$, then
a function $f_{n,k,s}\in\Tilde H^1_{n,k,s}(\T)$ belongs to $\Tilde H^1_{\Delta|n,k,s}(\T)$
if and only if $f_{n,k,s}=U_{n,k,s}F_{n,k,s}$ with $F\in \Tilde H^1\big((t_n,L),p^{-n}q(t)\dd t\big)$ satisfying
\[
\langle (U_{n,k,s}F_{n,k,s})',(U_{n,k,s}G)'\rangle_{L^2(\T)}=0
\text{ for all $G\in C^\infty_c(t_n,L)$,}
\]
which can be equivalently rewritten as
\begin{align*}
\langle U_{n,k,s}(F_{n,k,s}'),U_{n,k,s}(G')\rangle_{L^2(\T)}&\equiv \langle F'_{n,k,s},G'\rangle_{L^2((t_n,L),p^{-n}q(t)\dd t)}\\
&\equiv p^{-n}\int_{t_n}^L F_{n,k,s}'(t) \overline{G'(t)}q(t)\dd t=0
\end{align*}
for all $G\in C^\infty_c(t_n,L)$, which means $(qF_{n,k,s}')'=0$ in $(t_n,L)$.
As $F_{n,k,s}$ can be chosen independent of $(k,s)$ we write simply
\[
F_n:=F_{n,k,s}.
\]
So we obtain $F'_n=c/q$ in $(t_n,L)$ with some $c\in \C$, and for $t\in (t_{n+m-1},t_{n+m})$ with $m\in\N$ one has
\begin{align*}
	F_n'(t)&= \dfrac{c}{(\alpha p)^{n+m}},\\
	F_n(t)&=c\Big(\dfrac{t-t_{n+m-1}}{(\alpha p)^{n+m}} + \sum_{k=1}^{m-1} \big(\frac{\ell}{\alpha p}\big)^{n+k}\Big)\\
	&\equiv c\Big(\dfrac{t-t_{n+m-1}}{(\alpha p)^{n+m}} + \Big(\dfrac{\ell}{\alpha p}\Big)^{n+1}\dfrac{1-\big(\frac{\ell}{\alpha p}\big)^{m-1}}{1-\frac{\ell}{\alpha p}}\Big).
\end{align*}
In order to check if $F_n\in \Tilde H^1\big((t_n,L),p^{-n}q(t)\dd t\big)$ we again compute
\begin{align*}
	\|F'_n&\|^2_{L^2((t_n,L),p^{-n}q(t)\dd t)}=p^{-n}\sum_{m=1}^\infty \int_{t_{n+m-1}}^{t_{n+m}} |F'_n(t)|^2 (\alpha p)^{n+m}\dd t\\
	&=|c|^2 p^{-n}\sum_{m=1}^\infty \ell^{n+m} (\alpha p)^{-2(n+m)} (\alpha p)^{n+m}
	=|c|^2 p^{-n}\sum_{m=1}^\infty \Big(\dfrac{\ell}{\alpha p}\Big)^{n+m},
\end{align*}
which is finite if and only if $\ell<\alpha p$, and in that case
\begin{equation}
	\label{fnks1}
\|F'_n\|^2_{L^2((t_n,L),p^{-n}q(t)\dd t)}=\dfrac{|c|^2}{p^n}\Big(\dfrac{\ell}{\alpha p}\Big)^{n+1}\dfrac{1}{1-\frac{\mathstrut\ell}{\alpha p}}.
\end{equation}
Furthermore, in this case for $t\in (t_{n+m-1},t_{n+m})$ with $m\in\N$ one has
\begin{align*}
	|F_n(t)|^2&=|c|^2\bigg(
	\dfrac{(t-t_{n+m-1})^2}{(\alpha p)^{2n+2m}}+2(t-t_{n+m-1})\dfrac{\ell^{n+1}}{(\alpha p)^{2n+m+1}}\dfrac{1-\big(\frac{\ell}{\alpha p}\big)^m}{1-\frac{\ell}{\alpha p}}
	\\
	& \qquad +\Big( \dfrac{\ell}{\alpha p}\Big)^{2n+2}\Big(\dfrac{1-\big(\frac{\ell}{\alpha p}\big)^m}{1-\frac{\ell}{\alpha p}}\Big)^2\bigg),
\end{align*}
therefore,
\begin{align*}
	\int_{t_{n+m-1}}^{t_{n+m}} \big| F_n(t)\big|^2 \dd t&=|c|^2\bigg[
	\dfrac{1}{3}\dfrac{\ell^{3m+3n}}{(\alpha p)^{2n+2m}}
	+\ell^{2n+2m} \dfrac{\ell^{n+1}}{(\alpha p)^{2n+m+1}}\dfrac{1-\big(\frac{\ell}{\alpha p}\big)^{m-1}}{1-\frac{\ell}{\alpha p}}\\
	& \qquad+\ell^{m+n}\Big( \dfrac{\ell}{\alpha p}\Big)^{2n+2}\Big(\dfrac{1-\big(\frac{\ell}{\alpha p}\big)^{m-1}}{1-\frac{\ell}{\alpha p}}\Big)^2\bigg],
\end{align*}
and 
\begin{align*}
	\|F_n&\|^2_{L^2((t_n,L),p^{-n}q(t)\dd t)}=p^{-n}\sum_{m=1}^\infty (\alpha p)^{n+m}\int_{t_{n+m-1}}^{t_{n+m}} \big| F_n(t)\big|^2 \dd t\\
	&=|c|^2\bigg[
	\dfrac{1}{3} \Big(\dfrac{\ell^3}{\alpha p}\Big)^{n+m}
	+\dfrac{\ell^{3n+2m+1}}{(\alpha p)^{n+1}}	\dfrac{1-\big(\frac{\ell}{\alpha p}\big)^{m-1}}{1-\frac{\ell}{\alpha p}}\\
	&\qquad\qquad+(\alpha p\ell)^{n+m} \Big( \dfrac{\ell}{\alpha p}\Big)^{2n+2}\Big(\dfrac{1-\big(\frac{\ell}{\alpha p}\big)^{m-1}}{1-\frac{\ell}{\alpha p}}\Big)^2
	\bigg]
\end{align*}
is finite if and only if $\ell^3<\alpha p$ and $\alpha p\ell<1$ (as one always has $\ell<1$).
For what follows we normalize $F_{n,k,s}$ to have unit norm in $\Tilde H^1\big((t_n,L),p^{-n}q(t)\dd t\big)$: as follows
from \eqref{fnks1}, this means the choice
\[
c:=p^n \Big( \dfrac{\alpha}{\ell}\Big)^\frac{n}{2}\sqrt{\dfrac{\alpha p-\ell}{\ell}}
\]
and then
\begin{equation}
	\label{fnks2}
	\begin{aligned}
	F_n(t)&=p^n \Big( \dfrac{\alpha}{\ell}\Big)^\frac{n}{2}\sqrt{\dfrac{\alpha p-\ell}{\ell}}
	\Big(\dfrac{t-t_{n+m-1}}{(\alpha p)^{n+m}} + \sum_{k=1}^{m-1} \big(\frac{\ell}{\alpha p}\big)^{n+k}\Big),\\
	& \quad\text{$t\in (t_{n+m-1},t_{n+m})$ with $m\in\N$.}
\end{aligned}
\end{equation}
Therefore,
\begin{gather*}
	\Tilde H^1_{n,k,s}(T)\ne \Tilde H^1_{0|n,k,s}(T) \text{ if and only if }
	\ell<\alpha p<\dfrac{1}{\ell},\\
	\Tilde H^1_{\Delta|n,k,s}(\T)=\begin{cases}
		\C U_{n,k,s} F_n,& \ell<\alpha p<\dfrac{1}{\ell},\\
		0, &\text{otherwise.}
	\end{cases}
\end{gather*}

We summarize the preceding computations:
\begin{lemma}\label{onb-delta}
One has $\Tilde H^1_\Delta(\T)\ne \{0\}$ or, equivalently,
 $\Tilde H^1_0(T)\ne \Tilde H^1(\T)$, if and only if
\begin{equation}
	\label{lap}
	\ell<\alpha p<\dfrac{1}{\ell}.
\end{equation}	
If this condition is satisfied, then the functions
\[
\phi_\rad:=U_\rad F_\rad \text{ and } \phi_{n,k,s}:=U_{n,k,s}F_n \text{ with $e_{n,k}\in E$ and $s\in\{1,\dots,p-1\}$}
\]
with $F_\rad$ from \eqref{frad2} and $F_n$ from \eqref{fnks2}
form an orthonormal basis in $\Tilde H^1_\Delta(\T)$.
\end{lemma}	

We remark that the criterion \eqref{lap} for $\Tilde H^1_0(T)\ne \Tilde H^1(\T)$ and, equivalently, for $H^1_0(T)\ne  H^1(\T)$, was already obtained in \cite{jks} (even for a more general configuration) by different methods, but in our case it appears naturally as a by-result of the construction of the orthonormal basis.

For subsequent constructions we will need the limits
\begin{equation}
	\label{finf}
F^\infty_\rad:=\lim_{t\to L^-} F_\rad(t), 
\quad
F^\infty_n:=\lim_{t\to L^-} F_n(t).
\end{equation}
As $F_\rad$ and $F_n$ are increasing functions, we have
\begin{equation}
	\label{finf1}
	\begin{aligned}
	F^\infty_\rad&=\lim_{n\to \infty} F_\rad(t_n)=\lim_{n\to \infty} \sqrt{1-\dfrac{\ell}{\alpha p}}\sum_{k=0}^{n} \big(\frac{\ell}{\alpha p}\big)^k\\
	&= \sqrt{1-\dfrac{\ell}{\alpha p}}\lim_{n\to \infty} \dfrac{1-\big( \frac{\ell}{\alpha p}\big)^{n+1}}{1-\frac{\ell}{\alpha p}}= \sqrt{1-\dfrac{\ell}{\alpha p}} \dfrac{1}{{1-\frac{\ell}{\alpha p}}}
	=\sqrt{\dfrac{\alpha p}{\alpha p-\ell}}
\end{aligned}
\end{equation}
and, similarly,
\begin{align*}
	F^\infty_n&=\lim_{m\to \infty} F_n(t_{n+m})=\lim_{m\to \infty} p^n \Big( \dfrac{\alpha}{\ell}\Big)^\frac{n}{2}\sqrt{\dfrac{\alpha p-\ell}{\ell}}\sum_{k=1}^m \Big(\dfrac{\ell}{\alpha p}\Big)^{n+k}.
\end{align*}
Due to
\begin{align*}
\sum_{k=1}^m \Big(\dfrac{\ell}{\alpha p}\Big)^{n+k}&=\Big(\dfrac{\ell}{\alpha p}\Big)^{n+1} \sum_{k=0}^{m-1} \Big(\dfrac{\ell}{\alpha p}\Big)^{k}\\
&=\Big(\dfrac{\ell}{\alpha p}\Big)^{n+1}\dfrac{1-\Big(\dfrac{\ell}{\alpha p}\Big)^m}{1-\dfrac{\ell}{\alpha p}}
\xrightarrow{m\to\infty}
\Big(\dfrac{\ell}{\alpha p}\Big)^{n+1}\,\dfrac{\alpha p}{\alpha p-\ell}
\end{align*}
we obtain
\begin{equation}
	\label{finf2}
	F^\infty_n=p^n\Big(\dfrac{\alpha}{\ell}\Big)^\frac{n}{2}\sqrt{\dfrac{\alpha p-\ell}{\ell}}
	\Big(\dfrac{\ell}{\alpha p}\Big)^{n+1}\,\dfrac{\alpha p}{\alpha p-\ell}
	\equiv
	\sqrt{\dfrac{\ell}{\alpha p-\ell}\Big(\dfrac{\ell}{\alpha}\Big)^n}.
\end{equation}

\subsection{Abstract trace operator}\label{sec-abstrace}

Everywhere in this subsection we assume that
\[
\fbox{\text{ the inequalities \eqref{lap} are satisfied.}}
\]

Now we construct the first version of  the trace operator. We define
\[
\cZ:=\{\rad\}\cup\big\{(n,k,s):\, e_{n,k}\in E,\, s\in\{1,\dots,p-1\} \big\}
\]
and define
\[
\nu:\cZ\to \{\rad\}\cup \N_0,
\quad
\nu(\rad):=\rad,\quad \nu(n,k,s):=n.
\]
For numerical operations it will be useful to identify
\[
\rad:=-1.
\]

Recall that we have the orthogonal decomposition $\Tilde H^1(\T)=\Tilde H^1_0(\T)\oplus \Tilde H^1_\Delta(\T)$ and let $\Tilde P_\Delta:\Tilde H^1(\T)\to \Tilde H^1_\Delta(\T)$
be the orthogonal projector. By Lemma \ref{onb-delta}, the map
\[
\Theta:\ell^2(\cZ)\ni (a_z)\mapsto \sum_{z\in\cZ}a_z U_z F_z\in H^1_\Delta (\T)
\]
is unitary. Remark that the behavior of $U_zF_z$ near the boundary of $\T$
is uniquely quantified by $z$ and the limiting values $F^\infty_{\nu(z)}$ defined in \eqref{finf}.
Therefore, it is reasonable to consider the multiplication operator
\[
M:\ell^2(\cZ)\to \ell^2(\cZ), \quad (a_z)\mapsto \big(p^{-\frac{\nu(z)}{2}}F^\infty_{\nu(z)} a_z\big).
\]
The explicit expressions \eqref{finf1} and \eqref{finf2}
show that the coefficients $p^{-\frac{\nu(z)}{2}} F^\infty_{\nu(z)}$ are strictly positive and uniformly bounded, therefore, the operator $M$ is injective and bounded. However, it is not surjective.
Namely, for $r\ge 0$ introduce
\begin{equation}
	\label{l2r}
\begin{aligned}	
\ell^2_r(\cZ)&:=\big\{
(a_z)\in\ell^2(\cZ):\ (p^{r\nu(z)} a_z)\in \ell^2(\cZ)\big\},\\
\|(a_z)\|_{\ell^2_r(\cZ)}&:=\big\|(p^{r\nu(z)} a_z)\|_{\ell^2(\cZ)}.
\end{aligned}
\end{equation}

Denote
\begin{equation*}
\sigma:=\dfrac{1}{2\log p}\log\dfrac{\alpha p}{\ell}\equiv \frac{1}{2}\Big ( 1-\frac{\log\ell-\log\alpha}{\log p}\Big)>0,
\end{equation*}
then
\[
F^\infty_{\nu(z)}=\sqrt{\dfrac{\ell}{\alpha p-\ell}}\,p^{-\sigma\nu(z)} \text{ for all }z\in\cZ.
\]
and it follows that $M:\ell^2(\cZ)\to \ell^2_\sigma(\cZ)$ is a bounded bijective operator. This gives the abstract trace operator
\begin{equation}
	\label{eq-tau1}
\tau:=M\Theta^{-1}\Tilde P_\Delta:\Tilde H^1(\T)\to \ell^2_\sigma(\cZ),
\end{equation}
which is a linear operator which is bounded and surjective by construction. Remark that $M\Theta^{-1}$ is injective, which shows that
$\ker\tau=\ker \Tilde P_\Delta=\Tilde H^1_0(\T)$.
In addition we extend $\tau$ to $H^1(\T)$: for any $f\in H^1(\T)$ set
\begin{equation}
	\label{tauff}
\tau f:=\tau \Tilde f
\end{equation}
with any $\Tilde f\in \Tilde H^1(\T)$ such that $f=\Tilde f$ in $\T\setminus\T^n$ for some $n\in\N_0$.

\begin{theorem}\label{th28}
	The abstract trace operator $\tau: H^1(\T)\to \ell^2_\sigma(\cZ)$ is well defined, linear, bounded, surjective, with $\ker\tau=H^1_0(\T)$.
\end{theorem}

\begin{proof}
	Let $f\in H^1(\T)$ and $\Tilde f,\Tilde g\in \Tilde H^1(\T)$ such that $f=\Tilde f$
	in $\T\setminus\T^n$ for some $n\in\N_0$ and $f=\Tilde g$ in $\T\setminus\T^m$ for some $m\in\N_0$. Without loss of generality assume $n\le m$, then $\Tilde f=\Tilde g$ in $\T\setminus\T^m$, i.e. $\Tilde f-\Tilde g=0$ in $\T\setminus\T^m$.
	This means that $\tau(\Tilde f-\Tilde g)=0$, i.e. $\tau \Tilde f=\tau \Tilde g$. This shows that $\tau$ is a well defined map.

	Let $\varphi:(0,L)\to \R$	be a $C^\infty$ function such that $\varphi=0$ in $(0,\frac{1}{2})$ and $\varphi=1$ in $(\frac{3}{4},L)$. For $f\in H^1(\T)$ the function $\Tilde f: \T\ni x\mapsto \varphi(|x|)f(x)$ belongs to $\Tilde H^1(\T)$ and coincides with $f$ in $\T\setminus\T^0$, so one has $\tau f=\tau\Tilde f$. As $f\mapsto \Tilde f$ is linear, it follows that the extended $\tau$ is also linear.
	
	To show the boundedness, it is sufficient to show the boundedness of the map $H^1(\T)\ni f\mapsto \Tilde f\in \Tilde H^1(\T)$. For any $f\in H^1(\T)$ one has the identities
\[
	\Tilde f_{0,0}=\varphi f, \quad \Tilde f_{n,k}=f_{n,k} \text{ for } (n,k)\ne (0,0).
\]
This gives
\begin{align*}
\|\Tilde f\|^2_{\Tilde H^1(\T)}&=\sum_{e_{n,k}\in E} (\alpha p)^n \|\Tilde f'_{n,k}\|^2_{L^2(t_{n-1},t_n)}\\
	&=\|\Tilde f'_{0,0}\|^2_{L^2(0,1)}+\sum_{\substack{e_{n,k}\in E\\(n,k)\ne (0,0)}} (\alpha p)^n \|f'_{n,k}\|^2_{L^2(t_{n-1},t_n)}\\
	&\le \|\Tilde f'_{0,0}\|^2_{L^2(0,1)}+\|f\|^2_{H^1(\T)},
\end{align*}
and
\begin{align*}
\|\Tilde f'_{0,0}\|^2_{L^2(0,1)}&=\|\varphi f'_{0,0}+\varphi' f_{0,0}\|^2_{L^2(0,1)}\le 2\|\varphi f'_{0,0}\|^2_{L^2(0,1)} +2 \|\varphi' f_{0,0}\|^2_{L^2(0,1)}\\
&\le 2\|\varphi\|_\infty^2 \|f'_{0,0}\|^2_{L^2(0,1)} + 2\|\varphi'\|_\infty^2 \|f_{0,0}\|^2_{L^2(0,1)}\\
&\le 2b \|f_{0,0}\|^2_{H^1(0,1)} \text{ with } b:=\max\{\|\varphi\|_\infty^2, \|\varphi'\|_\infty^2\},
\end{align*}
so one obtains $\|\Tilde f\|^2_{\Tilde H^1(\T)}\le 2b\|f_{0,0}\|^2_{H^1(0,1)}+\|f\|^2_{H^1(\T)}
\le (2b+1)\|f\|^2_{H^1(\T)}$, which gives the result.

The map $H^1(\T)\ni f\mapsto \Tilde f\in \Tilde H^1(\T)$ is surjective, which shows that the range of $\tau$ is the same as before.

If $f\in H^1_0(\T)$, then there exist $f_n\in H^1_c(\T)$ with $f_n\xrightarrow{H^1(\T)}f$. Then $\Tilde f_n\in \Tilde H^1_c(\T)$ with $\Tilde f_n\xrightarrow{H^1(\T)} \Tilde f$, so $\Tilde f\in \Tilde H^1_0(\T)$ and $\tau f=\tau \Tilde f=0$.

On the other hand, if $\tau f=0$, then $\tau\Tilde f=0$ and $\Tilde f\in \Tilde H^1_0(\T)$.
Then there exist $g_n\in \Tilde H^1_c(\T)$ with $g_n\xrightarrow{H^1(\T)}\Tilde f$.
The function $g:\T\ni x\mapsto \big(1-\varphi(|x|)\big)f(x)$ is supported in $\T^0$ and, hence, it belongs to $H^1_c(\T)$. Therefore, $g+g_n\in H^1_c(\T)$ with $g+g_n\xrightarrow{H^1(\T)}g+\Tilde f=f$, so $f\in H^1_c(\T)$.
\end{proof}

The above definition of $\tau$ is involved due to the application of the orthogonal projector and the expansion into an orthonormal basis. Let us show that it can be recovered using more elementary operations.

\begin{lemma}\label{lem39}
	One has continuous embeddings
	\begin{align*}
		\Tilde H^1\big( (0,L),q(t)\dd t\big)&\hookrightarrow C^0([0,L]),\\
		\Tilde H^1\big((t_n,L),p^{-n}q(t)\dd t\big)&\hookrightarrow C^0([t_n,L]) \text{ for any $n\in\N_0$,}
	\end{align*}
	where the right-hand sides are endowed with $\|\cdot\|_\infty$.
\end{lemma}

\begin{proof}
	The continuity inside the respective intervals is clear due to the one-dimensional Sobolev theorem, and it remains to establish norm estimates. Let $f\in \Tilde H^1\big( (0,L),q(t)\dd t\big)$ and $t\in(0,L)$, then
	\begin{align*}
		\big|f(t)\big|^2&=\Big|\int_0^t f'(s) \dd s\Big|^2\le \Big(\int_0^L \big|f'(s)\big|\dd s\Big)^2\\
		&= \Big(\int_0^L \dfrac{1}{\sqrt{q(s)}}\cdot\big|f'(s)\big|\sqrt{q(s)}\dd s\Big)^2
		\le \int_0^L \dfrac{\dd s}{q(s)}\cdot \int_0^L \big|f'(s)\big|^2 q(s)\dd s.
	\end{align*}
	The second factor on the right-hand side is $\|f\|^2_{\Tilde H^1((0,L),q(t)\dd t)}$, while
	\[
	\int_0^L \dfrac{\dd s}{q(s)}=\sum_{n=0}^\infty \int_{t_{n-1}}^{t_n} \dfrac{\dd s}{(\alpha p)^n}
	=\sum_{n=0}^\infty \Big( \dfrac{\ell}{\alpha p}\Big)^n=:a<\infty.
	\]
	As $t\in (0,L)$ was arbitrary, this yields $\|f\|_\infty\le \sqrt{a} \|f\|_{\Tilde H^1((0,L),q(t)\dd t)}$.
	
	If $f\in \Tilde H^1\big((t_n,L),p^{-n}q(t)\dd t\big)$ with some $n\in\N_0$, then its extension to $(0,L)$ by zero $\Tilde f$ belongs to $f\in \Tilde H^1\big( (0,L),q(t)\dd t\big)$, and 
	one uses the first part of the proof.
\end{proof}

\begin{lemma}\label{prop-unique}
	Let $\tau': H^1(\T)\to \ell^2(\cZ)$ be a bounded linear map such that
	\begin{itemize}
		\item[(a)] $\tau' f=0$ for any $f\in H^1_c(\T)$,
		\item[(b)] for any $F\in  H^1\big((0,L),q(t)\dd t\big)$ one has
		\[
		(\tau' U_\rad F)_z=\begin{cases}
			\lim_{t\to L^-}F(t), & z=\rad,\\
			0, &\text{otherwise,}
		\end{cases}
		\]
		\item[(c)] for any $\lambda\in \cZ\setminus\{\rad\}$ and any $F\in \Tilde H^1\big( (t_{\nu(\lambda)},L),p^{-\nu(\lambda)}q(t)\dd t\big)$ one has
		\[
		(\tau' U_\lambda F)_z=p^{-\frac{\nu(\lambda)}{2}}\begin{cases}
			\lim_{t\to L^-}F(t), & z=\lambda,\\
			0, &\text{otherwise,}
		\end{cases}
		\]
	\end{itemize}
	then $\tau'=\tau$. Moreover, these properties are satisfied by $\tau$.
\end{lemma}

\begin{proof}
We first remark that the limits on the right-hand sides of (b) and (c) are well-defined by Lemma \ref{lem39}; for (b) one uses the fact that $F$ coincides with some function $\Tilde F\in \Tilde H^1\big( (0,L),q(t)\dd t\big)$ in $(\frac{1}{2},L)$.

(i) The boundedness of $\tau'$ and the condition (a) give $H^1_0(\T)\subset \ker \tau'$.
For any $\lambda,z\in\cZ$ we have
	\begin{align*}
		(\tau'\phi_{\lambda})_z = (\tau'U_{\lambda}F_{\nu(\lambda)})_z
		&= \begin{cases}
			\lim_{t\to L^-}F_{\nu(\lambda)}(t)p^{-\frac{\nu(z)}{2}}, & z=\lambda,\\
			0, &\text{otherwise.}
		\end{cases}\\
		& = \begin{cases}
			F_{\nu(\lambda)}^{\infty}p^{-\frac{\nu(z)}{2}}, & z=\lambda,\\
			0, &\text{otherwise}.
		\end{cases}
	\end{align*} 
	On the other side, by definition we have $(\Theta^{-1}\phi_{\lambda})_z=\delta_{\lambda,z}$ (where $\delta_{\lambda,z}$ are the usual Kronecker symbols)
	and, therefore,
	\begin{equation*}
		(\tau \phi_{\lambda})_z = \begin{cases}
			p^{-\frac{\nu(\lambda)}{2}}F_{\nu(\lambda)}^{\infty}, & z=\lambda,\\
			0, &\text{otherwise},
		\end{cases}
	\end{equation*} 
which coincides with $\tau' \phi_\lambda$. As the linear span of $\Tilde H^1_0(\T)$ and $(\phi_\lambda)_{\lambda\in\cZ}$ is dense in $\Tilde H^1(\T)$ and $\tau'$ is bounded,
it follows that $\tau'=\tau$ on $\Tilde H^1(\T)$.

(ii) Let $\varphi:(0,L)\to \R$	be a $C^\infty$ function such that $\varphi=0$ in $(0,\frac{1}{2})$ and $\varphi=1$ in $(\frac{3}{4},L)$. For $f\in H^1(\T)$ consider $\Tilde f: \T\ni x\mapsto \varphi(|x|)f(x)\in \Tilde H^1(\T)$. As $\Tilde f- f=0$ in $\T\setminus\T^0$, we have $\Tilde f-f\in H^1_c(\T)$, and
\[
\tau' f\stackrel{\text{(a)}}{=}\tau'\Tilde f\stackrel{\text{(i)}}{=}\tau \Tilde f\stackrel{\eqref{tauff}}{=}\tau f. 
\]

(iii) It remains to check that $\tau$ satisfies the three properties. Remark that (a) holds by construction. Let $F$ be as in (b), then one can represent $F:=cF_\rad +G$
with $c\in\C$ and $\lim_{t\to L^-}G(t)=0$. By linearity we have
\begin{align*}
(\tau U_\rad F)_z&=\tau(c U_\rad F_\rad +U_\rad G)\\
&=cF_\rad^\infty \delta_{\rad,z} + \tau U_\rad G\equiv \big(\lim_{t\to L^-} F(t)\big)\delta_{\rad,z} + \tau U_\rad G.
\end{align*}
Therefore, to show (b) it is sufficient to show that $\tau U_\rad G=0$. For that
we take $\varphi\in C^\infty(\R)$ such that $0\le \varphi\le 1$, with
$\varphi(t)=1$ for $t\le 0$ and $\varphi(t)=0$ for $t\ge 1$, and for
$n\in\N$ consider the functions
\[
\varphi_n:t\mapsto \varphi\Big( \dfrac{t-t_{n-1}}{\ell^n}\Big),
\quad
G_n:=\varphi_n G \in H^1_c\big((0,L),q(t)\dd t\big).
\]
As for any $n$ there holds $\tau U_\rad G_n=0$, it is sufficient to show that
$G_n\to G$ in $H^1\big((0,L),q(t)\dd t\big)$ as $n\to\infty$

The dominated convergence implies $G_n\to G$ in $L^2\big((0,L),q(t)\dd t\big)$ as $n\to\infty$.
We have $G_n'=\varphi_n' G + \varphi_n G'$ and the second summand converges
to $G'$ in $L^2\big((0,L),q(t)\dd t\big)$  as $n\to\infty$. It remains to check
$\varphi_n' G\to 0$ in $L^2\big((0,L),q(t)\dd t\big)$ for $n\to\infty$.
We have
\[
(\varphi_n' G)(t)=\dfrac{1}{\ell^n}\varphi'\Big( \dfrac{t-t_{n-1}}{\ell^n}\Big) G(t),
\]
and the function vanishes outside $(t_{n-1},t_n)\subset (t_{n-1},L)$. If follows that
\[
\|\varphi_n' G\|^2_{L^2((0,L),q(t)\dd t)}\le \dfrac{\|\varphi'\|_\infty}{\ell^{2n}}
\int_{t_{n-1}}^L |G(t)|^2q(t)\dd t.
\]
As $G$ vanishes at $L$, for all $t\in (t_{n-1},L)$ we have
\begin{align*}
|G(t)|^2&=\Big|\int_t^L G'(s)\dd s\Big|=\Big|\int_t^L G'(s)\sqrt{q(s)} \dfrac{1}{\sqrt{q(s)}}\dd s\Big|\\
&\le \int_t^L |G'(s)|^2 q(s)\dd s\int_t^L \dfrac{\dd s}{q(s)}
\le \int_{t_{n-1}}^L |G'(s)|^2 q(s)\dd s\int_{t_{n-1}}^L \dfrac{\dd s}{q(s)},
\end{align*}
and we obtain
\begin{gather}
\|\varphi_n' G\|^2_{L^2((0,L),q(t)\dd t)}\le C_n\int_{t_{n-1}}^L |G'(t)|^2 q(s)\dd s, \label{glok}\\
C_n:=\dfrac{\|\varphi'\|_\infty}{\ell^{2n}} \int_{t_{n-1}}^L \dfrac{\dd s}{q(s)}
\int_{t_{n-1}}^L q(s)\dd s.\nonumber
\end{gather}
In order to show the sought convergence it suffices to show that $C_n$ remain bounded for large $n$. For that, using the explicit expression of $q$ and the relations \eqref{lap} we compute
\begin{align*}
\int_{t_{n-1}}^L \dfrac{\dd s}{q(s)}
\int_{t_{n-1}}^L q(s)\dd s&=\sum_{k=n-1}^\infty (\ell \alpha p)^k \cdot \sum_{k=n-1}^\infty \Big(\dfrac{\ell}{\alpha p}\Big)^k\\
&=(\ell\alpha p)^{n-1}(1-\ell\alpha p)^{-1}\Big(\dfrac{\ell}{\alpha p}\Big)^{n-1}
\big(1-\dfrac{\ell}{\alpha p}\big)^{-1}\\
&\equiv\ell^{2n-2} (1-\ell\alpha p)^{-1}\big(1-\dfrac{\ell}{\alpha p}\big)^{-1},
\end{align*}
which gives
\[
C_n=\dfrac{\|\varphi'\|_\infty}{\ell^2} (1-\ell\alpha p)^{-1}\big(1-\dfrac{\ell}{\alpha p}\big)^{-1},
\]
i.e. $C_n$ are independent of $n$. This concludes the proof of (b) for $\tau$,
and the property (c) is proved in the same way.
\end{proof}

We complement the preceding observations by the following approximation result, which will be useful for the geometric interpretation of the embedded trace:
\begin{lemma}\label{proplim}
For $f\in H^1(\T)$	and $N\in \N$ let $f_N$ be the extension of $f|_{\T^N}$ by constants, i.e.
\[
f_N(x):= \begin{cases}
	f(x), & x\in \T^N,\\
	f(X_{N,K}), & x \in \T_{N,K}, \ K\in\{0,\dots,p^N-1\}.
\end{cases}
\]	
then $f_N\xrightarrow{N\to \infty}f$ in $H^1(\T)$, in particular, $\tau f=\lim_{N\to \infty}\tau f_N$ in $\ell^2_\sigma (\cZ)$.
\end{lemma}

\begin{proof}
Due to the one-dimensional Sobolev inequality for any $N\in\N$ one can find some $B_N>0$ such that $|f(X_{N,K})|\le B_N\|f\|_{H^1(\T)}$ for any $f\in H^1(\T)$ and any $K\in\{0,\dots,p^N-1\}$. We have
\[
\|f_N\|^2_{L^2(\T)}=\|f\|^2_{L^2(\T^N)} 
		+\sum_{K=0}^{p^N-1} \big|f(X_{N,K})\big|^2\int_{\T_{N,K}}1\dd\mu,
\]
while
\begin{align*}
	\int_{\T_{N,K}}1\dd\mu&=\sum_{j=0}^{p-1}\sum_{n=N+1}^\infty\sum_{k=0}^{p^n-1}(\alpha \ell)^n
	\equiv p \sum_{n=N+1}^\infty(p\alpha \ell)^n=:s_N\stackrel{\eqref{lap}}{<}\infty.
\end{align*}
It follows that
\[
\|f_N\|^2_{L^2(\T)}\le\|f\|^2_{L^2(\T^N)}+ p^{N+1} B_N^2 s_N \|f\|^2_{H^1(\T)}<\infty,
\]
i.e. $f_N\in L^2(\T)$ for any $N$. At the same time,
\[
(f'_N)_{n,k}=\begin{cases}
	f'_{n,k},& n\le N,\\
	0,& \text{otherwise.}
\end{cases}
\]
In particular, $|f'_N|\le |f'|$, which yields $\|f'_N\|_{L^2(\T)}\le \|f'\|_{L^2(\T)}$
and $f_N\in H^1(\T)$. In addition,
\begin{equation}
	\label{fnf1}
	\|f'_N-f'\|_{L^2(\T)}=\|f'\|_{L^2(\T\setminus\T^N)}\xrightarrow{N\to\infty}0.
\end{equation}
Now let $C>0$ be the constant from the Poincar\'e inequality \eqref{poincare}. By construction
one has $f_N-f\in \Tilde H^1(\T)$, therefore, due to \eqref{fnf1},
\begin{align*}
	\|f_N-f\|_{L^2(\T)}\le C \|f'_N-f'\|_{L^2(\T)}\xrightarrow{N\to\infty}0,
\end{align*}
which concludes the proof.
\end{proof}

\section{Approximation spaces}\label{secappr}

\subsection{Excursus about Sobolev spaces}\label{sec31}

Let us briefly recall various definitions and basic facts related to fractional Sobolev spaces $H^s$ with $s\in(0,1)$ on open sets and manifolds.

Let $\Omega\subset \R^d$ be a non-empty open subset and $k\in \N_0$, then the $k$th Sobolev space $H^k(\Omega)$ is defined as
\begin{align*}
	H^k(\Omega)=\big\{f\in L^2(\Omega):\ \partial^\alpha f \in L^2(\Omega) \text{ for all $\alpha\in\N_0^d$ with $|\alpha|\le k$} \big\},
\end{align*}
which is a Hilbert space if equipped with the scalar product
\begin{align*}
	\scalar{f,g}_{H^k(\Omega)}=\sum_{|\alpha|\leq k} \scalar{\partial^{\alpha}f,\partial^{\alpha}g}_{L^2(\Omega)}.
\end{align*}

For $\Omega=\R^d$ we obtain an equivalent definition via the Fourier transform. Namely for $s\in [0,\infty)$ define
the $s$th Sobolev space by
	\begin{align*}
		H^s(\R^d) = \big\{ f\in L^2(\R^d):\ \scalar{\xi}^s\hat{f}\in L^2(\R^d)\big\}  \;\text{ with }\scalar{\xi}:= \sqrt{1+|\xi|^2},
	\end{align*} 
where $\hat f$ is the Fourier tranform of $f$, which becomes a Hilbert space if equipped with the norm
\begin{align*}
	\scalar{f,g}_{H^s(\R^d),\wedge} = \big\langle \scalar{\xi}^s\hat{f}, \scalar{\xi}^s\hat{g}\big\rangle_{L^2(\R^d)}.
\end{align*}
For $s\in\N_0$ the two above definitions of $H^s(\R^d)$ coincide and the two norms are equivalent.

Let $\Omega\subset \R^d$ be a bounded non-empty open subset and $s\in (0,1)$, then the $s$th Sobolev space on $H^s(\Omega)$ is defined as the space of the restrictions on $\Omega$ of the functions from $H^s(\R^d)$ with the quotient norm
\[
\|f\|_{H^s(\Omega),*}=\inf_{g\in H^s(\R^d),\  g|_\Omega=f}\|g\|_{H^s(\R^d)}.
\]
We will need various equivalent characterizations of these spaces as well as several equivalent norms \cite{adams}. Recall that $\Omega$ is said to be with Lipschitz boundary if for any $p\in \partial\Omega$ there exist
$\eps>0$, $a>0$, a Lipschitz function $h$
defined on the open ball $B_\eps(0)\subset\R^{d-1}$ with $h(0)=0$ and $|h(y_1,\dots,y_{d-1})|<a$ for all $(y_1,\dots,y_{d-1})\in B_\eps(0)$,
and Cartesian coordinates $(y_1,\dots,y_m)$ centered at $p$  such that
\begin{multline*}
\Omega\cap \big\{ (y_1,\dots,y_d):\, (y_1,\dots,y_{d-1})\in B_\eps(0),\, |y_d|<2a\big\}\\
=\big\{(y_1,\dots,y_d):\, (y_1,\dots,y_{d-1})\in B_\eps(0),\, h(y_1,\dots,y_{d-1})<y_d<2a\big\}.
\end{multline*}

The first reformulation comes from the interpolation theory, see \cite[Chapter 34]{tartar}.
Let $X$ and $Y$ be normalized spaces with $X\subset Y$ and $s\in (0,1)$.
Choose any $b>1$  and for $f\in Y$ and $t>0$ define
\begin{equation}
	\label{eq-kf}
\begin{aligned}
	K(f,t)&=\inf\limits_{g\in X} \big(\|f-g\|_Y+t\|g\|_X\big),\\
	F^s&:=(F^s)_{j\in\N} \text{ with } F^s_j:=b^{js}K(f,b^{-j}),
\end{aligned}
\end{equation}
then the interpolated space $[Y,X]_s$ is defined by
\begin{equation}
	\label{norm-yxs}
	[Y,X]_s:=\big\{f\in Y:\ \|f\|_{[Y,X]_s}^2 := \|f\|_Y^2 + \|F^s\|_{\ell^2}^2 < \infty \big\},
\end{equation}
and for any $0<s<s'\le 1$ one has
\begin{equation}
\label{Hss'}
H^s(\Omega)=\big[L^2(\Omega),H^{s'}(\Omega)\big]_{\frac{s}{s'}}
\end{equation}
with an equivalence of the associated norms, see \cite[Theorem 3.5.1]{cohen}

If $\Omega$ has Lipschitz boundary, then 
\begin{equation}
	\label{eq-fhs}
	H^s(\Omega)= \Big\{f\in L^2(\Omega):\ [f]_{H^s(\Omega)}^2:= \int_{\Omega \times \Omega} \dfrac{\big|f(x)-f(y)\big|^2}{|x-y|^{d+2s}}\dd x \dd y < \infty\Big\},
\end{equation}
while the semi-norms $f\mapsto \|F^s\|_{\ell^2}$ and $[\cdot]_{H^s(\Omega)}$ are equivalent, see \cite[Chapter 36]{tartar}.

Another group of equivalent characterizations comes from the theory of Besov spaces, which we discuss following \cite[Chapter 3]{cohen}. For $f\in L^2(\Omega)$ define its modulus of smoothness
by
\begin{equation}
	\label{wft}
	w(f,t):= \sup\limits_{|h|\leq t} \big\|f(\cdot + h)- f\big\|_{L^2(\Omega_h)}, \quad t>0,\quad  \Omega_h:=\Omega \cap (\Omega + h).
\end{equation}
For $s\in(0,1)$ one defines the Besov seminorm $[f]_{B,s}$ of $f$ by
\begin{align*}
	[f]_{B,s}^2:=\int_{0}^{1}t^{-2s}w(f,t)^2\frac{\dd t}{t},
\end{align*}
then the Besov space $B^{s}_{2,2}(\Omega)$ is defined by
\[
		B^{s}_{2,2}(\Omega):= \big\{ f\in L^2(\Omega):\ \|f\|_{B^{s}_{2,2}}^2 := \|f\|_{L^2}^2 + [f]_{B,s}^2 < \infty \big\}.
\]
Let $b>1$, then the substition $t:=b^{-x}$ gives
\[
	\int_{0}^1t^{-2s}w(f,t)^2\frac{\dd t}{t} = (\log b)\int_{0}^{\infty} b^{2sx}w(f,b^{-x})^2 \dd x,
\]
and using the monotonicity of $w(f,t)$ in $t$ we obtain the inequalities
\begin{align*}
	\frac{1}{b^2}\sum_{j=0}^{\infty} b^{2sj}w(f,p^{-j})^2 &\leq \sum_{j=0}^{\infty} b^{2sj}w(f,b^{-j-1})^2
	\leq 
	\sum_{j=0}^{\infty} \int_{j}^{j+1} b^{2sx}w(f,b^{-x})^2\dd x\\
	&\leq \sum_{j=0}^{\infty} b^{2s(j+1)}w(f,b^{-j})^2\leq b^{2s}\sum_{j=0}^{\infty} b^{2sj}w(f,b^{-j})^2,
\end{align*}
which shows that the seminorm
\begin{equation}
f\mapsto\big\|b^{sj}w(f,b^{-j})\big\|_{\ell^2} \label{tmp09}
\end{equation}
is equivalent to the above Besov seminorm. If, in addition, the set $\Omega$ has Lipschitz boundary, then
\begin{align*}
		B^s_{2,2}(\Omega) = \big[L^2(\Omega), H^1(\Omega)\big]_{s}\equiv H^s(\Omega),
\end{align*}
and the seminorms $f\mapsto \|b^{sj}w(f,b^{-j})\|_{\ell^2}$ and $[f]_{H^s(\Omega)}$ are equivalent. 

We summarize the above considerations as follows:
\begin{prop}\label{prop31}
	Let $\Omega\subset \R^d$ be a non-empty bounded open subset with Lipschitz boundary, $0<s<s'<1$. For $f\in L^2(\Omega)$
	define
	\begin{align*}
	W^s&:=(W^s_j)_{j\in\N}, &  F^{s,s'}&:=(F^{s,s'}_j)_{j\in\N}, \\
	W^s_j&:=p^{\frac{sj}{d}}w(f,p^{-\frac{j}{d}}),  & F^{s,s'}_j&:=p^{js/d}K_{s'}(f,p^{-\frac{js'}{d}}),\\
	&&K_{s'}(f,t)&=\inf\limits_{g\in H^{s'}(\Omega)} \big(\|f-g\|_{L^{2}(\Omega)}+t\|g\|_{H^{s'}(\Omega)}\big),
	\end{align*}
then $\|\cdot\|_{L^2}+[\cdot ]_{H,s}$ given by \eqref{eq-fhs} and $f\mapsto \|f\|_{L^2}+\|W^s\|_{\ell^2}$ and $f\mapsto \|f\|_{L^2}+\|F^{s,s'}\|_{\ell^2}$
are equivalent norms on $H^s(\Omega)$.
\end{prop}

Remark that the definitions of $W^s_j$ and $F^{s,s'}_j$ correspond to the choice $b:=p^{\frac{1}{d}}$ in \eqref{tmp09} and $b:=b:=p^{\frac{s'}{d}}$ in \eqref{eq-kf}, respectively.

Finally, recall that if $\Omega$ is bounded with Lipschitz boundary and $0\le s<\frac{1}{2}$, then $C^\infty_c(\Omega)$ is dense in $H^s(\Omega)$, see e.g. Eq.~(2.220) in \cite{mitrea}.

%
%


\subsection{Multiscale decompositions}\label{ssmult}

Let $\Omega\subset \R^d$ be a non-empty bounded open subset. Our next aim is to decompose $\Omega$ in a very special (but still quite natural) way. We adapt the construction proposed in \cite{msv} for $p=2$, which is in turn
a geometric realization of the approximation spaces used in the wavelet analysis, see e.g. \cite[Chapter 2]{cohen} or \cite[Chapter 2]{meyer}.

 By a \emph{$p$-multiscale decomposition} of $\Omega$ we mean a collection $(\Omega_{n,k})_{n\in\N_0, k\in\{0,\dots,p^n-1\}}$ of non-empty subsets of $\Omega$ such that
	\begin{enumerate}
		\item[(A1)] $\Omega_{0,0}=\Omega$,
		\item[(A2)] for any $n\in\N_0$ the sets $\displaystyle\Omega_{n,0}, \dots, \Omega_{n,p^{n}-1}$ are mutually disjoint,
		\item[(A3)] for any $n\in \N_0$ and $k=0,\dots,p^n-1$ one has
		\begin{gather*}
			\!\!
		\Omega_{n+1,pk+j}\subset \Omega_{n,k} \text{ for any $j\in\{0,\dots,p-1\}$,}\quad
		\Big| \Omega_{n,k}\setminus \bigcup\limits_{j=0}^{p-1}\Omega_{n+1,pk+j}\Big|=0.
		\end{gather*}
	\end{enumerate}
The above conditions can be viewed as an hierarchical decomposition procedure: one sets $\Omega_{0,0}:=\Omega$, and if for some $n$ all $\Omega_{n,k}$ are already constructed, then
one decomposes each $\Omega_{n,k}$ (up to zero measure sets) into $p$ disjoint pieces $\Omega_{n+1,pk+j}$, $j\in\{0,\dots,p-1\}$. In order to have a control of the size of $\Omega_{n,k}$ we introduce further classes of conditions.

A decomposition $(\Omega_{n,k})$ is \emph{weakly balanced} if
\begin{enumerate}
		\item[(A4)] there is  $C_0\ge 1$ such that
		\[
		\dfrac{1}{C_0} \frac{|\Omega|}{p^n} \leq |\Omega_{n,k}|\leq 	C_0\, \frac{|\Omega|}{p^n}
		\]
		 for all $n\in \N_0$ and $k\in\{0,\dots,p^n-1\}$,
\end{enumerate}
and is \emph{strongly balanced} if it satisfies the stronger condition
\begin{itemize}
	\item[(A4*)] $|\Omega_{n,k}|=\dfrac{|\Omega|}{p^n}$ for all $n\in \N_0$ and $k\in\{0,\dots,p^n-1\}$.
\end{itemize}
Finally, a decomposition $(\Omega_{n,k})$ is called \emph{regular} if it satisfes the following two conditions:
\begin{enumerate}
\item[(A5)] there exists $c_1>0$ such that for all $n\in\N_0$ and $k\in\{0,\dots,p^n-1\}$ one has
\[
\diam \Omega_{n,k} \leq c_1 p^{-\frac{n}{d}},
\]
\item[(A6)]\label{} there exists $c_2>0$ such that for all $h\in \R^d$, $n\in\N_0$ and $k\in\{0,\dots,p^n-1\}$ one has
\[
\big|\Omega_{n,k}\setminus(\Omega_{n,k}+h)\big| \leq c_2 |h|p^{-\frac{n(d-1)}{d}}.
\]
\end{enumerate}
Very roughly, the last two conditions say that the shape of $\Omega_{n,k}$ cannot become
``too complicated'' for large $n$.
For the rest of the subsection we assume that:
\begin{equation}
\framebox{\parbox{95mm}{$\Omega\subset\R^d$ is a bounded open set with Lipschitz boundary which admits a regular weakly balanced $p$-multiscale decomposition $\cO:=(\Omega_{n,k})$.
}}
\label{eq-omega}	
\end{equation}

 This covers a large class of $\Omega$: we refer to Section \ref{sec52} for a more detailed discussion.
We are going to establish further properties of $\cO$.
	
\begin{lemma}\label{Mnk} Under the assumption \eqref{eq-omega} there holds
	\[
	K:=\sup_{(n,k)}\# \big\{j:\,\dist (\Omega_{n,j},\Omega_{n,k})\leq p^{-\frac{n}{d}}\big\}<\infty.
	\]
\end{lemma}

\begin{proof} Let us pick some $(n,k)$. Recall that by assumption (A5) we have the inequality $\diam \Omega_{n,j}\leq c_1 p^{-\frac{n}{d}}$
for all $j$. Now let $j$ be such that $\dist (\Omega_{n,j},\Omega_{n,k})\leq p^{-\frac{n}{d}}$,
then there exist $x_{n,j}\in \Omega_{n,j}$  and $x_{n,k}\in \Omega_{n,k}$ with $|x_{n,j}-x_{n,k}|< 2p^{-\frac{n}{d}}$.
It follows that for any $x\in \Omega_{n,j}$ one has the inequalities
\[
|x-x_{n,k}|\le |x-x_{n,j}|+|x_{n,j}-x_{n,k}|< \diam \Omega_{n,j}+ 2p^{-\frac{n}{d}}\le (c_1+2)p^{-\frac{n}{d}},
\]
which shows the inclusion $\Omega_{n,j}\subset B_{(c_1+2)p^{-\frac{n}{d}}}(x_{n,k})$.
Using $|\Omega_{n,j}|\ge C_0^{-1} |\Omega| p^{-n}$,
the number of possible $j$'s is bounded from above by the number
\[
\dfrac{ \big|B_{(c_1+2)p^{-\frac{n}{d}}}(x_{n,k})\big| }{C_0^{-1}|\Omega| p^{-n}} \equiv
\frac{C_0\pi^{\frac{d}{2}}}{\Gamma(\frac{d\mathstrut}{2}+1)} \dfrac{\big((c_1+2)p^{-\frac{n}{d}}\big)^d}{|\Omega| p^{-n}}
\equiv \frac{C_0\pi^\frac{d}{2}}{\Gamma(\frac{d\mathstrut}{2}+1)}\dfrac{(c_1+2)^d}{|\Omega|}. \qedhere
\]
\end{proof}

For $n\in\N_0$ define
\begin{equation}
	\label{vn0}
\begin{aligned}
	V_n&:=\operatorname{span}\big\{\one_{\Omega_{n,k}}:\ k=0,\dots,p^n-1\big\} \subset L^2(\Omega),\\
    P_n&:=\text{the orthogonal projector on $V_n$ in $L^2(\Omega)$};
\end{aligned}
\end{equation}
in other words,
\[
P_n f=\sum_{k=0}^{p^n-1} \dfrac{ 1 }{|\Omega_{n,k}|}\int_{\Omega_{n,k}} f\dd x\, \one_{\Omega_{n,k}}.
\]
Due to the assumption (A3) for any $n$ we have
\[
\one_{\Omega_{n,k}}=\sum_{j=0}^{p-1}\one_{\Omega_{n+1,pk+j}} \text{ a.e.},
\]
which shows that $(V_n)_{n\in\N_0}$ is a strictly increasing sequence of closed subspaces.
We will be interested in approximating arbitrary $f$ by $P_n f$ with large $n$.

\begin{lemma}\label{lem33}
	For any $f\in L^2(\Omega)$ one has $f=\lim_{n\to \infty}P_n f$.
\end{lemma}

\begin{proof}
This is an adaptation of \cite[Lemma 4.5]{msv}.
	
(i)	Let $g\in C^0(\overline{\Omega})$ and $\eps>0$. As $g$ is uniformly continuous on $\overline{\Omega}$, one can find some $\delta>0$ such that $|g(x)-g(y)|<\eps$ for all
$x,y\in\Omega$ with $|x-y|<\delta$. By (A5) one can find some $N\in\N$ such that
$\diam \Omega_{n,k}<\delta$ for all $n\ge N$ and all $k$. Now let $n\ge N$ and pick $x_{n,k}\in \Omega_{n,k}$, then for any $x\in \Omega_{n,k}$ there holds
$\big|g(x)-g(x_{n,k})\big|<\eps$. Therefore,
\[
\|g-g_n\|_\infty<\eps \text{ for }g_n:=\sum_{k=0}^{p^n-1} g(x_{n,k})\one_{\Omega_{n,k}}\in V_n,
\]
which yields $\|g-P_ng\|_{L^2(\Omega)}\le \|g-g_n\|_{L^2(\Omega)}\le \|g-g_n\|_\infty^2 \sqrt{|\Omega|}\le \eps \sqrt{|\Omega|}$ for all $n\ge N$.

This shows that $\|g-P_ng\|_{L^2(\Omega)}\xrightarrow{n\to\infty} 0$ for any $g\in C^0(\overline{\Omega})$.

(ii) Let $f\in L^2(\Omega)$. Let $\eps>0$, then there exists $g\in C^0(\overline{\Omega})$
with $\|f-g\|_{L^2(\Omega)}<\eps$. By (i) there is $N\in\N$ such that $\|g-P_ng\|_{L^2(\Omega)}<\eps$ for all $n\ge N$. Then for all $n\ge N$ one has
\begin{align*}
	\|f-P_n f\|_{L^2(\Omega)}&\le \|f-g\|_{L^2(\Omega)}+\|g-P_ng\|_{L^2(\Omega)}+\|P_n(g-f)\|_{L^2(\Omega)}\\
		&\le 2\|f-g\|_{L^2(\Omega)}+\|g-P_ng\|_{L^2(\Omega)}< 3\eps,
\end{align*}
which shows the claim.
\end{proof}

We now introduce  the approximation spaces $A^r(\Omega)$ consisting of the functions $f\in L^2(\Omega)$ such that the speed of convergence in Lemma \ref{lem33} can be controlled in some special way. The construction is standard, see e.g. \cite[Sec. 3.5]{cohen}, but we need to recall the precise role of various parameters.

\begin{defin}
Let $r>0$. For $f\in L^2(\Omega)$ set
\[
\xi:=(\xi_n)_{n\in\N_0},\quad \xi_n:=p^{\frac{nr}{d}} \dist_{L^2(\Omega)}(f,V_n)\equiv  p^{\frac{nr}{d}}\|f-P_n f\|_{L^2(\Omega)}.
\]
Then the \emph{approximation space} $A^r(\Omega)$ and its norm are defined by
\begin{equation}
	A^r(\Omega) = \big\{f\in L^2(\Omega):\ \xi\in \ell^2\big\}, \quad
\|f\|^2_{A^r(\Omega)}=\|P_0f\|^2_{L^2(\Omega)}+\|\xi\|^2_{\ell^2}.
\end{equation}
Remark that
\[
\|f\|^2_{A^r(\Omega)}\ge \|P_0 f\|^2_{L^2(\Omega)}+|\xi_0|^2\equiv
\|P_0 f\|^2_{L^2(\Omega)}+\|f-P_0 f\|^2_{L^2(\Omega)}\equiv \|f\|^2_{L^2(\Omega)}.
\]
\end{defin}
We stress that the space $A^r(\Omega)$ depends on the decomposition $(\Omega_{n,k})$, but it is not reflected in the notation.

For what follows it will be useful to work with another equivalent norm on $A^r(\Omega)$. Using the spaces $V_n$ from \eqref{vn0}, for any $n\in \N_0$ we introduce  
\begin{equation}
	\label{un0}
	U_n:=\begin{cases}
		V_0,&n=0,\\
		V_n\cap V_{n-1}^\perp, & n\ge 1.
	\end{cases}
\end{equation}
By construction this gives the orthogonal decomposition
$L^2(\Omega)=\bigoplus_{n=0}^\infty U_n$, the orthogonal projector $Q_n$
on $U_n$ is given by
\[
Q_n:=\begin{cases}
	P_0, & n=0,\\
	P_n-P_{n-1}, & n\ge 1,
\end{cases}
\]
and for each $n\in\N_0$ and $f\in L^2(\Omega)$ we have
\begin{equation}
	\label{PQod}
f=\sum_{k=0}^\infty Q_k f= P_n f+\sum_{k=n+1}^\infty Q_k f,
\end{equation}
while the summands on the right-hand side are mutually orthogonal. The following result is a particular case of \cite[Thm. 3.5.3]{cohen}, but we include it for the sake of completeness.
\begin{lemma}\label{lem35}
	Let $r>0$. 
	For $f\in L^2(\Omega)$ set
	\[
	\zeta:=(\zeta_n)_{n\in\N_0},\quad \zeta_n:=p^{\frac{nr}{d}}\|Q_nf\|_{L^2(\Omega)}.
	\]
	Then $\|\cdot \|_{A^r(\Omega)}$ and $f\mapsto \|\zeta\|_{\ell^2}$ are equivalent norms on $A^r(\Omega)$. 
\end{lemma}
\begin{proof}
Recall that by definition there holds
\begin{align*}
	\|f\|^2_{A^r(\Omega)}=\|Q_0f\|^2_{L^2}+\sum\limits_{n=0}^{\infty}\|f-P_nf\|^2_{L^2(\Omega)}p^{2\frac{nr}{d}}.
\end{align*}
We have
\begin{align*}
	\|\zeta\|^2_{\ell^2}&\equiv \|Q_0 f\|^2_{L^2(\Omega)}
	+\sum\limits_{n=1}^{\infty}\|Q_nf\|^2_{L^2(\Omega)}p^{2\frac{nr}{d}}\\ 
	&\stackrel{\eqref{PQod}}{\leq}
	\|Q_0 f\|^2_{L^2(\Omega)}+ \sum\limits_{n=1}^{\infty}\|f-P_{n-1}f\|^2_{L^2(\Omega)}p^{2\frac{nr}{d}}\\
	&\equiv \|Q_0 f\|^2_{L^2(\Omega)}+ p^{2\frac{nr}{d}}\sum\limits_{n=0}^{\infty}\|f-P_n f\|^2_{L^2(\Omega)}p^{2\frac{nr}{d}}\\
	&\le p^{2\frac{r}{d}}\Big(  \|Q_0 f\|^2_{L^2(\Omega)}+ \sum\limits_{n=0}^{\infty}\|f-P_n f\|^2_{L^2(\Omega)}p^{2\frac{nr}{d}}\Big)
	\equiv 	p^{2\frac{r}{d}}\|f\|^2_{A^r(\Omega)}.
\end{align*}
At the same time,	
\begin{align*}
	\|f\|^2_{A^r(\Omega)}&=\|Q_0f\|^2_{L^2}+\sum\limits_{n=0}^{\infty}\|f-P_nf\|^2_{L^2(\Omega)}p^{2\frac{nr}{d}}\\
	&\stackrel{\eqref{PQod}}{=}\|Q_0 f\|^2_{L^2(\Omega)}+\sum\limits_{n=0}^{\infty}\sum_{k=n+1}^{\infty}\|Q_kf\|_{L^2(\Omega)}^2p^{2\frac{nr}{d}}\\
	&=\|Q_0 f\|^2_{L^2(\Omega)}+\sum\limits_{k=1}^{\infty} \|Q_kf\|_{L^2(\Omega)}^2 \sum_{n=0}^{k-1}p^{2\frac{nr}{d}}\\
	&=\|Q_0 f\|^2_{L^2(\Omega)}+\sum\limits_{k=1}^{\infty}\|Q_kf\|_{L^2}^2 \frac{p^{2\frac{kr}{d}}-1}{p^{2\frac{r}{d}}-1}\\
	&\leq \|Q_0 f\|^2_{L^2(\Omega)}+\frac{1}{p^{2\frac{r}{d}}-1}\sum\limits_{k=1}^{\infty}p^{2\frac{kr}{d}} \|Q_kf\|_{L^2}^2\\
	&\le\max\Big( 1, \dfrac{1}{p^{2\frac{r}{d}}-1}\Big)\sum_{k=0}^\infty p^{2\frac{kr}{d}}\|Q_k f\|^2_{L^2(\Omega)}
	\equiv \max\Big( 1, \dfrac{1}{p^{2\frac{r}{d}}-1}\Big) \|\zeta\|^2_{\ell^2}. \qedhere
\end{align*}
\end{proof}

\subsection{Relating approximation spaces to Sobolev spaces}\label{sec33}

Now we are going to compare the above equivalent norms on $A^r(\Omega)$ with suitable Sobolev norms.
All results of this section are suitable adaptations of respective results from \cite[Sec. 4]{msv}
for $p=2$ and for slightly different definitions of multiscale decompositions and the spaces $A^r$.

Recall that the seminorm $[\cdot]_{H^s}$ was defined in \eqref{eq-fhs}. 

\begin{lemma}\label{lem45}
For any $0<s<1$ there exists $B>0$ such that for any $n\in \N_0$ and $f\in H^s(\Omega)$ there holds $\|f-P_nf\|_{L^2(\Omega)} \leq B p^{-\frac{ns}{d}}[f]_{H^s(\Omega)}$.
\end{lemma}
\begin{proof}
Let $n\in\N_0$ and $k\in\{0,\dots,p^n-1\}$. Recall that $P_n f$ is piecewise constant,
	\begin{align*}
		P_nf (x)=\frac{1}{|\Omega_{n,k}|}\int_{\Omega_{n,k}}f(y)\dd y \text{ as } x\in \Omega_{n,k}.
	\end{align*}
Using Cauchy-Schwarz inequality we estimate
\begin{align*}
	\int_{\Omega_{n,k}}\big|f(x)-P_n f(x)\big|^2\dd x &
	=	\int_{\Omega_{n,k}}\Big|f(x)-\frac{1}{|\Omega_{n,k}|}\int_{\Omega_{n,k}}f(y)\dd y\Big|^2\dd x\\
	&= \int_{\Omega_{n,k}}\Big|\frac{1}{|\Omega_{n,k}|}\int_{\Omega_{n,k}} \big(f(x)-f(y)\big)\dd y\Big|^2\dd x\\
	&
	\le \int_{\Omega_{n,k}} \dfrac{1}{|\Omega_{n,k}|^2}\, |\Omega_{n,k}|
	\int_{\Omega_{n,k}} \big|f(x)-f(y)\big|^2\dd y \dd x\\
	&= \dfrac{1}{|\Omega_{n,k}|}\iint_{\Omega_{n,k}\times \Omega_{n,k}} 	\big|f(x)-f(y)\big|^2\dd x \dd y.
\end{align*}
For all $x,y\in\Omega_{n,k}$ with $x\ne y $ one has $|x-y|\le \diam\Omega_{n,k}$, which gives
\begin{align*}
\int_{\Omega_{n,k}}\big|f(x)-P_n f(x)\big|^2\dd x & \le \dfrac{(\diam\Omega_{n,k})^{d+2s}}{|\Omega_{n,k}|}\iint_{\Omega_{n,k}\times \Omega_{n,k}} 	\dfrac{\big|f(x)-f(y)\big|^2}{|x-y|^{d+2s}}\dd x \dd y\\
&\equiv \dfrac{(\diam\Omega_{n,k})^{d+2s}}{|\Omega_{n,k}|} [f]^2_{H^s(\Omega)}.
\end{align*}
Recall due to the choice of $\Omega_{n,k}$ (see Subsection \ref{ssmult}) and (A4) we have
\[
\dfrac{(\diam\Omega_{n,k})^{d+2s}}{|\Omega_{n,k}|}\le c_1 p^{-\frac{n}{d}},
\quad
|\Omega_{n,k}|\ge \frac{|\Omega|}{C_0 p^n},
\]
therefore,
\[
\dfrac{(\diam\Omega_{n,k})^{d+2s}}{|\Omega_{n,k}|}
\le
\dfrac{C_0(c_1 p^{-\frac{n}{d}})^{d+2s} p^n}{|\Omega|}
=B p^{-\frac{2sn}{d}}, \quad B:= \dfrac{C_0 c_1^{d+2s} }{|\Omega|}
\]
resulting in $\|f-P_n\|^2_{L^2(\Omega_{n,k})}\le B p^{-\frac{2sn}{d}}[f]^2_{H^s(\Omega_{n,k})}$. Then
%
\begin{align*}
		\|f-P_nf\|_{L^2(\Omega)}^2&=\sum_{k=0}^{p^n-1} \|f-P_nf\|_{L^2(\Omega_{n,k})}^2\\
		& \leq B p^{-\frac{2ns}{d}}\sum_{k=0}^{b^n-1}
		[f]^2_{H^s(\Omega_{n,k})}\le  B p^{-\frac{ns}{d}}\|f\|_{H^s(\Omega)}. \qedhere
	\end{align*}
\end{proof}

This allows one to show the first embedding result:
\begin{theorem}\label{thmemb1}
Assume \eqref{eq-omega}, then $H^{s}(\Omega)\hookrightarrow A^{r}(\Omega)$
for any $0\leq r\leq s<1$.
\end{theorem}

\begin{proof}
We equip $H^s(\Omega)$ with the interpolation norm (Proposition \ref{prop31})
	\[
	\|f\|_{H^s}^2:=\|f\|^2_{L^2(\Omega)}+ \big\|p^{\frac{js}{d}}K_{s'}(f,^{-\frac{js'}{d}})\big\|^2_{\ell^2}, \quad 0<s<s'<1.
	\]
Let $f\in L^2(\Omega)$. As $P_j$ are orthogonal projections, for any $j\in\N_0$  and any $g\in H^{s'}(\Omega)\subset L^2(\Omega)$ one has
	\begin{align*}
		\|f-P_jf\|_{L^2(\Omega)} \leq \|f-P_jg\|_{L^2(\Omega)}\leq \|f-g\|_{L^2(\Omega)} + \|g-P_jg\|_{L^2(\Omega)}.
	\end{align*}
	Using Lemma \ref{lem45} to estimate the last summand we conclude that one can choose
	some $C\ge1$ such that
	\begin{align*}
		\|f-P_jf\|_{L^2(\Omega)} \leq \|f-g\|_{L^2(\Omega)} + C b^{-\frac{js'}{d}}\|g\|_{H^{s'}(\Omega)} \text{ for all }j\in\N_0,
	\end{align*}
	which immediately gives
	\begin{equation}
		\label{fpjf}
		\|f-P_jf\|_{L^2(\Omega)} \leq C K_{s'}(f,b^{-\frac{js'}{d}})
	\end{equation}

	Now let $f\in H^{s}(\Omega)$, then
	\begin{align*}
		\|f\|^2_{A^{r}(\Omega)}&= \|P_0f\|^2_{L^2(\Omega)} + \sum_{j=0}^{\infty} p^{2\frac{jr}{d}}\|f-P_jf\|^2_{L^2(\Omega)}\\
		\text{use \eqref{fpjf}: }	&\leq C\Big( \|f\|^2_{L^2(\Omega)}+ \sum_{j=0}^{\infty} p^{2\frac{jr}{d}} K_{s'}(f,p^{-\frac{js'}{d}})^2 \Big) \\
		&\equiv C\Big( \|f\|^2_{L^2(\Omega)}+ \sum_{j=0}^{\infty} p^{2(r-s)\frac{j}{d}} p^{\frac{2sj}{d}} K_{s'}(f,p^{-\frac{js'}{d}})^2 \Big)\\
		\text{use $r\le s$: }	& \le C\Big( \|f\|^2_{L^2(\Omega)}+ \sum_{j=0}^{\infty}  p^{\frac{2sj}{d}} K_{s'}(f,p^{-\frac{js'}{d}})^2 \Big)\\
		&\equiv C \Big(\|f\|^2_{L^2(\Omega)}+\big\|p^{\frac{js}{d}}K_{s'}(f,^{-\frac{js'}{d}})\big\|^2_{\ell^2}\Big)\equiv
		C\|f\|^2_{H^s(\Omega)}.\qedhere
	\end{align*}
\end{proof}

In order to obtain embedding results in the other direction more work is needed.
Recall that the modulus of smoothness $w(f,t)$ was defined in \eqref{wft}.
\begin{lemma}\label{lem46}
There exists $C\ge 2$ such that for all $n\in \N_0$, $t>0$ and
$f\in V_n$ there holds $w(f,t)\leq C p^{\frac{n}{2d}}\, t^{\frac{1}{2}}\,\|f\|_{L^2(\Omega)}$.
\end{lemma}
\begin{proof}
Consider first the case $t\in(0,p^{-\frac{n}{d}})$ for some $n\in\N_0$.
Let $f\in V_n$, then
\begin{align}
f&=\sum\limits_{j=0}^{p^n-1}f_{n,j}\one_{\Omega_{n,j}},\quad f_{n,j}\in\C, \nonumber\\
\|f\|_{L^2(\Omega)}^2 &= \sum\limits_{j=0}^{p^n-1} |f_{n,j}|^2|\Omega_{n,j}| \geq  \frac{|\Omega|}{C_0 p^n}\sum\limits_{j=0}^{p^n-1} |f_{n,j}|^2, \label{tmp11}
\end{align}
where we used (A4) in the last step. Recall that by (A6) for any $(n,k)$ there holds 
\begin{equation}
	 \label{tmp12}
\big|\Omega_{n,k}\setminus(\Omega_{n,k}-h)\big|	\equiv
\|\mathbbm{1}_{\Omega_{n,k}-h}-\mathbbm{1}_{\Omega_{n,k}}\|_{L^1(\Omega_{n,k})}\leq c_2 |h|p^{-\frac{n(d-1)}{d}}, \quad h\in \R^d,
\end{equation}
and by Lemma \ref{Mnk} for some $K>0$ we have
\begin{equation}
	\label{tmp-kkk}
\#\big\{j:\  \dist (\Omega_{n,j},\Omega_{n,k})\leq p^{-\frac{n}{d}}\big\}\le K \text{ for any $(n,k)$.}
\end{equation}
Remark that $f(y)=f_{n,k}$ for any $y\in \Omega_{n,k}$. In particular, if $x\in \Omega_{n,k}$
and $x+h\in\Omega_{n,k}$, then $f(x+h)=f(x)$. It follows that for any $x\in \Omega_{n,k}$ one has
\begin{equation}
\begin{aligned}
\big| f(x+h)-f(x)\big|
&=\big| \one_{\Omega_{n,k}}(x+h)-\one_{\Omega_{n,k}}(x)\big| \cdot \big|f(x+h)-f(x)\big|\\
&\equiv
	\big| \one_{\Omega_{n,k}-h}(x)-\one_{\Omega_{n,k}}(x)\big| \cdot \big|f(x+h)-f(x)\big|.
\end{aligned}
   \label{tmp7}
\end{equation}

Now let $h\in\R^d$ with $|h|\le t$. Then for any $x\in \Omega_{n,k}$ we estimate
\begin{align*}
\big| f(x+h)-f(x)\big|^2&=\Big|\sum_{j=0}^{p^n-1} f_{n,j} \one_{\Omega_{n,j}}(x+h)-f_{n,k}\Big|^2\\
&\le \big(|f_{n,k}|+\sum_{j=0}^{p^n-1} |f_{n,j}|  \one_{\Omega_{n,j}}(x+h)\big)^2\\
&\le \big(|f_{n,k}|+\sum_{j:\,\dist(\Omega_{n,j},\Omega_{n,k})\le |h|} |f_{n,j}|\big)^2\\
&\le \big(|f_{n,k}|+\sum_{j:\,\dist(\Omega_{n,j},\Omega_{n,k})\le p^{-\frac{n}{d}}} |f_{n,j}|\big)^2\\
\text{use \eqref{tmp-kkk}} &\text{ and Cauchy-Schwarz: }\\
&\le (K+1)\Big( |f_{n,k}|^2+\sum_{j:\,\dist(\Omega_{n,j},\Omega_{n,k})\le p^{-\frac{n}{d}}} |f_{n,j}|^2\Big)
\end{align*}
Using this inequality on the right-hand side of \eqref{tmp7} we obtain
\begin{align*}
\int_{\Omega_{n,k}}	&\big| f(x+h)-f(x)\big|^2\dd x\le 
(K+1) \|\one_{\Omega_{n,k}-h}-\one_{\Omega_{n,k}}\|_{L^1(\Omega_{n,k})}\\
&\qquad\qquad\qquad \cdot \Big( |f_{n,k}|^2+\sum_{j:\dist(\Omega_{n,j},\Omega_{n,k})\le p^{-\frac{n}{d}}} |f_{n,j}|^2\Big)\\
\text{use \eqref{tmp12}: }
&\le(K+1)c_2|h|p^{-\frac{n(d-1)}{d}}\Big( |f_{n,k}|^2+\sum_{j:\dist(\Omega_{n,j},\Omega_{n,k})\le p^{-\frac{n}{d}}} |f_{n,j}|^2\Big).
\end{align*}
By summing over all $k$ we arrive at 
\begin{align*}
\int_{\Omega}	\big| f(x+h)&-f(x)\big|^2\dd x=\sum_{k=0}^{p^n-1} \int_{\Omega_{n,k}}	\big| f(x+h)-f(x)\big|^2\dd x\\
{}\le(K+1)&c_2|h|p^{-\frac{n(d-1)}{d}}
\sum_{k=0}^{p^n-1} \Big( |f_{n,k}|^2+\sum_{j:\,\dist(\Omega_{n,j},\Omega_{n,k})\le p^{-\frac{n}{d}}}|f_{n,j}|^2\Big)\\
\text{use \eqref{tmp-kkk}: }&\le(K+1)^2c_2|h|p^{-\frac{n(d-1)}{d}}\sum_{k=0}^{p^n-1} |f_{n,k}|^2\\
\text{use \eqref{tmp11}: }&\le (K+1)^2c_2|h|p^{-\frac{n(d-1)}{d}}\dfrac{C_0 p^n}{|\Omega|} \|f\|^2_{L^2(\Omega)}\\
&\le c^2 p^{\frac{n}{d}}|h|\,\|f\|^2_{L^2(\Omega)}, \quad c:=(K+1)\sqrt{\dfrac{C_0 c_2}{|\Omega|}}.
\end{align*}
%
The last estimate holds for all $h$ with $|h|\le t$,  which yields
\[
w(f,t)\le c p^{\frac{n}{2d}}\,t^{\frac{1}{2}}\,\|f\|^2_{L^2(\Omega)} \text{ for all $t\in(0,p^{-\frac{n}{d}})$.}
\]
For $t\ge p^{-\frac{n}{d}}$ we simply estimate $w(t,f)\leq 2 \|f\|_{L^2(\Omega)}\le 2 p^{\frac{n}{2d}}\,t^{\frac{1}{2}}\|f\|_{L^2(\Omega)}$ and obtain the claim by taking  $C:=\max(2,c)$.
\end{proof}

\begin{theorem}\label{thmemb2} Let \eqref{eq-omega} be satisfied, then for any $0<s<\frac{1}{2}$ and $s\le r$ one has $A^r(\Omega)\hookrightarrow H^{s}(\Omega)$.
\end{theorem}

\begin{proof}
We equip $H^s(\Omega)$ with the Besov norm
\begin{align*}
	\|f\|^2_{B^s} := \|f\|^2_{L^2(\Omega)} + \big\| p^{\frac{sj}{d}}w(f, p^{-\frac{j}{d}})\big\|_{\ell^2}^2
\end{align*}
and recall that for any $t>0$ the modulus of smoothness $w(f,t)$ satisfies the triangle inequalty with respect to $f$. For each $j$ we have
\begin{align*}
f&=\sum\limits_{k=0}^{j}Q_k f+(f-P_jf),\\
	w(f,p^{-\frac{j}{d}})&\le \sum_{k=0}^j w(Q_k,p^{-\frac{j}{d}}) +w(f-P_j f,p^{-\frac{j}{d}}).
\end{align*}
Remark that $Q_kf\in V_k$, so by Lemma \ref{lem46} we estimate with some $C\ge 2$
\begin{align*}
	w(Q_k,p^{-\frac{j}{d}})&\le C p^{\frac{k}{2d}-\frac{j}{2d}}\|Q_k f\|_{L^2(\Omega)},
	\quad
	w(f-P_j f,p^{-\frac{j}{d}})\le C\|f-P_j f\|_{L^2(\Omega)}.
\end{align*}
In addition, for $k\ge 1$ we have $P_{k}P_{k-1}=P_{k-1}$, therefore,
\[
\|Q_k f\|_{L^2(\Omega)}= \|P_k f - P_{k-1}f\|_{L^2(\Omega)}=\|P_k (f-P_{k-1}f)\|_{L^2(\Omega)}\le\|f-P_{k-1}f\|_{L^2(\Omega)},
\]
and then (recall that $Q_0=P_0$)
\begin{align*}
w(f,p^{-\frac{j}{d}})&\le Cp^{-\frac{j}{2d}}\|P_0 f\|_{L^2(\Omega)}\\
&\quad
+C\sum_{k=1}^j p^{\frac{k}{2d}-\frac{j}{2d}}\|f-P_{k-1}f\|_{L^2(\Omega)}
+C \|f-P_j f\|_{L^2(\Omega)}\\
&=Cp^{-\frac{j}{2d}}\|P_0 f\|_{L^2(\Omega)}\\
&\quad
+Cp^{\frac{1}{2d}}\sum_{k=0}^{j-1} p^{\frac{k}{2d}-\frac{j}{2d}}\|f-P_kf\|_{L^2(\Omega)}
+C \|f-P_j f\|_{L^2(\Omega)}\\
&\le Bp^{-\frac{j}{2d}}\|P_0 f\|_{L^2(\Omega)} +B\sum_{k=0}^{j} p^{\frac{k}{2d}-\frac{j}{2d}}\|f-P_kf\|_{L^2(\Omega)}
\end{align*}
with $B:=C p^{\frac{1}{2d}}>C$. 
It follows that for any $j$ one has
\begin{equation}
	\label{bff}
	p^{\frac{sj}{d}}w(f,p^{-\frac{j}{d}})\le B(F'_j +F_j)
\end{equation}	
with sequences $F':=(F'_j)$ and $F:=(F_j)$ given by
\begin{align*}
	F'_j&:=p^{\frac{s j}{d}-\frac{j}{2d}}\|P_0 f\|_{L^2(\Omega)}\equiv
	p^{(s-\frac{1}{2})\frac{j}{d}}\|P_0 f\|_{L^2(\Omega)},\\
	F_j&:=p^{\frac{s j}{d}}\sum_{k=0}^{j} p^{\frac{k}{2d}-\frac{j}{2d}}\|f-P_kf\|_{L^2(\Omega)}.
\end{align*}	
Due to $s<\frac{1}{2}$ we have
\begin{gather*}
\|F'\|_{\ell^2}^2=\sum_{j=0}^\infty |F'_j|^2=a\|P_0 f\|^2_{L^2(\Omega)}
\text{ with }
a:=\sum_{j=0}^\infty p^{2(s-\frac{1}{2})\frac{j}{d}}<\infty.
\end{gather*}
In order to control $F_j$ we represent it as
\begin{align*}
	F_j=\sum_{k=0}^{j} p^{\frac{s j}{d}+\frac{k}{2d}-\frac{j}{2d}-\frac{rk}{d}} p^{\frac{rk}{d}}\|f-P_kf\|_{L^2(\Omega)}.
\end{align*}
We estimate the exponents using $s\le r$:
\begin{align*}
\frac{s j}{d}+\frac{k}{2d}-\frac{j}{2d}-\frac{rk}{d}&=\Big(s-\frac{1}{2}\Big)\frac{j}{d}+\Big(\frac{1}{2}-r\Big)\frac{k}{d}\\
&\le \Big(s-\frac{1}{2}\Big)\frac{j}{d}+\Big(\frac{1}{2}-s\Big)\frac{k}{d}
\equiv \Big(s-\frac{1}{2}\Big)\frac{j-k}{d},
\end{align*}
and arrive at
\begin{equation}
	\label{fj2}
F_j\le\sum_{k=0}^{j} p^{(s-\frac{1}{2})\frac{j-k}{d}} p^{\frac{rk}{d}}\|f-P_kf\|_{L^2(\Omega)}.
\end{equation}
Define $\alpha:=(\alpha_j)_{j\in\Z}$ and $\beta:=(\beta_j)_{j\in\Z}$ by
\begin{align*}
	\alpha_j:=\begin{cases}
		0,& j< 0,\\
		p^{(s-\frac{1}{2})\frac{j}{d}}, & j\ge 0,
	\end{cases}
\qquad
	\beta_j:=\begin{cases}
	0,& j< 0,\\
	p^{\frac{rj}{d}}\|f-P_jf\|_{L^2(\Omega)}, & j\ge 0,
\end{cases}
\end{align*}
then \eqref{fj2} takes the form $F_j\le (\alpha\ast\beta)_j$ for any $j\in\N_0$, with $\ast$
being the convolution product. Using Young's convolution inequality we obtain
\begin{align*}
	\|F\|_{\ell^2}\le \|\alpha\ast \beta\|_{\ell^2(\Z)}\le \|\alpha\|_{\ell^1(\Z)}\|\beta\|_{\ell^2(\Z)}
\end{align*}	
Then it follows by \eqref{bff} that
\begin{align*}
\|f\|_{B^s}^2=\big\|p^{\frac{sj}{d}}w(f,p^{-\frac{j}{d}})\big\|_{\ell^2}^2
&\le 2B\|F'\|_{\ell^2}+2B\|F\|_{\ell^2}^2\\
&\le 2B a\|P_0 f\|_{L^2(\Omega)}^2
+2B\|\alpha\|_{\ell^1(\Z)}^2\|\beta\|_{\ell^2(\Z)}^2.
\end{align*}
Recall that $\|P_0 f\|^2_{L^2(\Omega)}+\|\beta\|^2_{\ell^2(\Z)}=\|f\|^2_{A^r(\Omega)}$, therefore,
\[
\|f\|^2_{B^s}\le 2B\max( a, \|\alpha\|_{\ell^1(\Z)}^2)\|f\|_{A^r(\Omega)} \text{ for all }
f\in A^r(\Omega). \qedhere
\]
\end{proof}

By combining Theorems \ref{thmemb1} and \ref{thmemb2} we obtain the following result.
\begin{corol}
	Under the assumption \eqref{eq-omega} there holds
	\[
	A^r(\Omega)=H^{r}(\Omega) \text{ for } 0\le r<\frac{1}{2}.
	\]	
\end{corol}

\subsection{Extension to open sets in manifolds}\label{sec44}
Let  $(S,g)$ be an $d$-dimensional Riemannian manifold of bounded geometry. For the construction of Sobolev spaces $H^s(S)$ with $s\ge 0$ we refer to \cite{gs}.
Let  $\Omega\subset S$ be a non-empty open set such that $\overline{\Omega}\subset S$ is compact (the case $\Omega=S$ is possible, if $S$ itself is compact). 
The Sobolev space  $H^s(\Omega)$ on $\Omega$ is then
defined as the space of the restrictions on $\Omega$ of the functions from $H^s(S)$ with the quotient norm
\[
\|f\|_{H^s(\Omega),*}=\inf_{F\in H^s(S),\  F|_\Omega=f}\|F\|_{H^s(S)}.
\]
It follows from the general construction of Sobolev spaces that:
\begin{itemize}
	\item for any open $\Omega_0\subset \Omega$ the linear map $H^s(\Omega) \ni f\mapsto f|_{\Omega_0}\in H^s(\Omega_0)$ is bounded.
	\item if for some local chart $\Phi:\R^d\supset \Tilde O\mapsto O\subset S$  one has $\overline\Omega\subset O$, then the map $f\mapsto f\circ\Phi$
	is an isomorphism between $H^s(\Omega)$ and $H^s(\Tilde \Omega)$, with $\Tilde\Omega:=\Phi^{-1}(\Omega)\subset\R^d$.
\end{itemize}

We say that non-empty subsets $(\Omega_{n,k})_{n\in\N_0, k=0,\dots,p^n-1} \subset \Omega$
form a  \emph{$p$-multiscale decomposition of $\Omega$}, if the following conditions hold:
	\begin{enumerate}
	\item[(B1)] $\Omega_{0,0}=\Omega$,
	\item[(B2)] for any $n\in\N_0$ the sets $\displaystyle\Omega_{n,0}, \dots, \Omega_{n,p^{n}-1}$ are mutually disjoint,
	\item[(B3)] for any $n\in \N_0$ and $k=0,\dots,p^n-1$ one has
	\begin{gather*}
		\!\!\!
		\Omega_{n+1,pk+j}\subset \Omega_{n,k} \text{ for any $j\in\{0,\dots,p-1\}$,}\quad
		\Big| \Omega_{n,k}\setminus \bigcup\limits_{j=0}^{p-1}\Omega_{n+1,pk+j}\Big|=0.
	\end{gather*}
\end{enumerate}
This decomposition is called \emph{regular and weakly balanced} if it satisfies additionally:
\begin{enumerate}
\item[(B4)] For some is $N\in\N_0$ each $\overline{\Omega_{N,K}}$
with $K\in\{0,\dots,p^N-1\}$ is covered by a local chart
$\Phi_{N,K}$ on $S$ such that the sets
\[
\Tilde\Omega_{N,K}:=\Phi_{N,K}^{-1}(\Omega_{N,K})
\]
are bounded open sets with Lipschitz boundaries in $\R^d$.	
\item[(B5)] for each $K\in\{0,\dots,p^N-1\}$ the sets
\[
\Tilde \Omega^{N,K}_{n,k}:=\Tilde \Omega_{N+n,p^n K+k}, \quad n\in\N_0, \quad k\in\{0,\dots,p^n-1\},
\]
form a \emph{regular weakly balanced $p$-multiscale decomposition of $\Tilde\Omega_{N,K}$,}
\end{enumerate}
and it is called \emph{regular and strongly balanced} if one has in addition, for $N$ from (B4),
\begin{enumerate}
	\item[(B6)] $|\Omega_{N,K}|=p^{-N} |\Omega|$ for all $K\in\{0,\dots,p^N-1\}$,
	\item[(B7)] for each $K\in\{0,\dots,p^N-1\}$ the sets
	\[
	\Tilde \Omega^{N,K}_{n,k}:=\Tilde \Omega_{N+n,p^n K+k}, \quad n\in\N_0, \quad k\in\{0,\dots,p^n-1\},
	\]
	form a \emph{regular strongly balanced $p$-multiscale decomposition of $\Tilde\Omega_{N,K}$.}
\end{enumerate}

For the rest of the subsection we assume that:
\begin{equation}
	\framebox{\parbox{100mm}{$\Omega\subset S$ is an open set with compact closure
			which admits a regular weakly balanced $p$-multiscale decomposition $(\Omega_{n,k})$
	}}
	\label{eq-omega2}	
\end{equation}
and let $N$ and $\Tilde\Omega_{n,k}$ be as in (B4)--(B5). Then the sets
\begin{equation}
	\label{omk}
\Omega^{N,K}_{n,k}:=\Omega_{N+n,p^n K+k},\quad
n\in\N_0, \quad k\in\{0,\dots,p^n-1\},
\end{equation}
form a regular balanced multiscale decomposition of $\Omega_{N,K}$ for  $K\in\{0,\dots,p^N-1\}$.
In addition, the decomposition $(\Omega_{n,k})$ gives rise to the projectors $P_n$ and the spaces $V_n$ and $A^r(\Omega)$ defined in the same way as in the Euclidean case: for $n\in\N_0$ define
\begin{equation}
	\label{pqmanif}
	\begin{aligned}
		V_n&:=\operatorname{span}\big\{\one_{\Omega_{n,k}}:\ k=0,\dots,p^n-1\big\} \subset L^2(\Omega),\\
		U_n&:=\begin{cases}
				V_0,&n=0,\\
				V_n\cap V_{n-1}^\perp, & n\ge 1.
			\end{cases}\\
		P_n&:=\text{the orthogonal projector on $V_n$ in $L^2(\Omega)$},\\
		Q_n&:=\text{the orthogonal projector on $U_n$ in $L^2(\Omega)$}.
	\end{aligned}
\end{equation}
Remark that Lemmas \ref{lem33} and \ref{lem35} are transferred directly to this new setting.
We establish several additional properties of the spaces $A^r(\Omega)$.

\begin{lemma}\label{lem48}
Let $A^r(\Omega_{N,K})$ be the approximation spaces associated with the decompositions $(\Omega^{N,K}_{n,k})$ from \ref{omk}, then for any $r>0$ and $N\in\N$ the map
\[
J:\ A^r(\Omega)\ni f\mapsto (f_{N,K})_{K\in\{0,\dots,p^n-1\}}\in \bigoplus_{K=0}^{p^n-1}A^r(\Omega_{N,K}),
\quad f_{N,K}:=f|_{\Omega_{N,K}},
\]
is an isomorphism.
\end{lemma}

\begin{proof} For $f\in L^2(\Omega)$ denote
	\[
	\alpha_{n,k}:=\dfrac{1}{|\Omega_{n,k}|}\int_{\Omega_{n,k}} f\dvol_g,
	\]
	then
	\[
	\dfrac{1}{|\Omega^{N,K}_{n,k}|}\int_{\Omega^{N,K}_{n,k}}f\dvol_g=\alpha_{N+n,p^nK+k}.
	\]
	For $n\in\N_0$ let $P^{N,K}_n$ be the orthogonal projector on 
	\[
	V^{N,K}_n:=\operatorname{span}\big\{\one_{\Omega^{N,K}_{n,k}}:\, k\in\{0,\dots,p^n-1\} \big\}\subset L^2(\Omega_{N,K})
	\]
	in $L^2(\Omega_{N,K})$, then 
	\begin{align*}
		P^{N,K}_n f_{N,K}&=\sum_{k=0}^{p^n-1} \dfrac{1}{|\Omega^{N,K}_{n,k}|}\int_{\Omega^{N,K}_{n,k}}f\dvol_g \,\one_{\Omega^{N,K}_{n,k}}\\
		&\equiv 
		\sum_{k=0}^{p^n-1} \alpha_{N+n,p^n K+k} \one_{\Omega_{N+n,p^n K+k}}.
	\end{align*}
	At the same time,
	\begin{align*}
		P_{N+n} f&=\sum_{j=0}^{p^{N+n}-1} \alpha_{n,j} \one_{\Omega_{n,j}}=\sum_{K=0}^{p^N-1}\sum_{k=0}^{p^n-1} \alpha_{N+n,p^nK+k} \,\one_{\Omega_{N+n,p^nK+k}}\\
		&=\sum_{K=0}^{p^N-1} (P^{N,K}_n f_{N,K}) \one_{\Omega_{N,K}},\\
		f&=\sum_{K=0}^{p^N-1} f_{N,K} \,\one_{\Omega_{N,K}},\\
		f-P_{N+n} f&=\sum_{K=0}^{p^N-1} (f_{N,K} -P^{N,K}_n f_{N,K}) \,\one_{\Omega_{N,K}}
	\end{align*}
	As the summands in the last sum have disjoint supports, they are orthogonal in $L^2(\Omega)$, and
	\begin{equation}
		\label{fpn}
		\|f-P_{N+n} f\|^2_{L^2(\Omega)}=\sum_{K=0}^{p^N-1} \|f_{N,K}-P^{N,K}_n f_{N,K}\|^2.
	\end{equation}
	For any $r\ge 0$ we have
	\begin{equation}
		\label{fpn2}
		\begin{aligned}
			\sum_{n=0}^\infty p^{2\frac{nr}{d}}\|f-P_{N+n} f\|^2_{L^2(\Omega)}&
			\stackrel{\eqref{fpn}}{=}\sum_{K=0}^{p^N-1}\sum_{n=0}^\infty p^{2\frac{nr}{d}}\|f_{N,K}-P^{N,K}_n f_{N,K}\|^2\\
			&\equiv \sum_{K=0}^{p^N-1} \|f_{N,K}\|^2_{A^r(\Omega_{N,K})}.
		\end{aligned}
	\end{equation}
	Therefore, if $f\in A^r(\Omega)$, then
	\begin{align*}
		\|Jf\|^2&=\sum_{K=0}^{p^N-1} \|f_{N,K}\|^2_{A^r(\Omega_{N,K})}
		\stackrel{\eqref{fpn2}}{=} \sum_{n=0}^\infty p^{2\frac{nr}{d}} 	\|f-P_{N+n} f\|^2_{L^2(\Omega)}\\
		&=p^{-2\frac{Nr}{d}}\sum_{n=N}^\infty p^{2\frac{nr}{d}} 	\|f-P_n f\|^2_{L^2(\Omega)}
		\le p^{-2\frac{Nr}{d}}\sum_{n=0}^\infty p^{2\frac{nr}{d}} 	\|f-P_n f\|^2_{L^2(\Omega)}\\
		&=p^{-2\frac{Nr}{d}}\|f\|^2_{A^r(\Omega)},
	\end{align*}
	which shows that $J$ is bounded, and it is clearly injective.
	For any $f_{N,K}\in A^r(\Omega_{N,K})$ the function
	\[
	f:=\sum_{K=0}^{p^N-1} f_{N,K} \,\one_{\Omega_{N,K}}
	\]
	belongs to $L^2(\Omega)$, and
	\begin{align*}
		\|f\|^2_{A^r(\Omega)}&=\sum_{n=0}^\infty p^{2\frac{nr}{d}}\|f-P_n\|^2_{L^2(\Omega)}\\
		&=\sum_{n=0}^{N-1} p^{2\frac{nr}{d}}\|f-P_n\|^2_{L^2(\Omega)}
		+p^{2\frac{Nr}{d}}\sum_{n=0}^{\infty} p^{2\frac{nr}{d}}\|f-P_{N+n}\|^2_{L^2(\Omega)}\\
		\text{use \eqref{fpn2}: }	&=
		\sum_{n=0}^{N-1} p^{2\frac{nr}{d}}\|f-P_n\|^2_{L^2(\Omega)}+
		p^{2\frac{Nr}{d}}\sum_{K=0}^{p^N-1} \|f_{N,K}\|^2_{A^r(\Omega_{N,K})}<\infty,
	\end{align*}
	which shows that $f\in A^r(\Omega)$. Therefore, $J$ is surjective as well, and it
	follows by the closed graph theorem that $J$ is an isomorphism.	
\end{proof}

\begin{lemma}\label{lem49}
Assume that (B4)  holds with $N=0$, i.e. there exists a local chart
$\Phi:\R^d\supset \Tilde O\mapsto O\subset S$
such that $\overline{\Omega}\subset O$ and the sets $\widetilde \Omega_{n,k}:=\Phi^{-1}(\Omega_{n,k})$
form a regular weakly balanced $p$-multiscale decomposition of $\widetilde \Omega:=\Phi^{-1}(\Omega)$. Consider the associated spaces $A^r(\Tilde\Omega)$, then for any $0\le r<1$ the map
\[
A^r(\Omega)\ni f\mapsto \Tilde f:=f\circ \Phi\in A^r(\widetilde \Omega)
\]
is an isomorphism.
\end{lemma}

\begin{proof}
		%
		Recall that there exists $\Tilde c_1>0$ such that 
		\begin{equation}
			\label{tc3}
			\diam \Tilde \Omega_{n,k} \leq \Tilde c_1 p^{-\frac{n}{d}}
			\text{  for all $n\in\N_0$ and $k\in\{0,\dots,p^n-1\}$.}
		\end{equation}
				
		For a function $f$ defined on $\Omega$ and the function $\Tilde f:=f\circ\Phi$
		defined on $\Tilde\Omega$ one has
		\begin{equation}
			\begin{aligned}
			\int_\Omega f(y)\dvol_g(y)&=\int_{\Tilde\Omega} \Tilde f(u) J_\Phi(u)\dd u,\\
			J_\Phi(u)&:=\sqrt{\det \Big(g_{\Phi}\big(\partial_j \Phi(u),\partial_k\Phi(u)\big)\Big)_{j,k\in\{1,\dots,d\}}},
			\end{aligned}
			\label{jf}
		\end{equation}
		and there exist $b_1,b_2>0$ such that $b_1\le J_\Phi(u)\le b_2$
		 for all $u\in \Tilde\Omega$. It follows that the map
		$\Theta: f\mapsto \Tilde f$
		defines an isomorphism between $L^2(\Omega)$ and $L^2(\Tilde\Omega)$,
		\[
		b_1\|\Theta f\|^2_{L^2(\Tilde\Omega)}\le \|\Theta f\|^2_{L^2(\Omega)}
		\le b_2\|\Theta f\|^2_{L^2(\Tilde\Omega)}
		\text{ for all } f\in L^2(\Omega),
		\]
		in particular,
		\begin{equation}
			\label{lbb}
			b_1|\Tilde\Omega_{n,k}|\le |\Omega_{n,k}|\le b_2|\Tilde\Omega_{n,k}| \text{ for all $(n,k)$.}
		\end{equation}

		For $n\in\N_0$ let $\Tilde P_n$ denote the orthogonal projector in $L^2(\Tilde\Omega)$ on the subspace
		\[
		\Tilde V_n:=\operatorname{span}\big\{\one_{\Tilde\Omega_{n,k}}: k\in\{0,\dots,p^n-1\} \big\},
		\]
then
		\begin{align*}
			P_n f&=\sum_{k=0}^{p^n-1}\dfrac{1}{|\Omega_{n,k}|}\int_{\Omega_{n,k}} f(y)\dvol_g(y),\nonumber\\
			\text{use \eqref{jf}: }
			\Theta P_n f(y)&\equiv (P_n f)\big(\Phi(y)\big)=
			\sum_{k=0}^{p^n-1}\dfrac{1}{|\Omega_{n,k}|}\int_{\Tilde \Omega_{n,k}} \Tilde f(u) J_\Phi(u)\dd u\,\one_{\Tilde \Omega_{n,k}}(y),\\
			\Tilde P_n\Theta f (y)&\equiv \Tilde P_n  \Tilde f(y)
			=\sum_{k=0}^{p^n-1}\dfrac{1}{|\Tilde \Omega_{n,k}|}\int_{\Tilde \Omega_{n,k}} 
			\Tilde f(u) \dd u\,\one_{\Tilde \Omega_{n,k}}(y).
		\end{align*}
		Therefore,
		\begin{equation}
			\label{ipf}
			\big(\Theta P_n f -\Tilde P_n\Theta f\big)(y)
			=
			\sum_{k=0}^{p^n-1}\int_{\Tilde \Omega_{n,k}} \Tilde f(u) \Big(  \dfrac{J_\Phi(u)}{|\Omega_{n,k}|} - \dfrac{1}{|\Tilde \Omega_{n,k}|}   \Big)\dd u\,\one_{\Tilde \Omega_{n,k}}(y).
		\end{equation}

		As $J_\Phi$ is a smooth function on a neighborhood of the closure of $\Tilde\Omega$, it is a Lipschitz function, and one finds some $a>0$
		with 
		\[
		|J_\Phi(u)-J_\Phi(u')|\le a|u-u'| \text{ for all } u,u'\in \Tilde \Omega.
		\]
		Pick any $u_{n,k}\in \Tilde\Omega_{n,k}$ and denote $J_{n,k}:=J_\Phi(u_{n,k})$, then for any $u\in\Tilde \Omega_{n,k}$
		there holds
		\begin{align*}
			|J_\Phi(u)-J_{n,k}|&\equiv |J_\Phi(u)-J_\Phi(u_{n,k})|\\
			&\le a|u-u_{n,k}|
			\le a\diam \Tilde \Omega_{n,k}\stackrel{\eqref{tc3}}{\le} \Tilde a p^{-\frac{n}{d}},
			\quad \Tilde a:=a\Tilde c_1.
		\end{align*}
		We have
		\begin{align*}
			\big||\Omega_{n,k}|-J_{n,k}|\Tilde\Omega_{n,k}|\big|&=\Big|\int_{\Tilde \Omega_{n,k}} \big(J_\Phi(y)- J_{n,k}\big)\dd y\Big|\\
			&	\le \int_{\Tilde \Omega_{n,k}} \big|J_\Phi(y)- J_{n,k}\big|\dd y\le \Tilde a p^{-\frac{n}{d}} |\Tilde\Omega_{n,k}|,
		\end{align*}
		and we can find $a_{n,k}\in[-\Tilde a,\Tilde a]$ such that $|\Omega_{n,k}|=\big(J_{n,k} +a_{n,k}p^{-\frac{n}{d}}\big) |\Tilde\Omega_{n,k}|$. Then for any $u\in \Tilde\Omega_{n,k}$
		we have
		\begin{align*}
			\Big|  \dfrac{J_\Phi(u)}{|\Omega_{n,k}|} - \dfrac{1}{|\Tilde \Omega_{n,k}|}   \Big| &=
			\Big|  \dfrac{J_\Phi(u)|\Tilde \Omega_{n,k}|- |\Omega_{n,k}|}{|\Omega_{n,k}| |\Tilde\Omega_{n,k}|}\Big|\\
			&\equiv \dfrac{\big| \big(
				J_\Phi(u)-J_{n,k}-a_{n,k}p^{-\frac{n}{d}}\big)
				\big|\, |\Tilde \Omega_{n,k}|}{|\Omega_{n,k}| |\Tilde\Omega_{n,k}|}\\
			&\le \dfrac{(\Tilde a +|a_{n,k}|)p^{-\frac{n}{d}}}{|\Omega_{n,k}|}\stackrel{\eqref{lbb}}{\le} \dfrac{b  p^{-\frac{n}{d}}}{|\Tilde\Omega_{n,k}|} \text{ with } b:=\dfrac{2\Tilde a}{b_1}.
		\end{align*}
		
		It follows from \eqref{ipf} that 
		\[
		\Big|(\Theta P_n f -\Tilde P_n\Theta f)(y)\Big|
		\le
		\sum_{k=0}^{p^n-1}\dfrac{b  p^{-\frac{n}{d}}}{|\Tilde\Omega_{n,k}|}\int_{\Tilde \Omega_{n,k}} \big|\Tilde f(u)\big|\dd u\,\one_{\Tilde \Omega_{n,k}}(y)
		\equiv b  p^{-\frac{n}{d}} \big(\Tilde P_n \Theta |f|\big)(y).
		\]
It follows that
		\begin{equation}
			\label{pnf}
			\begin{aligned}
				\|\Theta P_n f -\Tilde P_n\Theta f\|^2_{L^2(\Tilde \Omega)}&
				\le b^2  p^{-\frac{2n}{d}}\big\|\Tilde P_n\Theta |f|\big\|^2_{L^2(\Tilde \Omega)}\\
				&\le b^2  p^{-\frac{2n}{d}}\big\|\Theta |f|\big\|^2_{L^2(\Tilde \Omega)}
				\stackrel{\eqref{lbb}}{\le} b^2 b_1^{-1}p^{-\frac{2n}{d}}\|f\|^2_{L^2(\Omega)}.
			\end{aligned}
		\end{equation}
		
		Now let $f\in A^r(\Omega)$, then
		\begin{align*}
			\|\Tilde f-\Tilde P_n\Tilde f\|^2_{L^2(\Tilde \Omega)}&\equiv 
			\|\Theta f-\Tilde P_n\Theta f\|^2_{L^2(\Tilde \Omega)}\\
			&=\big\| (\Theta f-\Theta P_n f) + (\Theta P_n f-\Tilde P_n\Theta f)\big\|^2_{L^2(\Tilde \Omega)}\\
			&\le
			2\|\Theta (f-P_n f)\|^2_{L^2(\Tilde \Omega)} + 2\|\Theta P_n f-\Tilde P_n\Theta f\|^2_{L^2(\Tilde \Omega)}\\
			\text{use \eqref{lbb} and \eqref{pnf}: }	&\le 2b_1^{-1}\|f-P_n f\|_{L^2(\Omega)}+2b^2 b_1^{-1}p^{-\frac{2n}{d}}\|f\|_{L^2(\Omega)},\\[\smallskipamount]
			\|\Tilde P_0 \Tilde f\|^2_{L^2(\Tilde \Omega)}&\equiv \|\Tilde P_0 \Theta f\|^2_{L^2(\Tilde \Omega)}
			=\|\Theta P_0 f-(\Theta P_0 f -\Tilde P_0\Theta f)\|^2_{L^2(\Tilde\Omega)}	\\
			&\le 2\|\Theta P_0 f\|^2_{L^2(\Tilde\Omega)}+2\|\Theta P_0 f -\Tilde P_0\Theta f\|^2_{L^2(\Tilde \Omega)}\\
			\text{use \eqref{lbb} and \eqref{pnf}: }	& \le 2b_1^{-1}\|P_0 f\|_{L^2(\Omega)}+2b^2 b_1^{-1}\|f\|_{L^2(\Omega)}.
		\end{align*}
		
Therefore,
		\begin{align*}
			\|\Tilde f\|^2_{A^r(\Tilde \Omega)}&=\|\Tilde P_0 \Tilde f\|^2_{L^2(\Tilde \Omega)}
			+\sum_{n=0}^\infty p^{2\frac{nr}{d}} \|\Tilde f-\Tilde P_n\Tilde f\|^2_{L^2(\Tilde \Omega)}\\
			&\le 2b_1^{-1}\Big( \underbrace{
				\|P_0 \Tilde f\|^2_{L^2(\Omega)}
				+\sum_{n=0}^\infty p^{2\frac{nr}{d}} \| f-\Tilde P_n f\|^2_{L^2(\Omega)}
			}_{\equiv \|f\|^2_{A^r(\Omega)}}\Big)\\
		&\quad +2b^2 b_1^{-1}\Big( 1+ \underbrace{\sum_{n=0}^\infty p^{\frac{2n}{d}(r-1)}}_{=:S\,<\infty}\Big) \underbrace{\|f\|^2_{L^2(\Omega)}}_{\le \|f\|^2_{A^r(\Omega)}}\\
			&\le \big(2b_1^{-1} + 2b^2 b_1^{-1}(1+S)\big)\|f\|^2_{A^r(\Omega)},
		\end{align*}
		which shows that $\Tilde f\in A^r(\Tilde \Omega)$ and that $\Theta$ defines a bounded operator $A^r(\Omega)\to A^r(\Tilde\Omega)$, and remark that it is injective by construction.
		
		On the other hand, let $\Tilde f\in A^r(\Tilde\Omega)$
		and $f\in L^2(\Omega)$ with $\Tilde f=\Theta f$. Then
		\begin{align*}
			\|f-P_n f\|^2_{L^2(\Omega)}&\stackrel{\eqref{lbb}}{\le} b_2\|\Theta f-\Theta P_n f\|^2_{L^2(\Tilde\Omega)}\\
			&=b_2\big\|(\Theta f-\Tilde P_n\Theta f)+(\Tilde P_n\Theta f-\Theta P_n f)\big\|^2_{L^2(\Tilde\Omega)}\\
			\text{use $\Theta f=\Tilde f$: }	&\le 2b_2 \|\Tilde f - \Tilde P_n \Tilde f\|^2_{L^2(\Tilde\Omega)}
			+2b_2 \|\Tilde P_n\Theta f-\Theta P_n f\|^2_{L^2(\Tilde\Omega)}\\
			\text{use \eqref{pnf}: }
			&\le 2b_2 \|\Tilde f - \Tilde P_n \Tilde f\|^2_{L^2(\Tilde\Omega)}
			+2b_2 b^2 b_1^{-1}p^{-\frac{2n}{d}}\|f\|^2_{L^2(\Omega)}\\
			\text{use \eqref{lbb}: } &
			\le 2b_2 \|\Tilde f - \Tilde P_n \Tilde f\|^2_{L^2(\Tilde\Omega)}
			+2b^2_2 b^2 b_1^{-1}p^{-\frac{2n}{d}}\|\Tilde f\|^2_{L^2(\Tilde \Omega)},\\[\smallskipamount]
			\|P_0 f\|^2_{L^2(\Omega)}&
			\stackrel{\eqref{lbb}}{\le} b_2\|\Theta P_0 f\|^2_{L^2(\Tilde\Omega)}\\
			&=b_2\big\| \Tilde P_0 \Theta f - (\Tilde P_n\underbrace \Theta f-\Theta P_0 f)\big\|^2_{L^2(\Tilde\Omega)}\\
			\text{use $\Theta f=\Tilde f$:  }  &\le 2b_2 \|\Tilde P_0 \Tilde f\|^2_{L^2(\Tilde\Omega)}
			+2b_2 \|\Tilde P_n\Theta f-\Theta P_n f\|^2_{L^2(\Tilde\Omega)}\\
			\text{use \eqref{pnf}: }
			&\le 2b_2 \| \Tilde P_0 \Tilde f\|^2_{L^2(\Tilde\Omega)}
			+2b_2 b^2 b_1^{-1}\|f\|^2_{L^2(\Omega)}\\
			\text{use \eqref{lbb}: }  & 
			\le 2b_2 \|\Tilde P_0 \Tilde f\|^2_{L^2(\Tilde\Omega)}
			+2b^2_2 b^2 b_1^{-1}\|\Tilde f\|^2_{L^2(\Omega)}.
		\end{align*}
		Therefore,
		\begin{align*}
			\|f\|^2_{A^r(\Omega)}&=\|P_0 f\|^2_{L^2(\Tilde \Omega)}
			+\sum_{n=0}^\infty p^{2\frac{nr}{d}} \|f- P_n f\|^2_{L^2(\Tilde \Omega)}\\
			&\le 2b_2\Big( \underbrace{
				\|P_0 \Tilde f\|^2_{L^2(\Omega)}
				+\sum_{n=0}^\infty p^{2\frac{nr}{d}} \| \Tilde f-\Tilde P_n \tilde f\|^2_{L^2(\Omega)}
			}_{\equiv \|\Tilde f\|^2_{A^r(\Tilde \Omega)}}\Big)\\
			&\quad +2b^2_2 b^2 b_1^{-1}\Big( 1+ \underbrace{\sum_{n=0}^\infty p^{\frac{2n}{d}(r-1)}}_{=:S\,<\infty}\Big) \underbrace{\|\Tilde f\|^2_{L^2(\Tilde \Omega)}}_{\le \|\Tilde f\|^2_{A^r(\Tilde \Omega)}}\\
			&\le \big(2b_2 + b^2_2 b^2 b_1^{-1} (1+S)\big)\|\Tilde f\|^2_{A^r(\Tilde \Omega)},
		\end{align*}
		which shows that $f\in A^r(\Omega)$, and, consequently, that $\Theta^{-1}:A^r(\Tilde \Omega)\to A^r(\Omega)$ is everywhere defined and bounded.	
	\end{proof}

Now we can transfer the relations between $A^r$ and $H^s$ known for the Euclidian case to the case of manifolds.

\begin{theorem}\label{thmappr}
Assume \eqref{eq-omega2}. Then for all $r \geq 0$ and $0<s<1$ there holds
\begin{align}
	A^r(\Omega)&\hookrightarrow H^{s}(\Omega) \text{ if } s<\frac{1}{2} \text{ and }
	s\le r, \label{emb1a}\\
	H^{s}(\Omega)&\hookrightarrow A^{r}(\Omega) \text{ if } 0\leq r\leq s<1. \label{emb2a}
\end{align}
In particular,
\[
A^r(\Omega)=H^{r}(\Omega) \text{ for any } 0\le r<\frac{1}{2}.
\]
\end{theorem}

\begin{proof}
As $A^r(\Omega)\hookrightarrow A^{r'}(\Omega)$ for $r\ge r'$, it is sufficient to prove \eqref{emb1a} under the additional assumption $r<1$.
We first use the map
\[
A^r(\Omega)\ni f\mapsto (f_{N,K})_{K\in\{0,\dots,p^N-1\}}\in \bigoplus_{K=0}^{p^N-1} A^r(\Omega_{N,K}),
\quad
f_{N,K}:=f|_{\Omega_{N,K}},
\]
which is an isomorphism by Lemma \ref{lem48}.
For each $K$, the map
\[
A^r(\Omega_{N,K})\ni f_{N,K}\mapsto \Tilde f_{N,K}:=f_{N,K}\circ \Phi_{N,K}\in A^r(\Tilde\Omega_{N,K})
\]
is also an isomorphism by Lemma \ref{lem49},
and $A^r (\Tilde\Omega_{N,K})\ni \Tilde f_{N,K}\to \Tilde f_{N,K}\in H^s(\Tilde\Omega_{N,K})$
is an embedding by Theorem \ref{thmemb2}. In addition,
\[
H^s(\Tilde\Omega_{N,K})\ni \Tilde f_{N,K}\mapsto \Tilde f_{N,K}\circ \Phi_{N,K}^{-1}\equiv f_{N,K}\in H^s(\Omega_{N,K})
\]
is an isomorphism due to the construction of Sobolev spaces. Therefore, we have shown that
\[
A^r(\Omega)\ni f\mapsto (f_{N,K})_{K\in\{0,\dots,p^N-1\}}\in \bigoplus_{K=0}^{p^N-1} H^s(\Omega_{N,K})
\]
is an embedding. We now recall that due to $s<\frac{1}{2}$ the subspaces $C^\infty_c(\Tilde \Omega_{N,K})$ are dense in $H^s(\Tilde \Omega_{N,K})$, which in turn means that $C^\infty_c(\Omega_{N,K})$ are dense in $H^s(\Omega_{N,K})$, therefore, the operator $J_{N,K}$ of extension by zero from $\Omega_{N,K}$ to $\Omega$
extends by density from $C^\infty_c(\Omega_{N,K})$ to an embedding
$J_{N,K}:H^s(\Omega_{N,K})\to H^s(\Omega)$. Then
\[
J:\bigoplus_{K=0}^{p^N-1} H^s(\Omega_{N,K})\ni (\varphi_{N,K})_{K\in\{0,\dots,p^N-1\}}\mapsto
\sum_{k=0}^{p^N-1}J_{N,K}\varphi_{N,K}\in H^s(\Omega)
\]
is an embedding, which finishes the proof of \eqref{emb1a}.

For \eqref{emb2a} we consider the following maps:
\begin{align*}
	H^s(\Omega)\ni f\mapsto (f_{N,K})_{K\in\{0,\dots,p^N-1\}}&\in \bigoplus_{K=0}^{p^N-1} H^s(\Omega_{N,K})\\
	& \text{ with } f_{N,K}:=f|_{\Omega_{N,K}},  &&\text{(a)}\\
    H^s(\Omega_{N,K})\ni f_{N,K}\mapsto \Tilde f_{N,K}:=f_{N,K}\circ \Phi_{N,K}&\in H^s(\Tilde\Omega_{N,K}), && \text{(b)}\\
    H^s(\Tilde \Omega_{N,K})\ni \Tilde f_{N,K}\mapsto \Tilde f_{N,K}&\in A^r(\Tilde \Omega_{N,K}),
    && \text{(c)}\\
    A^r(\Tilde \Omega_{N,K})\ni \Tilde f_{N,K}\mapsto \Tilde f_{N,K}\circ \Phi_{N,K}^{-1}\equiv f_{N,K}&\in A^r(\Omega_{N,K}), && \text{(d)}\\
    \bigoplus_{K=0}^{p^N-1} A^r(\Omega_{N,K})\ni (f_{N,K})_{K\in\{0,\dots,p^N-1\}}\mapsto f&\in A^r(\Omega). && \text{(e)} 
\end{align*}
The map (a) is an embedding due to the definition of Sobolev spaces (in fact, even as an isomorphism due to the first half of the proof), (b) is an isomorphism due to the definition of Sobolev spaces, (c) is an embedding by Theorem \ref{thmemb1}, (d) is an isomorphism by Lemma \ref{lem49}, and (e) is an ismorphism by Lemma \ref{lem48}.
Taking the composition one arrives at the conclusion.
\end{proof}

We discuss the existence of multiscale decompositions and the additional condition (B7)
for some classes $\Omega$ in Section \ref{sec52}.

\section{Embedded traces}\label{sec4}

\subsection{Abstract trace space as an approximation space}\label{sec41}

Let $\Omega$ be an open set in $\R^d$ (as in Subsection \ref{ssmult}) or in a $d$-dimensional manifold (as in Subsection \ref{sec44}) admitting a $p$-multiscale decomposition $(\Omega_{n,k})$. We introduce an operator of identification $I_\Omega$ between the functions defined on $\cZ$ (see Subsection \ref{sec-abstrace}) and the functions defined on $\Omega$ as follows. First, for each $z\in \cZ$ consider the basis sequences
\[
e_z:=(\delta_{z,\zeta})_{\zeta\in \cZ}.
\]
Then we consider the linear map
\begin{gather}
I_\Omega: \ \operatorname{span}\{e_z:\, z\in\cZ\}\to \operatorname{span}\{\one_{\Omega_{n,k}}:\ n\in\N_0,\ k=0,\dots,p^n-1\},\nonumber \\
I_\Omega: e_z\mapsto \begin{cases}
	\one_\Omega, & z=\rad,\\[\bigskipamount]
	\displaystyle	p^{\frac{n}{2}}\,\sum_{j=0}^{p-1} \theta_s^j \one_{\Omega_{n+1,pk+j}},& z=(n,k,s).
\end{cases}
  \label{eq-iom}
\end{gather}

\begin{prop}[Euclidean case]\label{proptrace1}
Let $\Omega\subset\R^d$	be a bounded open set with Lipschitz boundary and the decomposition $(\Omega_{n,k})$ be regular and \underline{strongly} balanced.
Then for any $r\ge 0$ the map $I_\Omega$ extends by continuity to an isomorphism between $\ell^2_r(\Omega)$ and $A^{rd}(\Omega)$.
\end{prop}

\begin{proof}
(i) The linear span of $e_\zeta$ is dense in $\ell^2_r(\cZ)$,
and $\langle e_z,e_\zeta\rangle_{\ell^2_r(\cZ)}=p^{2r\nu(z)}\delta_{z,\zeta}$ for all $z,\zeta\in \cZ$.

(ii) Now remark that $I_\Omega e_\rad\in V_0$ and $I_\Omega e_{n,k,s}\in V_{n+1}$. At the same time (using the fact that the decomposition is strongly balanced), 
\[
\int_{\Omega_{n,k}}I_\Omega e_{n,k,s} (x)\dd x=p^{\frac{n}{2}}\sum_{j=0}^{p-1} \theta_s^j |\Omega_{n+1,pk+j}|=
p^{\frac{n}{2}}\frac{|\Omega_{n,k}|}{p}\sum_{j=0}^{p-1} \theta_s^j=0,
\]
and for any $k_0\ne k$ one has
\[
\int_{\Omega_{n,k_0}}I_\Omega e_{n,k,s} (x)\dd x=0
\]
as $e_{n,k,s}$ vanishes identically in $\Omega_{n,k_0}$. This shows that $I_\Omega e_{n,k,s}$
is orthogonal in $L^2(\Omega)$ to all $\one_{\Omega_{n,k_0}}$, $k_0\in\{0,\dots,p-1\}$,
in other words, $I_\Omega e_{n,k,s}\perp V_n$. Therefore, we have shown that
$I_\Omega e_\rad\in U_0$ and $I_\Omega e_{n,k,s}\in U_{n+1}$ or, in other words,
\begin{equation}
	\label{uu00}
I_\Omega e_z\in U_{\nu(z)+1} \text{ for all } z\in \cZ.
\end{equation}
As the subspaces $U_j$ are mutually orthogonal in $L^2(\Omega)$, one has $\langle I_\Omega e_z,I_\Omega e_\zeta\rangle_{L^2(\Omega)}=0$ for $\nu(z)\ne\nu(\zeta)$. In addition, if $z=(n,k_1,s_1)$ and $\zeta=(n,k_2,s_2)$ with $k_1\ne k_2$, then $I_\Omega e_z$ and $I_\Omega e_\zeta$ have disjoint supports (contained in the disjoint sets $\Omega_{n,k_1}$ and $\Omega_{n,k_2}$), so one has again $\langle I_\Omega e_z,I_\Omega e_\zeta\rangle_{L^2(\Omega)}=0$.
Finally,
\begin{align*}
\langle I_\Omega e_{n,k,s},	I_\Omega e_{n,k,s'}\rangle_{L^2(\Omega)}&=
p^n\sum_{j=0}^{p-1} \theta_s^j \overline{\theta_{s'}^j} |\Omega_{n+1,pk+j}|\\
&=p^n\dfrac{|\Omega_{n,k}|}{p}\sum_{j=0}^{p-1} \Big(\dfrac{\theta_s}{\theta_{s'}}\Big)^j=p^n|\Omega_{n,k}|\delta_{s,s'}.
\end{align*}
Alltogether we obtain
\begin{equation}
   \label{zzz}
\langle  I_\Omega e_z,I_\Omega e_\zeta\rangle_{L^2(\Omega)}=|\Omega|\,\delta_{z,\zeta}, \quad
z,\zeta\in\cZ.
\end{equation}

(iii) We will equip $A^{rd}(\Omega)$	with the norm
\[
\|f\|^2_{A,rd}:=\sum_{n=0}^\infty p^{2rn}\|Q_n f\|^2_{L^2(\Omega)},
\]
see Lemma \ref{lem35}.
For $N\in\N$ we denote $W_N:=\operatorname{span}\{e_z: \nu(z)\le N-1\}\subset \ell^2_r(\cZ)$.
Let $f\in W_N$, then
\begin{gather*}
f=\sum_{\nu(z)\le N-1} f_z e_z, \quad f_z\in \C,\\
I_\Omega f=\sum_{n=0}^{N} F_n, \quad F_n:=\sum_{\nu(z)=n-1}
f_z I_\Omega e_z.
\end{gather*}
Due to \eqref{uu00} one has $Q_n I_\Omega f=F_n$ for all $n\le N$ and $Q_n I_\Omega f=0$ for $n\ge N+1$, therefore,
\begin{align*}
	\|f\|^2_{A,rd}&=\sum_{n=0}^{N} p^{2rn} \|F_n\|^2_{L^2(\Omega)}\stackrel{\eqref{zzz}}{=}
\sum_{n=0}^{N} p^{2rn} |\Omega|\sum_{\nu(z)=n-1} |f_z|^2\\
&=p^{2r}|\Omega|\sum_{z\in \cZ} p^{2 r\nu(z)} |f_z|^2\equiv p^{2r}|\Omega| \big\| (f_z)\big\|^2_{\ell^2_r(\cZ)},
\end{align*}	
which shows that $I_\Omega$ is an isometry (up to a constant factor), in particular,
it is bounded and extends an isometry between of $\ell^2_r(\cZ)$ and some closed subspace $\ran I_\Omega\subset A^{rd}(\Omega)$.

(iv) It remains to show that $\ran I_\Omega=A^{rd}(\Omega)$. Remark that by construction we have $\dim W_N=p^{n}$. At the same time, $I_\Omega W_N\subset V_{N}$, so we obtain
\[
p^{N}=\dim W_N=\dim I_\Omega (W_N)\le \dim V_{N}=p^{N},
\]
which shows that $I_\Omega (W_N)= V_{N}$ for any $N$. As $N$ can be arbitrarily large, $\ran I_\Omega$ contains any finite linear combination of $\one_{\Omega_{n,k}}$.
As these linear combinations span a dense subset of $A^{rd}(\Omega)$ and $\ran I_\Omega$ is closed, we have $\ran I_\Omega=A^{rd}(\Omega)$.
\end{proof}
%
%

\begin{prop}[Manifold case]\label{proptrace2}
Let $\Omega$ be an open set with compact closure in a manifold of bounded geometry 
and the decomposition $(\Omega_{n,k})$ be regular and \underline{strongly} balanced.
Let $0\le s<1$ with $s\le rd$, then the map $I_\Omega$ extends by continuity to an
embedding $\ell^2_r(\Omega) \hookrightarrow A^{s}(\Omega)$. For  $s=rd<1$ this embedding is an isomorphism.
\end{prop}

\begin{proof}
Let $N$, $\Phi_{N,K}$ and $\Tilde \Omega_{n,k}$ be as in (B4)--(B5).

(i) If $N=0$, then the map $\Tilde I_\Omega: \xi\mapsto (I_\Omega \xi)\circ \Phi_{0,0}$ is covered by Proposition \ref{proptrace1} and defines an isomorphism between $\ell^2_r(\cZ)$ and $A^{rd}(\Tilde \Omega)$, for $\Tilde\Omega:=\Tilde \Omega_{0,0}\equiv \Phi_{0,0}^{-1}(\Omega)$.

If $rd<1$, then it follows by Lemma \ref{lem49} that  $I_\Omega$ is an isomorphism between $\ell^2_r(\cZ)$ and $A^{rd}(\Omega)$, and $A^{rd}(\Omega)\hookrightarrow A^s(\Omega)$ for $0\le s\le rd$.

If $rd\ge 1$, using $A^{rd}(\Tilde \Omega)\hookrightarrow A^s(\Tilde\Omega)$ we obtain $\Tilde I_\Omega: \ell^2_r(\cZ)\hookrightarrow A^s(\Tilde\Omega)$,
and Lemma \ref{lem49} gives $I_\Omega: \ell^2_r(\cZ)\hookrightarrow A^s(\Omega)$.

(ii) Now assume that $N\ge 1$ and consider
\begin{align*}
W_N&:=\operatorname{span}\{ e_z:\, \nu(z)\le N-1\},\\	
O_N&:=\operatorname{span}\{\one_{\Omega_{N,K}}:\, K=0,\dots,p^N-1\},
\end{align*}
then by construction one has $I_\Omega (W_N)\subset O_N$.
We will equip $A^{rd}(\Omega)$	with the norm
\[
\|f\|^2_{A,rd}:=\sum_{n=0}^\infty p^{2rn}\|Q_n f\|^2_{L^2(\Omega)},
\]
see Proposition \ref{PQod}. The computations (i)-(iii) in the proof of Proposition \ref{proptrace1} show that
\[
\|I_\Omega f\|^2_{A,rd}=p^{2r}|\Omega| \big\| (f_z)\big\|^2_{\ell^2_r(\cZ)}
\text{ for any $f\in W_N$.}
\]
in particular, $I_\Omega:W_N\to O_N$ is injective.
As both $W_N$ and $O_N$ have the same dimension $p^N$, the map $I_\Omega:W_N\to O_N$ is a linear isomorphism, and one can find a basis $b^0,\dots,b^{p^N-1}$ in $W_N$ such that
\[
I_\Omega b^K=p^{\frac{N}{2}} \one_{\Omega_{N,K}} \text{ for each } K\in\{0,\dots,p^N-1\}.
\]

(iii) For $\xi\in \ell^2_r(\cZ)$ define $\xi^N\in W_N$ by
\[
\xi^N_z:=\begin{cases}
	\xi_z, & \nu(z)\le N-1,\\
	0, & \text{otherwise,}
	\end{cases}
\]
and let $\gamma_K(\xi)$, $K\in\{0,\dots,p^N-1\}$, be the coordinates of $\xi^N$ in the basis $(b^K)$. Now we consider the map
\begin{gather*}
\Psi:\ \ell^2_r(\cZ)\ni (\xi_z)\mapsto \eta\equiv (\eta^0,\dots,\eta^{p^N-1})\in\bigoplus_{K=0}^{p^N-1} \ell^2_r(\cZ),\\
\eta^K_z=\begin{cases}
	\gamma_K(\xi), & z=\rad,\\
	\xi_{N+n,p^n K+k,s}, & z=(n,k,s),
\end{cases}
\quad K\in \{0,\dots,p^N-1\}.
\end{gather*}
By construction $\Psi$ is an isomorphism, with $\Psi^{-1}$ given by
\[
\big(\Psi^{-1}(\eta^0,\dots,\eta^{p^N-1})\big)_z=
\begin{cases}
	\displaystyle \sum_{K=0}^{p^N-1} \eta^K_\rad b^K_z, & \nu(z)\le N-1,\\
\eta^K_{n-N,k-p^{n-N}K,s}, & z=(n,k,s),\ n\ge N,\\
& p^{n-N}K\le k<p^{n-N}(K+1).
\end{cases}
\]
One computes
\[
\big(\Psi^{-1}(\underbrace{0,\dots,0}_{K-1 \text{ times}},e_\lambda,0,\dots,0)\big)_z
=\begin{cases}
	b^K_z,& \lambda=\rad,\ \nu(z)\le N-1,\\
	1, & \lambda=(n,k,s),\ z=(n+N,p^{n} K+k,s),\\
	0, & \text{otherwise},
\end{cases}
\]
or, equivalently,
\[
\Psi^{-1}(\underbrace{0,\dots,0}_{K-1 \text{ times}},e_\lambda,0,\dots,0)=\begin{cases}
	b^K, & \lambda=\rad,\\
	e_{(n+N,p^{n}K+k,s)}, & \lambda=(n,k,s).
\end{cases}
\]
Due to the definition of $I_\Omega$ one has then
\begin{align*}
I_\Omega\Psi^{-1}(\underbrace{0,\dots,0}_{K-1 \text{ times}},e_\lambda,0,\dots,0)
&=\begin{cases}
	p^{\frac{N}{2}} \one_{\Omega_{N,K}}, & \lambda=\rad,\\
	\displaystyle p^{\frac{N+n}{2}}\sum_{j=0}^{p-1} \theta_s^j \one_{\Omega_{N+n+1, p^{n+1}K+pk+j}}
\end{cases}\\
	&
	=
	p^{\frac{N}{2}}J_{N,K} I_{\Omega_{N,K}} e_\lambda,
\end{align*}
where
\begin{itemize} 
	\item $J_{N,K}$ is the operator of extension by zero from $\Omega_{N,K}$ to $\Omega$,
	\item $I_{\Omega_{N,K}}:\ell^2_r(\cZ)\to A^{rd}(\Omega_{N,K})$ is the identification operator for the decomposition $(\Omega_{N+k,p^n K+k})_{n,k}$, which is already covered by (i). 
\end{itemize}	

(iv) The above computations show that $I_\Omega \Psi^{-1}$ acts as
\[
I_\Omega \Psi^{-1}(\eta^0,\dots,\eta^{p^N-1})=p^{\frac{N}{2}}\sum_{K=0}^{p^K-1}J_{N,K} I_{\Omega_{N,K}} \eta^K.
\]
By (i), each $I_{\Omega_{N,K}}$ is an isomorphism (for $s=rd<1$) or an embedding (for all other cases), and it follows by Lemma \ref{lem48} that
\[
I_\Omega \Psi^{-1}: \bigoplus_{K=0}^{p^N-1} \ell^2_r(\cZ)\to A^{rd}(\Omega)
\]
is an isomorphism (for $s=rd<1$) or an embedding (for all other cases),
and then $I_\Omega\equiv (I_\Omega \Psi^{-1})  \Psi$ preserves the same properties.
\end{proof}

\subsection{Embedded trace operator}\label{sec42}
For all assertions in this subsection, let $\Omega$ be an open set with compact closure in a $d$-dimensional manifold of bounded geometry $S$ admitting a regular strongly balanced $p$-multiscale decomposition $(\Omega_{n,k})$ as described in Subsections \ref{ssmult} and \ref{sec44}.

Recall (Theorem \ref{th28}) that we have constructed an abstract
trace operator
\[
\tau:\, H^1(\T)\to\ell^2_\sigma(\cZ)
\]
with 
\begin{equation}
	\label{eq-sigma}
\sigma:=\dfrac{1}{2\log p}\log\dfrac{\alpha p}{\ell}\equiv \frac{1}{2}\Big ( 1-\frac{\log\ell-\log\alpha}{\log p}\Big)>0,
\end{equation}
which is bounded and surjective with $\ker\tau=H^1_0(\T)$.
We recall that $p\in \N$
with $p\ge 2$ and that the parameters $\alpha$ and $\ell$ satisfy 
\begin{equation}
	\label{lap2}
0<\ell<1, \quad 	\ell<\alpha p<\dfrac{1}{\ell},
\end{equation}
see Lemma \ref{onb-delta}. We define the identification/embedding operator
\[
I_\Omega:\ell^2_\sigma(\cZ)\to A^{s}(\Omega)
\]
as in Propositions \ref{proptrace1} and \ref{proptrace2}.
This gives rise to the \emph{embedded trace operator}
\[
\gamma_\Omega:=I_\Omega \tau:H^1(\T)\to A^{s}(\Omega),
\]
with the following options for $s$:
\begin{itemize}
	\item if $\Omega$ is a $d$-dimensional Euclidean open set (as in Subsection \ref{ssmult}), then $\gamma_\Omega$ is a bounded linear operator for any $0\le s\le \sigma d$, surjective for $s=\sigma d$,
	\item if $\Omega$ is an open set in $d$-dimensional manifold (as in Subsection \ref{sec44}), then $\gamma_\Omega$ is a bounded
	linear operator for any $0\le s< 1$ such that $s\le \sigma d$, surjective for $s=\sigma d<1$.
\end{itemize}
In all these cases one has by construction
\[
\ker \gamma_\Omega=\ker\tau= H^1_0(\T).
\]

In addition, using the identification between the approximation and Sobolev spaces (Theorem \ref{thmappr}) we obtain:
\begin{itemize}
	\item if $\Omega$ is an open set in $d$-dimensional manifold, then
	\[
	\gamma_\Omega: H^1(\T)\hookrightarrow H^s(\Omega)
	\text{ for and $0\le s<\frac{1}{2}$ with $s\le \sigma d$.}	
	\]
	In particular,
	\begin{equation}
		\label{tracesob}
	\gamma_\Omega\big(H^1(\T)\big)=H^{\sigma d}(\Omega) \text{ if } \sigma d<\frac{1}{2}.
	\end{equation}
\end{itemize}

\begin{remark}
	It is useful to check that the condition $\sigma d<\frac{1}{2}$ in \eqref{tracesob} can really be satisfied under the restrictions \eqref{eq-sigma} and \eqref{lap2}. In view of \eqref{eq-sigma} the condition can be rewritten as
	\[
	\dfrac{d}{\log p}\log\dfrac{\alpha p}{\ell}<1
	\quad
	\text{i.e.}\quad
	\dfrac{\alpha p}{\ell}<p^{\frac{1}{d}},
	\]
	so together with \eqref{lap2} we arrive at
	\[
	0<\ell<1, \quad
	1<\dfrac{\alpha p}{\ell}<\min\big\{ \dfrac{1}{\ell^2}, p^\frac{1}{d}\big\}.
	\]
Therefore, if one fixes arbitrary $\ell\in(0,1)$	and $p\in \N$ with $p\ge 2$, the required condition is satisfied
for 
\[
\dfrac{\ell}{p}<\alpha<\dfrac{\ell}{p} \, \min\big\{ \dfrac{1}{\ell^2}, p^\frac{1}{d}\big\},
\]
i.e. for a non-trivial range of $\alpha$.
\end{remark}

Finally, we give a more illustrative description of the embedded trace operator, which uses more classical terms:
\begin{theorem}[Embedded trace using limit values]\label{thm-geom}
Let $0\le s <\frac{1}{2}$ with $s\le \sigma d$, then for any $f\in H^1(\T)$ there holds
\[
\gamma_\Omega f =\lim_{N\to \infty} \sum_{K=0}^{p^N-1} f(X_{N,K})\,\one_{\Omega_{N,K}},
\]
where the limit is taken in $H^s(\Omega)$.
\end{theorem}

\begin{proof}
(i) Let $f\in H^1(\T)$. For $N\in \N$ consider $f_N:\T\to\C$ defined by
	\[
	f_N(x):= \begin{cases}
		f(x), & x\in \T^N,\\
		f(X_{N,k}), & x \in \T_{N,K}, \ K\in\{0,\dots,p^N-1\}.
	\end{cases}
	\]	
In Lemma	\ref{proplim} we have shown that $f_N\in H^1(\T)$ with
$f_N\xrightarrow{N\to \infty}f$ in $H^1(\T)$. Due to the boundedness of $\gamma_\Omega$
we have $\gamma_\Omega f_N\xrightarrow{N\to\infty} \gamma_\Omega f$ in $H^s(\Omega)$. Therefore, it is sufficient to show that for any $N\in\N$ one has
\begin{equation}
	\label{eq-temp1}
\gamma_\Omega f_N=\sum_{K=0}^{p^N-1} f(X_{N,K})\,\one_{\Omega_{N,K}}.
\end{equation}

(ii) Let $N\in\N$ be fixed. Pick a function $F\in H^1\big((0,L),q(t)\dd t\big)$ such that $F(t)=0$ for $t\le t_{N}$ and $F(t)=1$ for $t\ge t_{N+1}$. Define $\varphi:=U_\rad F\in H^1(\T)$, i.e. $\varphi:\T\ni x\mapsto F(|x|)$, then $\varphi \one_{\T_{N,K}}\in H^1(\T)$ 
for any $K\in\{0,\dots,p^N-1\}$.  We will show that
\begin{equation}
   \label{eq-temp2}
	\gamma_\Omega (\varphi \one_{\T_{N,K}})=\one_{\Omega_{N,K}} \text{ for any } K\in\{0,\dots,p^N-1\}.
\end{equation}
In fact, if \eqref{eq-temp2} is proved, then \eqref{eq-temp1} follows directly:
one has
\[
f_N=\sum_{K=0}^{p^N-1}f(X_{N,K})\varphi\,\one_{\T_{N,K}} \text{ in }\T\setminus\T^{N+1},
\]
which implies
\[
\gamma_\Omega f_N =\sum_{K=0}^{p^N-1}f(X_{N,K})\gamma_\Omega\big(\varphi \one_{\T_{N,K}}\big)\stackrel{\eqref{eq-temp2}}{=}\sum_{K=0}^{p^N-1}f(X_{N,K}) \one_{\Omega_{N,K}}.
\]

(iii) It remains to prove \eqref{eq-temp2}. Consider
\[
S_N:=\operatorname{span}\{\varphi 1_{\T_{N,K}}:\, K=0,\dots,p^N-1\}\subset H^1(\T).
\]
The functions $\varphi 1_{\T_{N,K}}$ form a basis of $S_N$, so $\dim S_N=p^N$.
Now we remark that that for any $z\in\cZ$ one has the inclusion $\dim S\cap H^1_z(\T)\subset \C U_z F$, and
\[
\dim \big(S\cap H^1_\rad(\T)\big)=1,
\quad
\dim \big(S\cap H^1_{n,k,s}(\T)\big)=\begin{cases}
	1, &n\le N-1,\\
	0, &n\ge N.
\end{cases}
\]
Due to the orthogonal decomposition $H^1(\T)=\bigoplus_{z\in\cZ} H^1_z(\T)$ we conclude
that
\[
S_N=\operatorname{span}\{U_z F:\, z\in\cZ,\, \nu(z)\le N-1\},
\]
and the functions $U_z F$ with $\nu(z)\le N-1$ form a basis in $S_N$.
Recall that by Lemma \ref{prop-unique} we have
$\tau U_z F=p^{-\frac{\nu(z)}{2}}e_z$ for $\nu(z)\le N-1$, and then, using \eqref{eq-iom},
\begin{equation*}
\gamma_\Omega U_z F=p^{-\frac{\nu(z)}{2}} I_\Omega e_z=
	  \begin{cases}
		\one_\Omega, & z=\rad,\\[\bigskipamount]
		\displaystyle \sum_{j=0}^{p-1} \theta_s^j \one_{\Omega_{n+1,pk+j}},& z=(n,k,s),
	\end{cases}
\quad \nu(z)\ne N-1.
\end{equation*}
Recall that $U_\rad F=\varphi$ and that for $(n,k,s)\in\cZ$ with  $n\le N-1$ one has
\[
U_{n,k,s} F=\sum_{j=0}^{p-1} \theta_s^j \varphi \one_{\T_{n+1,pk+j}}.
\]

Now let us define a linear map $R:S_N\to L^2(\Omega)$ by
\[
R(\varphi \one_{\T_{N,K}}):=\one_{\Omega_{N,K}} \text{ for any } K\in\{0,\dots,p^N-1\}.
\] 
Using the linearity one obtains
\[
R(U_\rad F)=R(\varphi)=R\Big( \sum_{K=0}^{p^N-1} \varphi\one_{\T_{N,K}}\Big)=
	\sum_{K=0}^{p^N-1} R(\varphi \one_{\T_{N,K}})=\sum_{K=0}^{p^N-1} \one_{\Omega_{N,K}}=\one_\Omega,
\]
and for any $(n,k,s)\in\cZ$ with $n\le N-1$
\begin{align*}
R(&\varphi \one_{\T_{n+1,pk+j}})=
R( \sum_{K:\, \T_{N,K}\subset \T_{n+1,pk+j}} \varphi \T_{N,K}\Big)
=\sum_{K:\, \T_{N,K}\subset \T_{n+1,pk+j}} R(\varphi \T_{N,K})\\
&=\sum_{K:\, \T_{N,K}\subset \T_{n+1,pk+j}} \one_{\Omega_{N,K}}
=\sum_{K:\, \Omega_{N,K}\subset \Omega_{n+1,pk+j}} \one_{\Omega_{N,K}}=\one_{\Omega_{n+1,pk+j}},\\
R(&U_{n,k,s} F)=\sum_{j=0}^{p-1} \theta_s^j R(\varphi \one_{\T_{n+1,pk+j}})
=\sum_{j=0}^{p-1} \theta_s^j \one_{\Omega_{n+1,pk+j}},
\end{align*}
which shows that $R(U_z F)=\gamma_\Omega U_z F$ for any $z\in\cZ$ with $\nu(z)\le N-1$.
As $U_zF$ form a basis of $S_N$, it follows that $R=\gamma_\Omega|_{S_N}$. In particular,
$\gamma_\Omega(\varphi \one_{\T_{N,K}})=R(\varphi \one_{\T_{N,K}})$ for all $K\in\{0,\dots,p^N-1\}$, which shows \eqref{eq-temp2} and concludes the proof.	
\end{proof}

\begin{remark}
In Theorem \ref{thm-geom} one can also take the limit in $A^s(\Omega)$ with any $s$ such that $\gamma_\Omega:H^1(\Omega)\to A^s(\Omega)$ is bounded:
the proof remains unchanged.
\end{remark}

\subsection{Proof of Theorem \ref{thm11}}\label{sec43}

By now we have proved all assertions of Theorem \ref{thm11} for the special case $\cT=\T$. Recall that in Theorem \ref{thm11} we require the condition \eqref{cc}, i.e.
\begin{equation}
	\label{cc2}
	c^{-1}\le \dfrac{\ell_{n,k}}{\ell^n}\le c,\quad 
	c^{-1}\le \dfrac{w_{n,k}}{\alpha_n}\le c,
\end{equation}
and that $\T$ corresponds to $c=1$.
In order to cover the case of general $\cT$ we employ a suitable bijection
between $\cT$ and $\T$. Namely, define $\varphi:\T\to \cT$ by
\[
\varphi(n,k,t):=\Big(n,k,L_{n,k}-\ell_{n,k}+\dfrac{t-t_{n-1}}{\ell^n}\, \ell_{n,k}\Big),
\]
then $\varphi$ maps the vertices $X_{n,k}$ on $\T$ to the same vertices on $\cT$,
the restrictions $\varphi|_{e_{n,k}}$ are dilations by constant factors, and both $\varphi$ and $\varphi^{-1}$ are continious.

If $f:\cT\to \C$, consider $g:=f\circ\varphi:\T\to \C$. Remark that $f$ is continuous if and only if $g$ is continuous. Furthermore, if $f=(f_{n,k})$ and $g=(g_{n,k})$, then
\begin{align*}
	\|f\|^2_{L^2(\cT)}&=\sum_{n=0}^\infty\sum_{k=0}^{p^n-1} \int_{L_{n,k}-\ell_{n,k}}^{L_{n,k}}\big|f_{n,k}(s)\big|^2w_{n,k}(s)\dd s,\\	
	\|g\|^2_{L^2(\T)}&=\sum_{n=0}^\infty\sum_{k=0}^{p^n-1} \alpha^n\int_{t_{n-1}}^{t_n} \big|g_{n,k}(t)\big|^2\dd t\\
	&=\sum_{n=0}^\infty\sum_{k=0}^{p^n-1} \alpha^n\int_{t_{n-1}}^{t_n} \Big|g_{n,k}\Big(n,k,L_{n,k}-\ell_{n,k}+\frac{t-t_{n-1}}{\ell^n}\, \ell_{n,k}\Big)\Big|^2\dd t\\
	&=\sum_{n=0}^\infty\sum_{k=0}^{p^n-1} \alpha^n \dfrac{\ell^n}{\ell_{n,k}}
	\int_{L_{n,k}-\ell_{n,k}}^{L_{n,k}}\big|g_{n,k}(s)\big|^2\dd s.
\end{align*}
In view of \eqref{cc2} we have
\[
c^{-2}\dfrac{\ell^n}{\ell_{n,k}}\le c^{-1}\alpha^n \le w_{n,k}\le c \alpha^n\le c^2 \alpha^n \dfrac{\ell^n}{\ell_{n,k}},
\]
and we infer
\begin{equation}
	c^{-2}\|g\|^2_{L^2(\T)}\le \|f\|^2_{L^2(\cT)}\le c^{2}\|g\|^2_{L^2(\T)}.
	\label{fgcc}
\end{equation}
In addition, $f_{n,k}$ is weakly differentiable if and only if $g_{n,k}$ is weakly differentiable, and then
\[
g'_{n,k}=\dfrac{\ell_{n,k}}{\ell^n}(f'\circ\varphi)_{n,k}.
\]
By \eqref{cc2}
It follows that $c^{-1} |f'\circ\varphi|\le |g'|\le c |f'\circ\varphi|$, and then
\[
c^{-4}\|f'\|^2_{L^2(\cT)}
\stackrel{\eqref{fgcc}}{\le}
c^{-2}\|f'\circ \varphi\|^2_{L^2(\T)}
\|g'\|^2_{L^2(\T)}\le c^2\|f'\circ \varphi\|^2_{L^2(\T)}\stackrel{\eqref{fgcc}}{\le}c^4\|f'\|^2_{L^2(\cT)}.
\]

It follows that the linear operator $\Theta:f\mapsto f\circ\varphi$ is an isomorphism between
$L^2(\cT)$ and $L^2(\T)$ as well as between $H^1(\cT)$ and $H^1(\T)$. In addition, it is bijective from $H^1_c(\cT)$ and $H^1_c(\T)$ by construction, so it is also an isomorphism
between $H^1_0(\cT)$ and $H^1_0(\T)$. This shows that $H^1(\cT)\ne H^1_0(\cT)$ if and only if
$H^1(\T)\ne H^1_0(\T)$, which is equivalent to the inequalities \eqref{lap}.

Due to Theorem \ref{thm-geom} we actually have $\gamma_\Omega^\cT=\gamma_\Omega\circ\Theta$, so the properties of $\gamma_\Omega$ from the preceding subsection are directly transferred to $\gamma^\cT_\Omega$. In particular:
\begin{itemize}
\item $\ker\gamma^\cT_\Omega=\Theta^{-1}(\ker\gamma_\Omega)=\Theta^{-1}\big(H^1_0(\T)\big)
=H^1_0(\cT)$,
\item if $\Omega$ is a $d$-dimensional Euclidean open set admitting a regular strongly balanced $p$-multiscale decomposition (Subsection \ref{ssmult}), then $\gamma^\cT_\Omega:H^1(\cT)\to A^s(\Omega)$ is a bounded linear operator for any $0\le s\le \sigma d$, and it is surjective for $s=\sigma d$,
\item if $\Omega$ is an open set in $d$-dimensional manifold and admitting a regular strongly balanced $p$-multiscale decomposition (Subsection \ref{sec44}), then:
\begin{itemize}
\item $\gamma_\Omega:H^1(\cT)\to A^s(\Omega)$ is a bounded linear operator for any $0\le s< 1$ such that $s\le \sigma d$, surjective if $s=\sigma d<1$,
\item $\gamma_\Omega:H^1(\cT)\to H^s(\Omega)$ is a bounded linear operator for any $0\le s< \frac{1}{2}$ such that $s\le \sigma d$, surjective if $s=\sigma d<\frac{1}{2}$.
\end{itemize}
\end{itemize}
All assertions are proved.

\section{Existence of regular balanced decompositions}\label{sec52}

The construction of the embedded trace in the preceding subsection requires the existence of a regular strongly balanced $p$-multiscale decomposition. Let us show that such decomposition really exist for a wide class of $\Omega$.

\begin{example}[Hypercubes]\label{hcubes} Let us show first how to construct
a regular strongly balanced $p$-multiscale decomposition
of the $d$-dimensional hypercube
\[
Q^{(d)}:=(0,1)^d\subset\R^d.
\]

For $d=1$ we decompose iteratively each interval into $p$ equal subintervals 
to obtain the decomposition
\begin{equation}
	\label{1ddecomp}
	Q^{(1)}_{n,k}=\big(kp^{-n}, (k+1)p^{-n}\big), \quad n\in\N_0,\quad k\in\{0,\dots,p^n-1\}.
\end{equation}

For $d\ge 2$ we obtain a decomposition by dividing alternately each side into $p$ equal parts.
First, set $\Omega_{0,0}:=\Omega$. Now assume that $Q^{(d)}_{n,k}$ are already constructed for some $n\in\N_0$ and all $k\in\{0,\dots,p^n-1\}$ and that for each $(n,k)$ one has
\[	
Q^{(d)}_{n,k}=Q^{(1)}_{n_1,k_1}\times \dots\times Q^{(1)}_{n_d,k_d}
\]
with suitable $n_s\in\N_0$ and $k_s\in \{0,\dots,p^{n_s}-1\}$.
Let $i\in\{1,\dots,d\}$ be such that $(n+1)\equiv i \mod d$, then we obtain $Q^{(d)}_{n+1,pk+j}$ with $j\in \{0,\dots,p-1\}$ by subdividing the $i$-th side $Q^{(1)}_{n_i,k_i}$ of $Q^{(d)}_{n,k}$
into  $p$ equal subintervals
\[
\Tilde I_j:=Q^{(1)}_{n_i+1,p k_i+j},\quad j\in\{0,\dots,p-1\},
\]
and then by setting, for each $j\in\{0,\dots,p-1\}$,
\[
Q^{(d)}_{n+1,pk+j}=Q^{(1)}_{n_1,k_1}\times Q^{(1)}_{n_{i-1},k_{i-1}} \times \Tilde I_j \times Q^{(1)}_{n_{i+1},k_{i+1}}\times\dots \times Q^{(1)}_{n_d,k_d}.
\]

Let us show that this decomposition is regular and strongly balanced. The assumptions (A1)--(A3)
are obviously satisfied, as well as (A4*), as on each passage from $Q^{(d)}_{n,k}$ to $Q^{(d)}_{n+1,k}$ one divides the volumes exactly by $p$. In order to check (A5)--(A6) we remark that for $n\ge d+1$ each $Q^{(d)}_{n,k}$ has the form $I_1\times\dots\times I_d$, where $I_j$ are intervals with
\[
p^{-K}\le|I_j|\le p^{1-K} \text{ for } (K-1)d+1\le n\le Kd, \quad K\in\N_0.
\]
We infer 
\begin{equation}
	\label{ipd}
	p^{-\frac{n+d}{d}}\le |I_j|\le p^{-\frac{n-d}{d}}
\end{equation}
and it follows that
\[
\diam Q^{(d)}_{n,k}\equiv\sqrt{|I_1|^2+\dots+|I_d|^2}\le \sqrt{d}\,p^{-\frac{n-d}{d}}\equiv p\sqrt{d} \,p^{-\frac{n}{d}},
\]
i.e. (A5) is satisfied. Now let $h=(h_1,\dots,h_d)\in\R^d$, then
\begin{align*}
Q^{(d)}_{n,k}&\setminus(Q^{(d)}_{n,k}+h)\\
&=(I_1\times I_2\times \dots\times I_d) \setminus \big( (I_1+h_1)\times (I_2+h_2)\times \dots \times (I_d+h_d)
\big)\\
&\subset \big(I_1\setminus(I_1+h_1)\big)\times I_2\times \dots\times I_d\\
&\qquad \cup \ I_1\times\big(I_2\setminus(I_2+h_2)\big)\times  \dots\times I_d\\
&\qquad\ldots\cup\  I_1\times I_2\times \dots\times \big(I_d\setminus (I_d+h_d)\big).
\end{align*}
We have $\big|I_k\setminus(I_k+h_k)\big|\le|h_k|\le|h|$, which gives the volume estimate
\[
\big|Q^{(d)}_{n,k}\setminus(Q^{(d)}_{n,k}+h)\big|\le  \sum_{k=1}^d |h_k| \prod_{j\ne k} |I_j|\stackrel{\eqref{ipd}}{\le}
\sum_{k=1}^d |h_k| \big(p^{-\frac{n-d}{d}}\big)^{d-1}
\le dp^{d-1}|h| p^{-n\frac{d-1}{d}}
\]
and shows (A6).
\end{example}

\begin{example}[Piecewise smooth star-shaped open sets]\label{hstar}
One says that a bounded open set $\Omega\subset\R^d$ belongs to the class (H) if:
\begin{itemize}
	\item $\Omega$ is star-shaped with respect to a point $x_0\in\Omega$.
	\item there exist $\eps>0$ with $B_\eps(x_0)\subset \Omega$ and a finite partition of $\Omega=\Omega_1\cup\dots\cup\Omega_n$ such that for each $j=1,\dots,n$:
\begin{itemize}
	\item each $\Omega_j$ is a cone with vertex at $x_0$,
	\item $\partial \Omega_j\cap \partial\Omega$ is a $C^1$ surface,
	\item the set $B_\eps(x_0)\cap \Omega_j$ is convex.
\end{itemize}	
\item there exists $\delta>0$ such that $\nu(x)\cdot (x-x_0)\ge \delta$ for all $x\in\partial\Omega$, where $\nu(x)$ denotes the outward unit normal to $\partial\Omega$ at $x$ (defined almost everywhere on $\partial\Omega$).
\end{itemize}
Remark that the class (H) contains all convex polyhedrons and all convex open sets with smooth boundaries. It is shown in \cite[Thm.~5.4]{fonseca} that for arbitrary $\Omega,\Omega'$ in (H) with $|\Omega|=|\Omega'|$ there exists a bi-Lipschitz bijection $\Phi:\Omega'\to \Omega$ with $|\det D\Phi|=1$ (i.e. $\Phi$ preserves the volumes). Note that for many special classes of $\Omega,\Omega'$ like cubes, balls, cylinders, simplices such a map $\Phi$ can be given by explicit formulas, see \cite{bilip,hr} and references therein.

If $\Omega'$ admits a regular strongly balanced $p$-multiscale decomposition $(\Omega'_{n,k})$, then the sets $\Omega_{n,k}:=\Phi(\Omega'_{n,k})$ form a regular strongly balanced $p$-multiscale decomposition of $\Omega$, as the conditions (A5)--(A6) remain true under bi-Lipschitz transformations. In particular, for each $\Omega$ in (H)
one can take a hypercube $Q$ with $|Q|=|\Omega|$ and translate a decomposition of $Q$
(Example \ref{hcubes}) into a decomposition of $\Omega$. 

This discussion shows that any open set $\Omega$ of the class  (H) admits a regular strongly balanced $p$-multiscale decomposition.
\end{example}

\begin{example}[Composed open sets]
Let $\Omega$ be an open set with compact closure in a $d$-dimensional manifold of bounded geometry $S$. Assume that $\Omega$ can be decomposed (up to zero measure sets) into disjoint open pieces $W_j$, $j=1,\dots,p^N$ such that
\begin{itemize}
\item all $W_j$ have the same volume,
\item there exist local charts $\Phi_j:\,\R^d\ni \Tilde O_j\to O_j\subset S$
with $\overline{W_j}\subset O_j$,
\item the sets $\Tilde W_j:=\Phi_j^{-1}(W_j)\subset \R^d$ 
are with Lipschitz boundaries and admit regular strongly balanced $p$-multiscale decompositions,
\end{itemize}
then the decompositions of $\Tilde W_j$ are first transfered to $W_j$ with the help
of $\Phi_j$ and then suitably renumerated to produce a regular strongly balanced $p$-multiscale decomposition of the whole $\Omega$.
\end{example}

\begin{example}[Compact manifolds]
By combining the preceding observations one can show that each compact manifold
admits a regular strongly balanced $p$-multiscale decomposition. The idea comes from Beno\^{\i}t Kloeckner's comments
in the MathOverflow discussion \cite{tl}.
	
Let $(\Omega,g)$ be a compact $d$-dimensional Riemannian manifold. It is known \cite{wh} that $\Omega$ admits a triangulation: there exist disjoint open $W_1,\dots,W_N\subset \Omega$ 
with 
\[
\big|\Omega\setminus(W_1\cup\dots\cup W_N)\big|=0
\]
and local charts
$\Phi_j:\R^d\ni \Tilde O_j\mapsto O_j\subset\Omega$ with $\overline{W_j}\subset O_j$
such that the sets $\Tilde\Omega_j:=\Phi^{-1}_j(W_j)$ are $d$-dimensional simplices.
Without loss of generality we assume that $N=p^n$ (otherwise one cuts some of the simplices $\Tilde\Omega_j$ into smaller subsimplices to obtain a required number). Then one can find a smooth function $f:\Omega\to(0,\infty)$ such that
\[
\int_{W_j} f \dvol_g=\dfrac{|\Omega|}{N}.
\]
By \cite{jm} there exists a diffeomorpism $\phi:\Omega\to\Omega$ with
$\phi_*(f\dvol_g)=\dvol_g$. The open sets $\Omega_j:=\phi(W_j)$ satisfy
\[
|\Omega_j|=\int_{\Omega_j} 1\dvol_g=\int_{\phi(W_j)} 1\dvol_g=\int_{W_j} f\dvol_g=\dfrac{|\Omega|}{N},
\]
i.e. they have the same volume and exhaust $\Omega$ up to a zero-measure subset.
In addition, each $\Omega_j$ is covered by the local chart $\Psi_j:=\phi\circ \Phi_j$ with
$\Psi_j^{-1} (\Omega_j)=\Tilde\Omega_j$.
As discussed in Example \ref{hstar}, each simplex $\Tilde\Omega_j$ admits a regular strongly balanced $p$-multiscale decomposition. This decomposition is transferred to $\Omega_j$
with the help of $\Psi_j$, and the resulting decompositions of $\Omega_j$ are then combined
into  a regular strongly balanced $p$-multiscale decomposition of $\Omega$.	
\end{example}



\end{document}